\documentclass[12pt,a4paper]{book}
\usepackage{inputenc}
\usepackage{hyperref}
\usepackage[T1]{fontenc}
\usepackage{amsmath,amsthm,amssymb,amsfonts,mathrsfs}
\usepackage[english]{babel}
\usepackage{latexsym}
\usepackage[normalem]{ulem}
\usepackage{xcolor}
\usepackage{txfonts,pxfonts,dsfont}
\usepackage{graphicx}
\usepackage{epstopdf}
\usepackage{enumitem}
\usepackage{indentfirst}
\usepackage{appendix}
\usepackage[nottoc,numbib]{tocbibind}
\usepackage{comment}
\usepackage{pdfpages}

\usepackage{natbib}
\setcitestyle{numbers,comma,square}

\usepackage{anysize}
\usepackage[ruled,chapter]{algorithm}
\usepackage{makeidx}
\usepackage{subfig}
\usepackage{fancyhdr}
\usepackage{array}
\usepackage{geometry}
\usepackage{simplewick}
\usepackage[all]{xy}
\usepackage{tikz}
\usetikzlibrary{intersections}
\newcommand{\comm}[2]{\left[#1,#2\right]}

\newcommand{\Hor}{\mathscr{H}} 
\newcommand{\cB}{\mathcal{B}} 
\newcommand{\sB}{\mathscr{B}} 
\newcommand{\sC}{\mathscr{C}} 
\newcommand{\sD}{\mathscr{D}} 
\newcommand{\M}{\mathcal{M}} 
\newcommand{\sK}{\mathscr{K}} 
\newcommand{\N}{\mathcal{N}} 
\newcommand{\cT}{\mathcal{T}} 
\newcommand{\cC}{\mathcal{C}} 
\newcommand{\K}{\mathsf{K}} 
\newcommand{\fA}{\mathfrak{A}} 
\newcommand{\sA}{\mathscr{A}} 
\newcommand{\F}{\mathscr{F}} 
\newcommand{\W}{\mathscr{W}} 

\theoremstyle{plain}
\newtheorem{thm}{Theorem}[subsection]
\newtheorem{lem}[thm]{Lemma}
\newtheorem{prop}[thm]{Proposition}

\theoremstyle{definition}
\newtheorem{mydef}[thm]{Definition}

\numberwithin{equation}{section}

\makeindex

\marginsize{25mm}{25mm}{25mm}{25mm}

\hypersetup{pdfpagelabels,hyperindex,colorlinks=true,breaklinks=true,bookmarks=true,linkcolor=black,citecolor=black,
urlcolor=black}

\graphicspath{{figures/}}

\begin{document}

\setlength{\abovecaptionskip}{0.0cm}
\setlength{\belowcaptionskip}{0.0cm}
\setlength{\baselineskip}{24pt}

\pagestyle{fancy}
\lhead{}
\chead{}
\rhead{\thepage}
\lfoot{}
\cfoot{}
\rfoot{}

\fancypagestyle{plain}
{
	\fancyhf{}
	\lhead{}
	\chead{}
	\rhead{\thepage}
	\lfoot{}
	\cfoot{}
	\rfoot{}
}

\renewcommand{\headrulewidth}{0pt}


\frontmatter 

\thispagestyle{empty}

\begin{figure}[h]
	\includegraphics[scale=0.8]{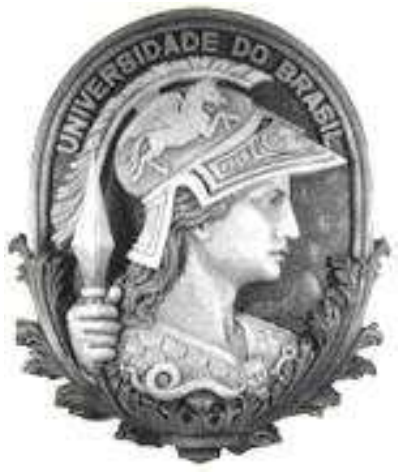}
\end{figure}

\vspace{15pt}

\begin{center}

\textbf{UNIVERSIDADE FEDERAL DO RIO DE JANEIRO}

\textbf{INSTITUTO DE F{\'I}SICA}

\vspace{30pt}

{\Large \bf Explicit construction of Hadamard states for Quantum Field Theory in curved spacetimes}

\vspace{25pt}

{\large \bf Marcos Carvalho Brum de Oliveira}

\vspace{35pt}

\begin{flushright}
\parbox{10.3cm}{Tese de Doutorado apresentada ao Programa de P{\'o}s-Gradua\c c\~ao em F\'\i sica do Instituto de F\'\i sica da Universidade Federal do Rio de Janeiro - UFRJ, como parte dos requisitos necess{\'a}rios \`a obten\c c\~ao do t\'\i tulo de Doutor em Ci\^encias (F\'\i sica).}

\vspace{18pt}

{\large \bf Orientador: Sergio Eduardo de Carvalho Eyer Jor\'{a}s}

\vspace{12pt}

{\large \bf Coorientador: Klaus Fredenhagen}
\end{flushright}

\vspace{90pt}

\textbf{Rio de Janeiro}

\textbf{Maio de 2014}

\end{center}



\newpage
\includepdf[pages={-}]{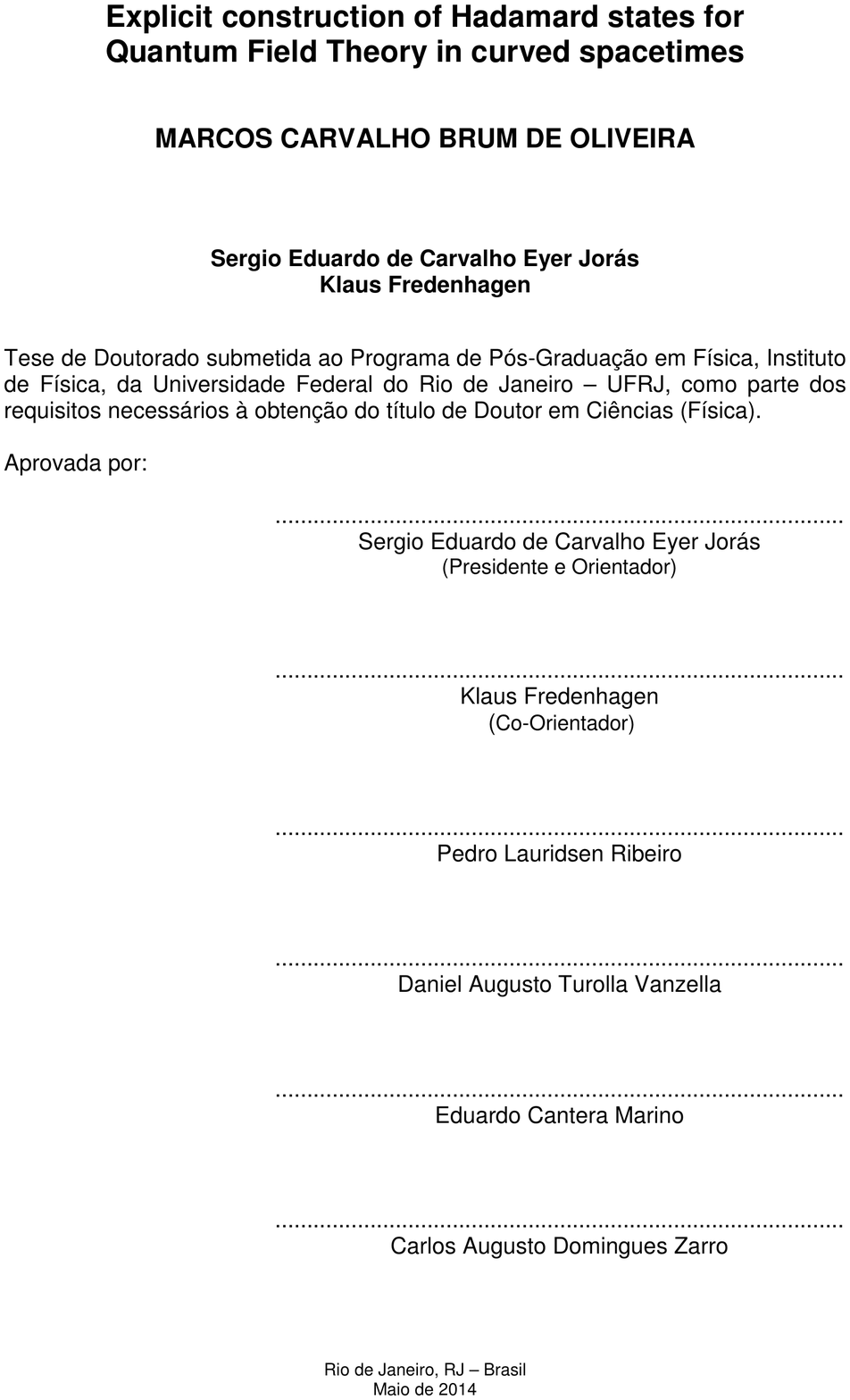}



\newpage
\includepdf[pages={-}]{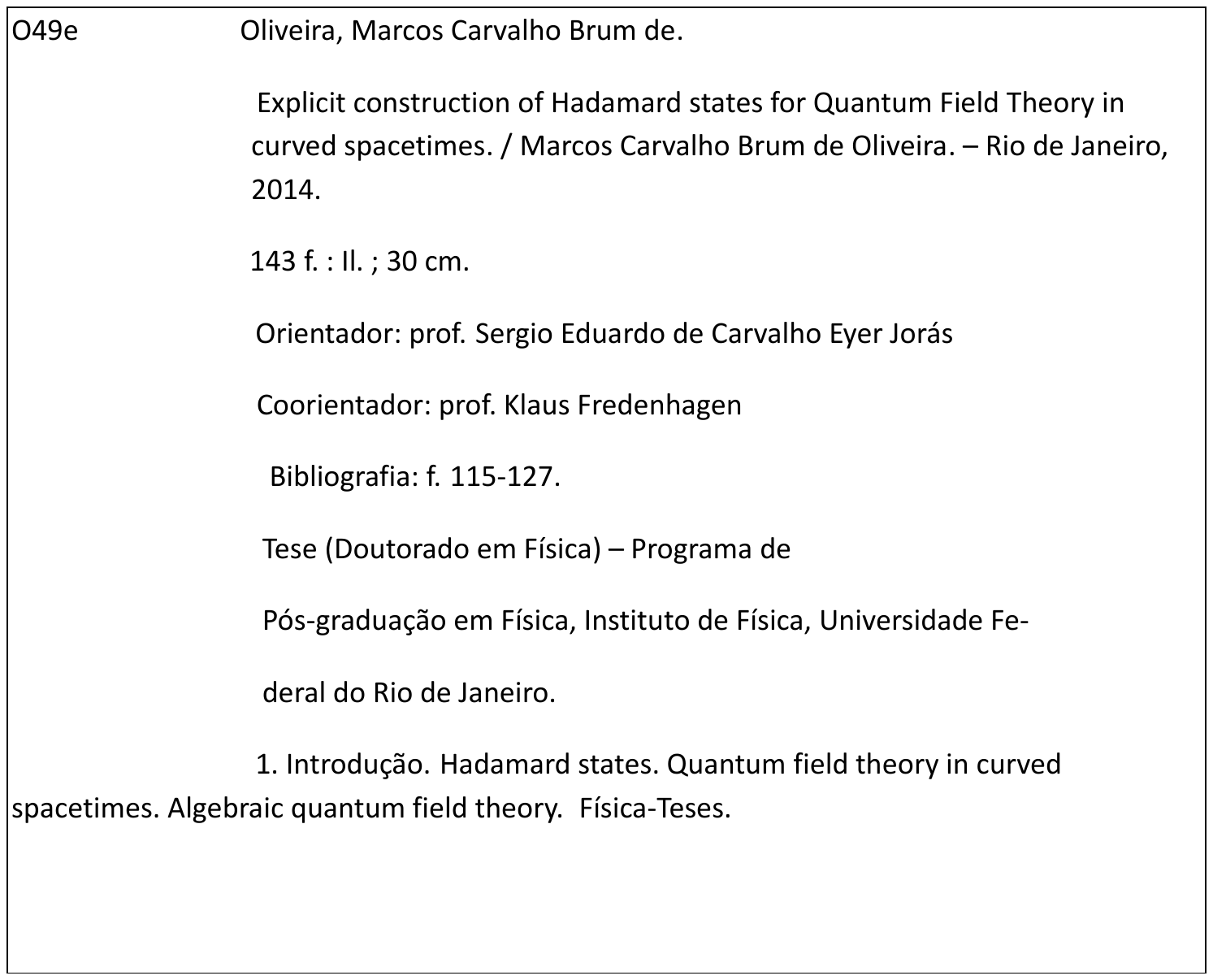}



\newpage

\noindent

\vspace*{20pt}
\begin{center}
{\LARGE\bf Abstract}\\
\vspace{15pt}
{\Large\bf Explicit construction of Hadamard states for Quantum Field Theory in curved spacetimes}\\
\vspace{6pt}
{\bf Marcos Carvalho Brum de Oliveira}\\
\vspace{12pt}
{\bf Advisor: Sergio Eduardo de Carvalho Eyer Jor\'{a}s}\\
{\bf Co-advisor: Klaus Fredenhagen}\\
\vspace{20pt}
\parbox{14cm}{\uline{Abstract} of the Doctorate Thesis presented to the Graduate Program in Physics of the Physics Institute of the Universidade Federal do Rio de Janeiro - UFRJ, as part of the necessary pre-requisites for the obtainment of the title Doctor in Sciences (Physics).}
\end{center}
\vspace*{35pt}

A crucial feature of the quantum field theory in curved spacetimes is the absence of a state which corresponds to the vacuum state in Minkowski spacetime. Generically, this absence is due to the absence of a timelike Killing vector. It was recently proved that it is not possible to exist a unique prescription for the construction of a vacuum state which remains valid in any globally hyperbolic spacetime \citep{FewsterVerch12,FewsterVerch_sf12}. On the other hand, the algebraic approach to the quantum field theory formulates the theory, taking, as a starting point, the algebraic relations satisfied by the local field operators associated to each spacetime region. In this approach, states play, fundamentally, the role of the expectation values of observables.

In this way, the question of the determination of a state becomes intimately related to the renormalizability. Since the end of the 1970's, culminating in the first decade of this century, it was shown that the knowledge about the location of the singularities of the states permits us to reduce the determination of the perturbative renormalizability of an interacting theory to the usual power-counting method, as in the Minkowski spacetime \citep{BruFre00,HoWa02,HoWa03,HoWa05}. More specifically, the authors of \citep{Wald77,BruFreKoe96,HoWa01,HoWa02,Moretti03,HoWa05} showed that the normal ordering of important observables, such as the energy-momentum tensor, with respect to the {\it Hadamard states}, have finite fluctuations in these states. Such states have all their singularities located on the future light cone. This condition is the local remnant of the spectral condition in Minkowski spacetime, together with the fact that the two-point functions of these states are bissolutions of the field equations \citep{KayWald91,Radzikowski96,RadzikowskiVerch96}. With this characterization at hand, the question of the determination of states in curved spacetimes shifted its focus from the determination of states which would be analogous to the vacuum state to focus on the construction of Hadamard states.

In this Thesis we will explicitly construct Hadamard states with different features. In the two first cases \citep{BrumFredenhagen14,ThemBrum13}, the states will be constructed in static and in cosmological spacetimes. In one of the cases \citep{ThemBrum13}, the states will be constructed in such a way as to minimize the expectation value of the energy-momentum tensor of a free massive scalar field, multiplied by a suitable smooth function of compact support, in a cosmological spacetime. The final form of the state clearly reduces to the vacuum in a static spacetime. Besides, we will show that these are Hadamard states. They do not contradict the above mentioned nonexistence of a unique prescription because they depend on the chosen smooth function. In the other case \citep{BrumFredenhagen14}, we will calculate the positive part of the spectral decomposition of the commutator function in a specific subset of the spacetime and will show that, in static and in cosmological spacetimes, the kernel of this function, once again multiplied by a suitable smooth function of compact support, corresponds to the kernel of Hadamard states. In both cases, the construction of the states does not rely on invariance under the action of a group of symmetries. Such a group may not even exist, in these cases.

In the last case we will construct a state in the Schwarzschild-de Sitter spacetime, which describes a universe with a black hole and a cosmological constant. This state is invariant under the action of the group of symmetries of this spacetime, although it is not defined in its whole Kruskal extension. In this way, we avoid the situation in which the authors of \citep{KayWald91} showed that a Hadamard state cannot exist. We will show that the state we construct is Hadamard and that it is not thermal, differently from the case of states constructed in spacetimes containing only one event horizon \citep{DappiaggiMorettiPinamonti09,GibbonsHawking77}. This result is not yet published.

\vspace{15pt}

 Quantum field theory in curved spacetimes. Algebraic quantum field theory.



\newpage

\noindent

\vspace*{20pt}
\begin{center}
{\LARGE\bf Resumo}\\
\vspace{15pt}
{\Large\bf Explicit construction of Hadamard states for Quantum Field Theory in curved spacetimes}\\
\vspace{6pt}
{\bf Marcos Carvalho Brum de Oliveira}\\
\vspace{12pt}
{\bf Orientador: Sergio Eduardo de Carvalho Eyer Jor\'{a}s}\\
{\bf Coorientador: Klaus Fredenhagen}\\
\vspace{20pt}
\parbox{14cm}{Resumo da Tese de Doutorado apresentada ao Programa de P{\'o}s-Gradua\c c\~ao em F\'\i sica do Instituto de F\'\i sica da Universidade Federal do Rio de Janeiro - UFRJ, como parte dos requisitos necess{\'a}rios \`a obten\c c\~ao do t\'\i tulo de Doutor em Ci\^encias (F\'\i sica).}
\end{center}
\vspace*{35pt}

Um ponto crucial da Teoria Qu\^antica de Campos em espa\c cos-tempos Curvos \'{e} a aus\^encia de um estado que corresponda ao estado de v\'{a}cuo no espa\c co-tempo de Minkowski. Em geral, esta aus\^encia se d\'{a} pela falta de um vetor de Killing temporal. Recentemente, foi provado que n\~ao \'{e} poss\'\i vel que haja uma prescri\c c\~ao para a constru\c c\~ao de um estado de v\'{a}cuo que seja v\'{a}lida para qualquer espa\c co-tempo globalmente hiperb\'{o}lico \citep{FewsterVerch12,FewsterVerch_sf12}. Por outro lado, a abordagem alg\'{e}brica para a Teoria Qu\^antica de Campos constr\'{o}i toda a teoria de uma maneira formal, a partir das rela\c c\~oes alg\'{e}bricas satisfeitas pelos operadores de campo locais associados a cada regi\~ao do espa\c co-tempo. Nesta abordagem, os estados assumem, fundamentalmente, o papel de valores esperados de observ\'{a}veis.

Desta forma, o problema da determina\c c\~ao de um estado fica intimamente relacionado \`a renormalizabilidade da teoria. Desde o final da d\'{e}cada de 1970, culminando na primeira d\'{e}cada deste s\'{e}culo, mostrou-se que o conhecimento sobre a localiza\c c\~ao das singularidades dos estados permite reduzir a determina\c c\~ao da renormalizabilidade perturbativa de uma teoria interagente ao m\'{e}todo usual de {\it power-counting}, tal como no espa\c co-tempo de Minkowski \citep{BruFre00,HoWa02,HoWa03,HoWa05}. Mais especificamente, os autores de \citep{Wald77,BruFreKoe96,HoWa01,HoWa02,Moretti03,HoWa05} mostraram que o ordenamento normal de observ\'{a}veis importantes, como o tensor energia-momento, com respeito aos {\it Estados de Hadamard}, possuem flutua\c c\~oes finitas nestes estados. Tais estados tem todas as suas singularidades localizadas sobre o cone de luz futuro. Esta condi\c c\~ao \'{e} o remanescente local da condi\c c\~ao espectral no espa\c co-tempo de Minkowski, combinada com o fato de que as fun\c c\~oes de dois-pontos destes estados \~ao bissolu\c c\~oes das equa\c c\~oes de campo \citep{KayWald91,Radzikowski96,RadzikowskiVerch96}. Com esta caracteriza\c c\~ao, a quest\~ao da determina\c c\~ao de estados em espa\c cos-tempos curvos deixou de ser focada na determina\c c\~ao de estados an\'{a}logos ao estado de v\'{a}cuo para focar na constru\c c\~ao de estados de Hadamard.

Nesta Tese, n\'{o}s constru\'\i mos, explicitamente, alguns estados de Hadamard com caracter\'\i sticas diferentes. Nos dois primeiros casos \citep{BrumFredenhagen14,ThemBrum13}, os estados s\~ao constru\'\i dos em espa\c cos-tempos est\'{a}ticos e em espa\c cos-tempos cosmol\'{o}gicos. Em um dos casos \citep{ThemBrum13}, os estados s\~ao constru\'\i dos de forma a minimizarem o valor esperado do tensor energia-momento de um campo escalar massivo (livre), multiplicado por uma fun\c c\~ao suave de suporte compacto adequada, em um espa\c co-tempo cosmol\'{o}gico. A forma final do estado deixa claro que esta prescri\c c\~ao se reduziria ao v\'{a}cuo num espa\c co-tempo est\'{a}tico. Al\'{e}m disso, n\'{o}s mostramos que estes s\~ao estados de Hadamard. Eles n\~ao contradizem a inexist\^encia de uma prescri\c c\~ao \'{u}nica, mencionada acima, porque eles dependem da fun\c c\~ao suave escolhida. No outro caso \citep{BrumFredenhagen14}, n\'{o}s calculamos a parte positiva da decomposi\c c\~ao espectral da fun\c c\~ao comutador num determinado subespa\c co do espa\c co-tempo e mostramos que, em espa\c cos-tempos est\'{a}ticos e em espa\c cos-tempos cosmol\'{o}gicos, o n\'{u}cleo distribucional desta fun\c c\~ao, mais uma vez multiplicado por uma fun\c c\~ao suave de suporte compacto adequada, corresponde ao n\'{u}cleo distribucional de estados de Hadamard. Nestes dois casos, a constru\c c\~ao dos estados n\~ao est\'{a} apoiada na invari\^ancia sob a a\c cao de um grupo de simetrias. Tal grupo pode nem existir, nestes casos.

No \'{u}ltimo caso n\'{o}s constru\'\i mos um estado no espa\c co-tempo de Schwarzschild-de Sitter, que descreve um universo com um buraco negro e constante cosmol\'{o}gica. Este estado \'{e} invariante sob a a\c c\~ao do grupo de simetrias deste espa\c co-tempo, apesar de n\~ao estar definido em toda a sua extens\~ao de Kruskal. Desta forma, n\'{o}s evitamos recair nas situa\c c\~oes nas quais os autores de \citep{KayWald91} mostraram que n\~ao \'{e} poss\'\i vel existir um estado de Hadamard. N\'{o}s mostramos que o estado constru\'\i do \'{e} um estado de Hadamard e que ele n\~ao \'{e} um estado t\'{e}rmico, diferente do caso de estados constru\'\i dos num espa\c co-tempo contendo apenas um horizonte de eventos \citep{DappiaggiMorettiPinamonti09,GibbonsHawking77}. Este resultado ainda n\~ao foi publicado.

\vspace{15pt}

\textbf{Palavras-chave:} Estados de Hadamard. Teoria qu\^antica de campos em espa\c cos tempos curvos. Teoria qu\^antica de campos alg\'{e}brica.



\newpage

\noindent

\vspace*{20pt}

\begin{center}

{\LARGE\bf Acknowledgements}

\end{center}

\vspace*{40pt}

I wish to thank my Advisor, Sergio Jor\'{a}s, from the Federal University of Rio de Janeiro (UFRJ), for suggesting the theme of quantum fields in de Sitter spacetime, which was both the initial and final points of my PhD research. The constant stimulus, the freedom, the discussions, the personal conversations and the friendship were of fundamental importance for the completion of this Thesis. I also wish to thank my Co-Advisor, Klaus Fredenhagen, from the University of Hamburg (UHH), for welcoming me for a one-year stay there. It was an extremely profitable experience, both professionally and personally. I also wish to thank him for patiently introducing me to the algebraic approach to quantum field theory and to the change of perspective brought by it. The encouragement and the suggestion of the theme of Hadamard states in globally hyperbolic spacetimes, along with the many explanations and discussions, were also of fundamental importance for the completion of the Thesis. To both of you, I owe a large debt of gratitude. I also wish to thank professors Daniel Vanzella, Pedro Ribeiro, Eduardo Marino and Carlos Zarro for the very interesting discussions and useful comments at my Thesis Defense. You have surely contributed a lot to it. To professors Carlos Farina, Marcus Ven\'\i cius, Ribamar Reis, Maur\'\i cio Calv\~ao, Miguel Quartin and Ioav Waga, from UFRJ, for discussions, conversations and support, thank you. To my dear friend Marcelo Vargas, for the friendship and everything that comes along with it, thank you. To my dear friends Thiago Hartz, Vin\'\i cios Miranda, Rodrigo Neumann, Josephine Rua and Daniel Micha, for showing that a temporary spacelike separation does not weaken the friendship, thank you. To Bruno Lago, Emille Ishida, Daniel Kroff, Tiago Castro, M\'{a}rcio O'Dwyer, Jo\~ao Paulo Vianna, Jos\'{e} Hugo Elsas, Reinaldo de Mello, Daniel Niemeyer, Daniela Szilard, Anderson Kendi, Michael Moraes, Maur\'\i cio Hippert, Elvis Soares, Ana B\'{a}rbara, M\'{a}rcio Taddei, Wilton Kort-Kamp, thank you. To the secretary of the Graduate School of the Physics Institute of UFRJ, thanks for minimizing the effects of bureaucracy. To Daniel Paulino and Christian Pfeifer, thanks for the overseas friendship. To Falk Lindner, Thomas-Paul Hack, Andreas Degner, Stefano Porto, Martin Sprenger, Katarzyna Rejzner, thanks for all the help with getting adapted to Hamburg. To the above mentioned europeans (and the pseudo-european) and to Reza Safari, Sion Chan-Lang, Mojtaba Taslimi, Paniz Imani, Maikel de Vries, Marco Benini, Daniel Siemssen, Claudio Dappiaggi, Nicola Pinamonti and Romeo Brunetti, for the many discussions which helped me a lot through algebraic quantum field theory and everything related to it, the conversations and the company, thank you. To Elizabeth Monteiro Duarte, secretary of the II. Institute for Theoretical Physics at UHH, for conducting me through the intrincate german bureaucracy and for speaking portuguese, thank you. Agrade\c co a meus pais Mauro Brum e Deise L\'{u}cia, pelo amor incondicional. Ao meu irm\~ao Mateus, por ser meu irm\~ao, com todo o significado da palavra. Aos meus av\'{o}s Djalma e Odahina, porque ``o futuro brilha''. Aos meus tios e tias Paulo, Sandra, Maria L\'{u}cia e Cezar, e \`as minhas primas Fernanda, Carolina, Isabela e B\'{a}rbara, pelo carinho e pelo apoio imensur\'{a}veis. Aos meus irm\~aos do Centro Esp\'\i rita Pedro de Alc\^antara e do Movimento Esp\'\i rita, por me dedicarem tanto carinho e me ajudarem a me tornar uma pessoa melhor. Aos meus amigos brasileiros-alem\~aes Mariley, Anderson, Andre, Karina, Vania, Patr\'\i cia, Sonia, Anna, Fani, Iriana, Alcione, Regina e a todos do Grupo Esp\'\i rita Schwester Scheilla, por me terem aberto as portas da Casa e do cora\c c\~ao. A Deus, pela vida.



\newpage

\begin{flushright}
\begin{minipage}{10cm}
``No explorer knows where his work will lead him, but he persists with the confidence that it is useful to know what is out there.''

Jason Rosenhouse and Laura Taalman, {\it Tak1ng Sudoku Ser1ously}

\vspace{0.2cm}

``Do not worry about your difficulties in mathematics; I can assure you that mine are still greater.''

Albert Einstein -- In Calaprice, {\it The Ultimate Quotable Einstein}

\vspace{0.2cm}

``Somos o que pensamos, com a condi\c c\~ao de pensarmos com for\c ca, vontade e persist\^encia.''

L\'{e}on Denis, {\it O Problema do Ser e do Destino}
\end{minipage}
\end{flushright}



\newpage
\phantomsection
\tableofcontents

\newpage
\phantomsection
\listoffigures

\mainmatter

\chapter{Introduction}\label{intro}

A crucial feature of quantum field theory in curved spacetimes is the absence of a state analogous to the vacuum state in Minkowski spacetime. In general, this absence is due to the nonexistence of a timelike Killing vector field. A direct consequence of this absence is the lack of a preferred Hilbert space, along with the usual interpretation of fields as operators on this space. From a more physical point of view, this is intimately related to the absence of a unique notion of particles, a feature which is dramatically demonstrated in the Parker and Hawking effects \citep{Parker69,Hawking74}. The absence of a vacuum state has nowadays the status of a no-go theorem \citep{FewsterVerch12,FewsterVerch_sf12} which is valid under very general conditions. In the algebraic approach, this problem is alleviated because the quantum field theory is formulated from the assignment of an algebra of fields, constructed from the space of solutions of the field equations, to each spacetime region. From this algebra one can construct a Hilbert space on which the fields are represented as operators and this Hilbert space contains a cyclic vector, from which one can construct a Fock space of states \citep{BraRob-I}. Thus algebraic states can be constructed, but the absence of a preferred notion of state persists.

The absence of the vacuum state, whose physical properties are well understood, claims for a description of the desirable physical properties of a state. One such property is that the expectation values of observables are renormalizable. It is known that the energy-momentum tensor, normal ordered with respect to a reference {\it Hadamard state}, is well defined \citep{Wald77,Moretti03}. A Hadamard state is a state whose two-point function has all its singularities lying along the future null cone. This characterization is reminiscent of the spectral condition in Minkowski spacetime. It was first rigorously given as an imposition on the form of the antisymmetric part of the two-point function of the state \citep{KayWald91}. More recently, it was elegantly described as a restriction on the wavefront set of the two-point function of the state \citep{Radzikowski96,RadzikowskiVerch96}. It was used to show that the Hadamard condition suffices for the quantum field theory to be renormalizable and to allow that interacting fields be perturbatively incorporated into this framework \citep{BruFreKoe96,BruFre00,HoWa01,HoWa02}. On the other hand, it is still an open question whether the Hadamard condition is necessary, in general, for renormalizability \citep{FewsterVerch13}. From a more geometrical point of view, the authors of \citep{FullingSweenyWald78,FullingNarcowichWald81} used deformation arguments to prove that Hadamard states exist in a general globally hyperbolic spacetime. Although general, their construction is quite indirect. Explicitly constructing Hadamard states in specific (classes of) spacetimes is still a hard task, which is tackled in this Thesis.

The first two examples we will construct here, in chapters \ref{chap_sle} and \ref{chap_vac-like}, are of classes of states constructed in expanding spacetimes. These can be useful as initial states for the analysis of the backreaction problem. The first of these classes, the {\it States of Low Energy} (SLE) \citep{ThemBrum13}, is defined from the requirement that the expectation value of the energy-momentum tensor, smeared with a suitable test function, is minimized. We will achieve this minimization by exploiting the degrees of freedom of the solutions of the equations of motion. The smearing is necessary because, in spite of the fact that the expectation value of the energy-momentum tensor of a scalar field on a Hadamard state is coherently renormalizable, it possesses no lower bound \citep{EpsGlaJaffe65}. On the other hand, if instead of calculating the renormalized energy density at a particular point of spacetime, one smears it with the square of a smooth test function of compact support along the worldline of a causal observer, one finds that the resulting quantity cannot be arbitrarily negative. It was more recently shown that this result can also be obtained by smearing over a spacelike submanifold of spacetime. These results are known as Quantum Energy Inequalities (QEIs) \citep{Fewster00,FewsterSmith08}. At last, we will prove that the SLE are Hadamard states by comparing them to the adiabatic states \citep{LuRo90}, which satisfy the Hadamard property \citep{JunSchrohe02}, and exploiting properties of the wavefront set of the difference of two distributions.

The second of these classes, the {\it ``Vacuum-like'' states} \citep{BrumFredenhagen14}, originates from a modification of an attempt to construct vacuum states in a general globally hyperbolic spacetime \citep{AfshordiAslanbeigiSorkin12}. Their idea is based on the fact that the commutator function may be considered as the integral kernel of an antisymmetric operator on some real Hilbert space, as discussed long ago by, e.g., Manuceau and Verbeure \citep{ManuceauVerbeure68}. Under some technical conditions, the polar decomposition of this operator yields an operator having the properties of the imaginary unit, and a positive operator in terms of which a new real scalar product can be defined. The new scalar product then induces a pure quasifree state. This method of constructing a state can be applied, e.g., for a free scalar quantum field on a static spacetime where the energy functional provides a quadratic form on the space of Cauchy data in terms of which a Hilbert space can be defined. The result is the ground state with respect to time translation symmetry (see, e.g. \citep{Kay78}). They thus constructed a state in a relatively compact globally hyperbolic spacetime, isometrically embedded in a larger globally hyperbolic spacetime, whose kernel is the positive spectral part of the commutator function in the smaller spacetime. Fewster and Verch \citep{FewsterVerch-SJ12} proved that the commutator function, in this submanifold, is a bounded operator, thus rendering the state well defined, but this state, in general, is not a Hadamard state. Our modification consists of, instead of taking the positive spectral part of the commutator function as integral kernel of the two-point function, we multiply this operator, in both variables, by a positive smooth, compactly supported function, which is equal to 1 in the submanifold. We show that this smearing, in some cases in which the construction can be explicitly carried out, is sufficient for the state to be Hadamard.

In chapter \ref{chap_Sch-dS} we will construct a Hadamard state in the Schwarzschild-de Sitter spacetime \citep{BrumJoras14}. Since this spacetime possesses horizons, a natural question to ask is whether such a state will be thermal, as in the Schwarzschild spacetime \citep{Hawking74,KayWald91} and in the de Sitter spacetime \citep{GibbonsHawking77}. The authors of \citep{GibbonsHawking77,KayWald91} affirmed that such a state would not be thermal because each of the horizons in Schwarzschild-de Sitter spacetime would work as a ``thermal reservoir'', each at a different temperature, a situation in which there is no possibility of attaining thermal equilibrium. The authors of \citep{KayWald91} went even further and proved the nonexistence of a state in the maximally extended Kruskal diagram of the Schwarzschild-de Sitter spacetime by proving that such would give rise to contradictions related to causality\footnote{The point of view adopted in \citep{KayWald91} is more robust because the concepts of thermal state and thermal reservoir are not equivalent. While the former is associated to the concept of ``absolute temperature'', the latter is associated to the concepts of ``calorimetric temperature'' and thermal equilibrium, and these two notions of temperature are not, in general, equivalent \citep{BuchholzSolveen13,Solveen12}.}. We thus construct the state, not in the whole extended Kruskal region, but in the ``physical'' (nonextended) region between the singularity at $r=0$ and the singularity at $r=\infty$. In this sense, our state is not to be interpreted as the Hartle-Hawking state in this spacetime. Neither can it be interpreted as the Unruh state because, in the de Sitter spacetime, the Unruh state can be extended to the whole spacetime while retaining the Hadamard property \citep{NarnhoferPeterThirring96}. Our state cannot have such a feature. We constructed the state solely from the geometrical features of the Schwarzschild-de Sitter spacetime and, consequently, it is invariant under the action of its group of symmetries. We use the bulk-to-boundary technique \citep{DappiaggiMorettiPinamonti06,DappiaggiMorettiPinamonti09b,DappiaggiMorettiPinamonti09c,DappiaggiMorettiPinamonti09} to show that it can be isometrically mapped to the tensor product of two states, one defined on a subset of each event horizon. Each one of these states is a KMS state (this concept will be explained in section \ref{sec-alg_state}) at the temperature of the corresponding horizon. Since these temperatures are different, the resulting state is not KMS. Moreover, we will use results of \citep{DafermosRodnianski07} and an adaptation of the argument presented in \citep{DappiaggiMorettiPinamonti09} to show that it is a Hadamard state.

Before these constructions, we will present in chapter \ref{chap_MathMethods} the necessary mathematical tools, the concept of Hadamard states and the Quantum Energy Inequalities.

Our conclusions will be presented in the final chapter \ref{chap_concl}.

\chapter{Scalar Field quantization on Globally Hyperbolic Spacetimes}\label{chap_MathMethods}

We will first introduce the mathematical concepts needed for the definition of the states in the following chapters. Afterwards, in section \ref{sec_fieldquant}, we will use these concepts to present the general field quantization, define Hadamard states, introduce the Quantum Energy Inequalities and show the quantization in some classes of spacetimes which will be used later.

\section{Preliminaries}\label{sec_preliminaries}

\subsection{The geometry of a globally hyperbolic spacetime}

Let $\M$ be a smooth manifold on which is defined a Lorentzian metric $g$ (such a manifold will be called {\it lorentzian manifold}), i.e., a metric with signature\footnote{All manifolds considered in this thesis are four dimensional.} $(+---)$. A vector $t\in\cT\M$ is said to be: {\it timelike} if $g(t,t)<0$; {\it null} if $g(t,t)=0$; {\it spacelike} if $g(t,t)>0$. Moreover, a vector is said to be {\it causal} if it is either timelike or null. If the tangent vectors to a continuous and piecewise smooth curve $\gamma$ have all the same character at each point of $\gamma$, then this curve is said to have this same character. A {\it causal curve} is a curve whose tangent vectors are everywhere {\it causal}. Similarly, one can define {\it timelike}, {\it null} or {\it spacelike} curves. These definitions can be extended to subsets of the manifold $\M$, when all the curves contained in the subset have the same character. In a local coordinate chart, a given coordinate is said to be (timelike, spacelike, null) if this coordinate is a parameter along a (timelike, spacelike, null) curve.

Let $(E,\M ,\pi)$ be a vector bundle, where $E$ and $\M$ are smooth manifolds and $\pi:E\rightarrow\M$ is a surjective map. Let $s:\M\rightarrow E$ be a smooth map. If
\[\pi\circ s=\textrm{id}:\M\rightarrow\M \; ,\]
then $s$ is called a {\it smooth section} of the vector bundle $(E,\M ,\pi)$ \citep{ChernChenLam99}. The space of smooth sections in $E$ is denoted by $\Gamma^{\infty}(\M,E)$. If $E=\M\times\mathbb{K}$ ($\mathbb{K}$ is either $\mathbb{R}$ or $\mathbb{C}$), the so-called {\it line bundle}, the space of smooth sections in $E$ will be denoted by $\cC^{\infty}(\M;\mathbb{K})$\footnote{For the reader not acquainted with the concept of sections of a manifold, we give a few examples. Let $E=\M\times\mathbb{R}^{n}(\mathbb{C}^{n})$. The sections in $E$ are the real(complex)-valued functions on $\M$. Besides, the sections in $E=\cT\M$ are the {\it vector fields} on $\M$; the sections in $E=\cT^{\ast}\M$ are the {\it covector fields} on $\M$, and the sections in $E=\Lambda^{k}\cT^{\ast}\M$ are the {\it k-forms} on $\M$ ($\Lambda^{k}(V)$ is the space of all antisymmetric contravariant tensors of order $k$, built from the $k$-th antisymmetrized tensor power of the vector space $V$).}.

If a continuous, timelike vector field $t\in\Gamma^{\infty}(\M,\cT\M)$ can be defined, the manifold is said to be {\it time-orientable}. A spacetime is a smooth time-oriented lorentzian manifold, and will be denoted by $(\M,g)$, where $\M$ is the smooth manifold and $g$ is the lorentzian metric. A causal vector $v$ is said to be {\it future-directed} (respectively {\it past-directed}) if $g(t,v)<0$ (respectively $g(t,v)>0$). This definition can be extended to curves as well \citep{ONeill83}. The {\it chronological future} of $p\in\M$, denoted by $I^{+}(p)$, is the set of points $q\in\M$ for which there exists a future-directed timelike curve with base point at $p$ that passes through $q$. Similarly, one can define the {\it chronological past} of a point $p$, denoted by $I^{-}(p)$, and the {\it causal future/past} of $p$, denoted by $J^{\pm}(p)$, only by substituting {\it future} by {\it causal} in the former definition. A point $p\in\M$ is said to be a {\it future endpoint} of a curve $\gamma$ if $\gamma$ is parametrized by a parameter $\lambda$ and, for every neighborhood $O$ of $p$ there exists $\lambda_{0}$ such that $\forall \lambda'>\lambda_{0}$, $\gamma(\lambda')\subset O$. The curve $\gamma$ is said to be {\it future inextendible} if it has no future endpoints. The definition of {\it past endpoint} and {\it past inextendible} curve are similar.

A subset $S\subset\M$ is said to be {\it achronal} if there do not exist $p,q\in S$ such that $q\in I^{+}(p)$, i.e., $I^{+}(S)\cap S=\emptyset$. If $S$ is a closed, achronal set, the {\it future domain of dependence} of $S$, denoted by $D^{+}(S)$, is the set of points $p\in\M$ such that every past inextendible causal curve through $p$ intersects S. Similarly, one defined the {\it past domain of dependence} of $S$, $D^{-}(S)$, and $D(S)\coloneqq D^{+}(S)\cup D^{-}(S)$. A closed achronal set $\Sigma$ such that $D(\Sigma)=\M$ is called a {\it Cauchy hypersurface}.

A {\it convex normal neighborhood} is an open set such that, for any two points in the neighborhood, say $x_{1}$ and $x_{2}$, there exists a unique geodesic contained in the set which connects $x_{1}$ and $x_{2}$. An {\it open conic neighborhood} is a neighborhood which is invariant under the action of $\mathbb{R}_{+}$ by multiplication. We say $\mathcal{O}$ is a {\it causal normal neighborhood} of a Cauchy hypersurface $\Sigma$ of a globally hyperbolic spacetime $(\M,g)$ if $\Sigma$ is a Cauchy hypersurface for $\mathcal{O}$ and, for any two points $x_{1},x_{2}\in\mathcal{O}$, $x_{1}\in J^{+}(x_{2})$, one can find a convex normal neighborhood containing $J^{-}(x_{1})\cap J^{+}(x_{2})$.

Globally hyperbolic spacetimes $\M$ are smooth, orientable, time orientable and paracompact manifolds\footnote{A manifold $\M$ is said to be {\it paracompact} if every open cover of $\M$ has a locally finite subcover.}, also possessing smooth Cauchy hypersurfaces \citep{BernalSanchez03,Wald84}. They have the topological structure $\M = \mathbb{R}\times\Sigma$. In the following, when we refer to a given coordinate chart, the coordinate $x^{0}$ is to be understood as the timelike coordinate, and the coordinates $(x^{1},x^{2},x^{3})$ are to be understood as the spacelike coordinates.

\subsection{Wave equation in globally hyperbolic spacetimes}\label{subsec_wave-eq}

Since throughout this work we will only be interested in sections which are functions (or distributions), we will define differential operators only as operators on functions (distributions). For the corresponding definitions in more general bundles, see \citep{BarGinouxPfaffle07}

A {\it linear differential operator} of order at most $k$ is a $\mathbb{K}$-linear map
\[L:\cC^{\infty}(\M;\mathbb{K})\rightarrow\cC^{\infty}(\M;\mathbb{K}) \; .\]
In a local coordinate chart,
\[Ls=\sum_{|\alpha|\leq k}A^{\alpha}\frac{\partial^{|\alpha|}s}{\partial x^{\alpha}} \; .\]
The summation is taken over all multiindices $\alpha=(\alpha_{0},\alpha_{1},\alpha_{2},\alpha_{3})$ with $\sum_{l=0}^{3}\alpha_{l}\leq k$. If $L$ is a linear differential operator of order at most $k$, but not of order at most $k-1$, $L$ is said to be {\it of order} $k$.

The {\it principal symbol} of the operator $L$ is the map
\[\sigma_{L}:\cT^{\ast}\M\rightarrow \mathbb{K} \; .\]
Take a local coordinate chart in a neighborhood of a point $p\in\M$, in which $L=\sum_{|\alpha|\leq k}A^{\alpha}\partial^{|\alpha|}/\partial x^{\alpha}$. For every $\xi=\sum_{l=0}^{3}\xi_{l}\cdot dx^{l}\in\cT^{\ast}_{p}\M$,
\[\sigma_{L}(\xi)\coloneqq\sum_{|\alpha|=k}\xi_{\alpha}A^{\alpha}(p) \; .\]
The principal symbol of a differential operator is independent of the coordinate chart chosen.

The zeroes of $\sigma_{L}$ outside of the zero section of the cotangent bundle, i.e., the points $(p,\xi)$ with $\xi\in\cT^{\ast}_{p}\M\diagdown\{0\}$ such that $\sigma_{L}(\xi)=0$, are called the {\it characteristics} of $L$. The curves in $\cT^{\ast}\M$ along which $\sigma_{L}$ vanishes identically are called the {\it bicharacteristics} of $L$.

A {\it normally hyperbolic operator} is a second-order differential operator $P$ whose principal symbol is given by the metric, i.e.,
\[\sigma_{P}(\xi)=g^{-1}(\xi,\xi) \; ,\]
where, in a local coordiante chart in a neighborhood of a point $p\in\M$, $g^{-1}(\xi,\xi)=\sum_{\mu,\nu=0}^{3}g^{\mu\nu}\xi_{\mu}\xi_{\nu}$.

A {\it wave equation} is an equation of the form $Pu=f$, where $P$ is a normally hyperbolic operator, the right hand side $f$ is given and the section $u$ is to be determined. It is well known that the wave equation of a massive scalar field in a globally hyperbolic spacetime admits unique retarded and advanced fundamental solutions, which are maps $\mathds{E}^{\pm}:\cC_{0}^{\infty}(\M,\mathbb{K})\rightarrow\cC^{\infty}(\M,\mathbb{K})$, such that, for $f\in\cC_{0}^{\infty}(\M,\mathbb{K})$,
\begin{equation}
 \left(\Box +m^{2}\right)\mathds{E}^{\pm}f=\mathds{E}^{\pm}\left(\Box +m^{2}\right)f=f
 \label{KGfund}
\end{equation}
and
\[\textrm{supp}(\mathds{E}^{\pm}f)\subset J^{\pm}(\textrm{supp}f) \; .\]

The functions $f\in\cC_{0}^{\infty}(\M,\mathbb{K})$ are called test sections and $P\coloneqq \Box +m^{2}$ will denote the differential operator. From the fundamental solutions, one defines the {\it advanced-minus-retarded-operator} $\mathds{E}\coloneqq \mathds{E}^{-}-\mathds{E}^{+}$ as a map $\mathds{E}:C_{0}^{\infty}(\M,\mathbb{K})\rightarrow C^{\infty}(\M,\mathbb{K})$.

As stated in the introduction, it is important to allow the solutions of the wave equation to have singularities. For this purpose, we will now introduce the concept of distributional sections. As a starting point, we will present a notion of convergence on smooth 
functions.

Let $\phi\, ,\, \phi_{n}\in\cC^{\infty}_{0}(\M,\mathbb{K})$, where $n\in\mathbb{N}$. We will denote the sequence $\{\phi_{1},\phi_{2},\ldots\}$ by $(\phi_{n})_{n}$. We say that the sequence $(\phi_{n})_{n}$ converges to $\phi$ in $\cC^{\infty}_{0}(\M,\mathbb{K})$ if the following two conditions hold:
\begin{enumerate}[label=(\alph*)]
\item There is a compact set $K\subset\M$ such that the supports of all $\phi_{n}$ are contained in $K$.
\item The sequence $(\phi_{n})_{n}$ converges to $\phi$ in all $\cC^{k}$-norms over $K$, i.e., for each $k\in\mathbb{N}$
\end{enumerate}
\[\lVert \phi-\phi_{n} \rVert_{\cC^{k}(K)}\underset{n\rightarrow\infty}{\rightarrow}0 \; .\]

A $\mathbb{K}$-linear map $F:\cC^{\infty}_{0}(\M,\mathbb{K})\rightarrow\mathbb{K}$ is called a distribution in $\mathbb{K}$ if it is continuous in the sense that for all convergent sequences $\phi_{n}\rightarrow\phi$ in $\cC^{\infty}_{0}(\M,\mathbb{K})$ one has $F[\phi_{n}]\rightarrow F[\phi]$. We write $\left(\cC^{\infty}_{0}\right)^{'}(\M,\mathbb{K})$ for the space of all distributions in $\mathbb{K}$.

One can also find in \citep{BarGinouxPfaffle07} the observation that any linear differential operator $P:\cC^{\infty}(\M,\mathbb{K})\rightarrow\cC^{\infty}(\M,\mathbb{K})$ extends canonically to a linear operator $P:\left(\cC^{\infty}_{0}\right)^{'}(\M,\mathbb{K})\rightarrow\left(\cC^{\infty}_{0}\right)^{'}(\M,\mathbb{K})$ by
\[(PT)[\phi]\coloneqq T[P^{\ast}\phi]\]
where $\phi\in\cC^{\infty}_{0}(\M,\mathbb{K})$ and $P^{\ast}$ is called the formal adjoint of $P$. If a sequence $(\phi_{n})_{n}$ converges in $\cC^{\infty}_{0}(\M,\mathbb{K})$ to 0, then the sequence $(P^{\ast}\phi_{n})_{n}$ converges to 0 because $P^{\ast}$ is a differential operator. Hence $(PT)[\phi_{n}]=T[P^{\ast}\phi_{n}]\rightarrow 0$. Therefore $PT$ is also a distribution.

We will now characterize the location of the singularities of a distribution, introducing the concept of {\it wavefront sets}. Let $v$ be a distribution of compact support. Its Fourier transform is defined as $\hat{v}(k)=v(e^{i\langle x,k \rangle})$, where $\langle x,k \rangle =\sum_{n=1}^{m}x_{n}k_{n}$, and it possesses the same properties as the usual Fourier transform of functions (see \citep{Hormander-I}, but note that we do not use the same sign conventions as in this reference). If $\forall N \in \mathbb{N}_{0} \; ,\, \exists C_{N} >0$ such that
\begin{equation}
 \lvert \hat{v}(k) \rvert \leqslant C_{N}\left(1+\lvert k \rvert\right)^{-N} \; ,\, k \in \mathbb{R}^{n} \; ,
 \label{regsupp}
\end{equation}
then $v$ is in $C_{0}^{\infty}(\mathbb{R}^{n},\mathbb{K})$. If for a $k \in \mathbb{R}^{n}\diagdown\{0\}$ there exists a cone $V_{k}$ such that for every $p\in V_{k}$, \eqref{regsupp} holds, then $k$ is a direction of rapid decrease for $v$. Accordingly, the singular support ({\it singsupp}) of $v$ is defined as the set of points having no neighborhood where $v$ is in $C^{\infty}$. Moreover, we define the cone $\Sigma(v)$ as the set of points $k \in \mathbb{R}^{n}\diagdown\{0\}$ having no conic neighborhood $V$ such that \eqref{regsupp} is valid when $k \in V$.

For a general distribution $u \in \left(\cC_{0}^{\infty}\right)'(X,\mathbb{K})$, where $X$ is an open set in $\mathbb{R}^{n}$ and $\phi \in \cC_{0}^{\infty}(X,\mathbb{R})$, $\phi(x)\neq 0$, we define
\[\Sigma_{x}(u) \coloneqq \bigcap_{\phi}\Sigma(\phi u) \; .\]

\begin{mydef}\label{smooth-wf}
 If $u \in \left(\cC_{0}^{\infty}\right)'(X,\mathbb{K})$ then the {\it wave front set} of $u$ is the closed subset of $X \times (\mathbb{R}^{n}\diagdown\{0\})$ defined by
\[WF(u)=\{(x,k) \in X \times (\mathbb{R}^{n}\diagdown\{0\}) \arrowvert \, k \in \Sigma_{x}(u)\} \; .\]
\end{mydef}
We remark that the projection of $WF(u)$ on $\M$ is $singsupp\, u$. In \citep{Hormander-I} it was proved that the wave front set of a distribution $u$ defined on a smooth manifold $X$ is a closed subset of $\cT^{\ast}X\diagdown\{0\}$ which is conic in the sense that the intersection with the vector space $\cT^{\ast}_{x}X$ is a cone for every $x\in X$. The restriction to a coordinate patch $X_{\kappa}$ is equal to $\kappa^{\ast}WF(u\circ\kappa^{-1})$.

Moreover, the wave front set provides such a refinement on the location of the singularities of a distribution that we can now define the product of two distributions at the same point, as stated in the following theorem, which we quote without proof from \citep{Hormander-I}:
\begin{thm}
If $u,v\in\left(\cC_{0}^{\infty}\right)'(X,\mathbb{K})$ then the product $uv$ can be defined as the pullback of the tensor product $u\otimes v$ by the diagonal map\footnote{For $x\in X$, $\Delta(x)\coloneqq (x,x)$.} $\Delta:X\rightarrow X\times X$ unless $(x,\xi)\in Wf(u)$ and $(x,-\xi)\in Wf(v)$ for some $(x,\xi)$. When the product is defined we have
\begin{align*}
WF(uv)\subset\{(x,\xi+\eta);\, &(x,\xi)\in WF(u)\textrm{ or }\xi=0, \\
&(x,\eta)\in WF(v)\textrm{ or }\eta=0\} \; .
\end{align*}
\end{thm}

Another useful property of $WF$, which will be used later, is that for two distributions $\phi$ and $\psi$,
\begin{equation}
 WF(\phi +\psi)\subseteq WF(\phi) \cup WF(\psi) \; .
 \label{WFsum} 
\end{equation}
If the $WF$ of one of the distributions is empty, i.e., if one of them is smooth, then this inclusion becomes an equality.

The authors of \citep{DuistermaatHormander72} proved that the singularities of the solutions of a differential operator $P$ with real principal symbol propagate along the bicharacteristics of $P$. This implies that through every point in {\it singsupp} of $u$ ($u$ is a distribution satisfying $Pu=f\in\cC^{\infty}(\M,\mathbb{K})$) there is a bicharacteristic curve which stays in the {\it singsupp}.

There exists a refinement of the notion of wave front set, based on the concept of {\it Sobolev spaces}. A distribution $u \in \left(\cC_{0}^{\infty}\right)'(X,\mathbb{R}^{n})$, where $X$ is an open set in $\mathbb{R}^{n}$, is said to be {\it microlocally} $H^{s}$ at $(x,k) \in \mathbb{R}^{n}\times(\mathbb{R}^{n}\diagdown\{0\})$ if there exists a conic neighborhood $V$ of $k$ and $\phi \in C_{0}^{\infty}(\mathbb{R}^{n})$, $\phi(x)\neq 0$, such that
\begin{equation}
\int_{V}d^{n}k \left(1+\lvert k \rvert^{2}\right)^{s}\lvert [\phi u]^{\wedge}(k) \rvert^{2}\leqslant \infty \; .
\label{Sobolev-space}
\end{equation}
\begin{mydef}\label{Sobolev-wf}
The Sobolev wave front set $WF^{s}$ of a distribution $u \in \left(\cC_{0}^{\infty}\right)'(X,\mathbb{R}^{n})$ is the complement, in ${\mathcal T}^{*}X \diagdown \{0\}$, of the set of all pairs $(x,k)$ at which $u$ is microlocally $H^{s}$.
\end{mydef}
Junker and Schrohe \citep{JunSchrohe02} showed 
that, by choosing a suitable partition of unity, this definition can be extended for any paracompact smooth manifold $\M$.

\subsection{Algebras and states}\label{sec-alg_state}

An algebra is a vector space $\fA$ over a field $K$ such that $\fA$ is equipped with a product law which associates the product $AB$, where $A,B\in\fA$, to an element of $\fA$. This product is also required to be associative and distributive. Throughout this thesis, we will take the field $K$ to be the field of complex numbers $\mathbb{C}$. The mapping $A\in\fA\rightarrow A^{\ast}\in\fA$ is called an involution (or adjoint operation) if $A^{\ast\ast}=A$, $(AB)^{\ast}=B^{\ast}A^{\ast}$ and $(\alpha A+\beta B)^{\ast}=\overline{\alpha}A^{\ast}+\overline{\beta}B^{\ast}$. An involutive algebra is called a $\ast$-algebra.

By a normed algebra we mean an algebra with a norm. This norm defines a topology on the algebra $\fA$, the so-called uniform topology, whose open sets are given by
\[\mathcal{U}(A;\epsilon)=\{B;\, B\in\fA,\, \lVert B-A\rVert<\epsilon\} \; ,\]
where $\epsilon>0$. If the algebra $\fA$ is complete with respect to this norm, then it becomes a Banach algebra. A $\ast$-algebra which is complete with respect to the uniform topology and further has the property $\lVert A^{\ast}\rVert=\lVert A\rVert$ is a Banach $\ast$-algebra.

\begin{mydef}\label{def_Banach-star}
A Banach $\ast$-algebra $\fA$ which has the property
\[\lVert A^{\ast}A\rVert=\lVert A\rVert^{2}\]
for every $A\in\fA$ is a C$^{\ast}$-algebra.
\end{mydef}

States $\omega$ are functionals over an unital $\ast$-algebra $\fA$ with the following properties:
\begin{description}
 \item [Linearity] $\omega(\alpha A+\beta B)=\alpha\omega(A)+\beta\omega(B)$, $\alpha$, $\beta\in \mathbb{C}$, $A$, $B\in \fA$;
 \item [Positive-semidefiniteness] $\omega(A^{\ast}A)\geq 0$;
 \item [Normalization] $\omega(\mathds{1})=1$.
\end{description}
We call a state {\it pure} if it is not a convex combination of two distinct states, i.e.,
\[\nexists \, \omega_{1},\omega_{2}\, \textrm{distinct states over }\fA,\, \, \textrm{and }\lambda \in (0,1)\, \, |\, \omega=\lambda\omega_{1}+(1-\lambda)\omega_{2} \; .\]

The existence of representations of the fields as operators on a certain Hilbert space is achieved by means of the GNS (named after Gel'fand, Naimark and Segal) construction \citep{BraRob-I}, which we will sketch now: given a $\ast$-algebra $\fA$ and a state $\omega$ over this algebra, we can convert this $\ast$-algebra into a pre-Hilbert space (a normed vector space endowed with a positive semi-definite scalar product which is not topologically complete, i.e., the limit of a sequence of vectors does not necessarily lie inside the vector space), by introducing the positive semidefinite scalar product
\[\langle A,B\rangle=\omega(A^{\ast}B) \; .\]
Next we define the left ideal given by
\[\mathcal{N}_{\omega}=\left\{A\, ;A \in \fA\; ,\, \omega(A^{\ast}A)=0\right\} \; .\]
With the help of this ideal, we define equivalence classes
\[\psi_{A}=\{\hat{A};\hat{A}=A+I,\, I\in\mathcal{N}_{\omega}\} \; .\]
This latter space (which is also a vector space) is a strict pre-Hilbert space (a normed vector space endowed with a linear product) with respect to the scalar product
\[\left(\psi_{A},\psi_{B}\right)=\langle A,B\rangle=\omega(A^{\ast}B) \; .\]
$\left(\psi_{A},\psi_{B}\right)$ is clearly independent of the particular representative within the equivalence class. Strict pre-Hilbert spaces can be linearly embedded as a dense subspace of a Hilbert space. We define the completion of this space as the representation space $\mathscr{H}_{\omega}$.

The representatives $\pi_{\omega}(A)$ are defined by their action on the dense subspace of $\mathscr{H}_{\omega}$ spanned by $\psi_{B}$, $B\in\fA$:
\[\pi_{\omega}(A)\psi_{B}=\psi_{AB} \; .\]
This is again independent of the particular representative within the equivalence class. Besides, $\pi_{\omega}(A)$ is a linear operator with bounded closure, which will also be denoted by $\pi_{\omega}(A)$. It is easy to see that
\[\pi_{\omega}(A_{1})\pi_{\omega}(A_{2})\psi_{B}=\psi_{A_{1}A_{2}B} \; ,\]
thus $\pi_{\omega}(A_{1})\pi_{\omega}(A_{2})=\pi_{\omega}(A_{1}A_{2})$. We thus have constructed the representation $\left(\mathscr{H}_{\omega},\pi_{\omega}\right)$.

Now that we have a representation of the algebra $\fA$ as operators on a Hilbert space, it would be interesting to construct a cyclic vector $\Omega_{\omega}\in\mathscr{H}_{\omega}$. We define $\Omega_{\omega}\coloneqq\psi_{\mathds{1}}$. Indeed,
\[\left(\Omega_{\omega},\psi_{A}\Omega_{\omega}\right)=\left(\psi_{\mathds{1}},\psi_{A}\right)=\omega(A) \; .\]
The set $\{\pi_{\omega}(A)\Omega_{\omega};\, A\in\fA\}$ is the dense set of equivalence classes $\{\psi_{A};\, A\in\fA\}$, hence $\Omega_{\omega}$ is cyclic for $\left(\mathscr{H}_{\omega},\pi_{\omega}\right)$. The triple $\left(\mathscr{H}_{\omega},\pi_{\omega},\Omega_{\omega}\right)$ is unique up to unitary equivalence and the representation is irreducible if and only if the state $\omega$ is pure. Furthermore, any vector $\Psi\in\mathscr{H}_{\omega}$ defines a state
\[\omega_{\Psi}(A)=\langle\Psi|\pi_{\omega}(A)|\Psi\rangle \; .\]

We will finish this preliminary section with the definition of KMS states. These states generalize the concept of thermal states to situations where the density matrix cannot be defined \citep{Haag96}.

In the study of nonrelativistic statistical mechanics one usually analyses a system of $\mathcal{N}$ particles (this number is not necessarily fixed) enclosed in a box of finite volume $\mathcal{V}$ with energy $\mathcal{E}$. The thermodynamic limit is taken when $\mathcal{N},\mathcal{V},\mathcal{E}\rightarrow\infty$, but the ratios $\mathcal{N}/\mathcal{V}$, $\mathcal{E}/\mathcal{V}$ remain finite. In this setting, the Gibbs canonical ensemble describes the situation where the temperature is fixed. For quantum systems, to which a hamiltonian $H$ is assigned, the density matrix is defined as
\[\rho_{\beta}=Z^{-1}e^{-\beta H} \quad \textrm{where} \quad Z=\textrm{tr }e^{-\beta H} \; ,\]
where $\beta=(\kappa T)^{-1}$. The density matrix is a trace-class operator with trace $\textrm{tr }\rho_{\beta}=1$. The expectation value of a bounded operator $A$ is given by $\omega_{\beta}(A)=\textrm{tr }\rho_{\beta}A$. If one considers now the time evolution of this operator, given by
\[\alpha_{t}(A)=e^{itH}Ae^{-itH} \; ,\]
together with the cyclicity of the trace, we get (for $B$ another bounded operator)
\begin{align}
\omega_{\beta}(\alpha_{t}(A)B)&=Z^{-1}\textrm{tr }e^{-\beta H}e^{itH}Ae^{-itH}B=Z^{-1}\textrm{tr }Be^{iH(t+i\beta)}Ae^{-itH} \nonumber \\
&=Z^{-1}\textrm{tr }e^{-\beta H}Be^{iH(t+i\beta)}Ae^{-iH(t+i\beta)}=\omega_{\beta}(B\alpha_{t+i\beta}(A)) \; .
\label{KMS-density}
\end{align}

The KMS condition, named after Kubo, Martin and Schwinger, comes from the observation made by the authors of \citep{HHW67} (see also \citep{BraRob-II}) that the equality
\[\omega_{\beta}(\alpha_{t}(A)B)=\omega_{\beta}(B\alpha_{t+i\beta}(A))\]
remains valid even when one cannot define a density matrix. They arrived at this conclusion starting with the two functions
\begin{align*}
F_{A,B}(z)&=\omega_{\beta}(B\alpha_{z}(A)) \\
G_{A,B}(z)&=\omega_{\beta}(\alpha_{z}(A)B) \; ,
\end{align*}
with $z\in\mathbb{C}$. The function $F$ is an analytic function of $z$ in the strip
\[0<\textrm{Im }z<\beta\]
and $G$ is an analytic function of $z$ in the strip
\[-\beta<\textrm{Im }z<0 \; .\]
For real values of $z$, both functions are bounded and continuous, and one obtains $G(t)$ as the boundary value of $F(z)$ as $z\rightarrow t+i\beta$,
\[G_{A,B}(t)=F_{A,B}(t+i\beta) \; .\]
Further properties of KMS states, also in curved spacetimes, can be found in the recent review \citep{Sanders13b}.

\section{Quantized scalar field}\label{sec_fieldquant}

\subsection{General axioms}

As seen in the Introduction, the approach to Quantum Field Theory usual in Minkowski spacetime, starting from the (unique) vacuum state invariant under the Poincar\'e symmetries, is not suitable for a quantum field theory defined in a globally hyperbolic spacetime. One should instead start from the algebra of observables. A set of general axioms that any sensible algebra of observables must satisfy was first written down by Haag and Kastler in 1964 \citep{HaagKastler64}, but only for the case of a quantum field theory in Minkowski spacetime. Those axioms were generalized to quantum fields in globally hyperbolic spacetimes only in 1980 by Dimock \citep{Dimock80}.

The algebra of observables is a C$^{\ast}$-algebra, which means that, after the introduction of a state $\omega$ and the corresponding GNS triple, the elements of the algebra are represented by bounded operators on the Hilbert space $H_{\omega}$. The C$^{\ast}$-algebra can be adequately enlarged in order to rigorously encompass nonlinear functionals of the observables, as well as interacting fields \citep{BruFreKoe96,BruFre00,HoWa02}. On the other hand, it is more convenient to analyse the backreaction of the fields on the background geometry if these are given the structure of a $\ast$-algebra because, after the introduction of a state $\omega$, the elements of a $\ast$-algebra are represented by unbounded operators on the corresponding Hilbert space $H_{\omega}$. This problem will not be dealt with here, but it served as a motivation for this thesis. We will thus present the axioms having in mind that the fields are given the structure of a $\ast$-algebra:

For every bounded region $\mathcal{O}$ of spacetime we will consider the algebra $\fA(\mathcal{O})$ of fields whose support lie within $\mathcal{O}$. $\fA(\mathcal{O})$ is considered to be a unital ${\ast}$-algebra.

{\it Axiom 1}: Isotony \\
If a $\ast$-algebra $\fA$ is assigned to two contractible open and bounded regions $\mathcal{O}_{1}$ and  $\mathcal{O}_{2}$ such that $\mathcal{O}_{2}\supseteq\mathcal{O}_{1}$, then $\fA(\mathcal{O}_{2})\supseteq\fA(\mathcal{O}_{1})$. The $\ast$-algebra $\fA(\M)$ is defined as the union of all $\fA(\mathcal{O})$ with $\mathcal{O}\subseteq\M$.

{\it Axiom 2}: Covariance \\
Let $G$ be the group of isometries of spacetime which preserve orientation and time orientation. Then there should be a representation $\alpha$ of $G$ by automorphisms of $\fA$ such that
\[\alpha_{g}\fA(\mathcal{O})=\fA(g\mathcal{O}) \; .\]

{\it Axiom 3}: Einstein Causality \\
If $\mathcal{O}_{1}$ and  $\mathcal{O}_{2}$ are two causally separated regions, i.e., $\mathcal{O}_{1}\cap J(\mathcal{O}_{2})=\emptyset$, then $\forall a_{1}\in\fA(\mathcal{O}_{1})\, \textrm{and}\, a_{2}\in\fA(\mathcal{O}_{2}),\, \comm{a_{1}}{a_{2}}=0$.

{\it Axiom 4}: Time slice axiom \\
For a globally hyperbolic region $\mathcal{O}$ with a Cauchy hypersurface $\Sigma$, the algebra in any globally hyperbolic neighborhood $\mathcal{O}_{1}$ of $\Sigma$ already contain the information about the algebra in $\mathcal{O}$,
\[\fA(\mathcal{O})\subseteq\fA(\mathcal{O}_{1}) \; .\]

The algebra of free fields, which we will construct in the following, is an example of a $\ast$-algebra satisfying these axioms. Particularly the last axiom is a direct consequence of the fact that the fields are solutions of a hyperbolic differential equation as shown by Dimock \citep{Dimock80}. For interacting fields, this analysis is more involved and it was only recently proved \citep{ChillFre09} that the algebra of interacting fields also satisfies this axiom.

\subsection{The $CCR$ and {\it Weyl} algebras}\label{sec-CCR-Weyl}

The explicit construction of the algebra of solutions of the wave equation is made in a few steps. Firstly we will introduce the so-called {\it Borchers-Uhlmann} algebra, a $\ast$-algebra of smooth sections on the manifold. Afterwards we will require that the elements of this algebra also satisfy the wave equation and the commutation relations, thus obtaining the $CCR$-algebra. At last we will introduce the {\it Weyl} algebra, having the structure of a C$^{\ast}$-algebra. The $CCR$-algebra suffices for the presentation of the first result of this thesis, the construction of States of Low Energy in Expanding spacetimes (see chapter \ref{chap_sle}). The second result of this thesis, the construction of ``Vacuum-like'' Hadamard states (chapter \ref{chap_vac-like}), will not make much use of the Weyl algebra. Our last result, the construction of a Hadamard state in Schwarzschild-de Sitter spacetime (chapter \ref{chap_Sch-dS}), will largely use the Weyl algebra.

We still need to mention a couple of facts before we construct the algebra of fields on a globally hyperbolic spacetime. First, we note that the Lorentzian metric $g$ generates a measure on the spacetime and the inner product on the space of test sections $\cC_{0}^{\infty}(\M,\mathbb{K})$ is defined by
\begin{equation}
 (f,f')_{g}\coloneqq \int\textrm{d}^{4}x\sqrt{|g|}\, \overline{f}(x)f'(x) \; .
  \label{innerprod}
\end{equation}
Second, the fields will be required to solve the Klein-Gordon equation, whose differential operator is $P\coloneqq \Box +m^{2}$. From the fundamental solutions of $P$, one defines the {\it advanced-minus-retarded operator} $\mathds{E}\coloneqq \mathds{E}^{-}-\mathds{E}^{+}$, as explained in subsection \ref{subsec_wave-eq}. Using $\mathds{E}$, the antisymmetric form is defined as
\begin{equation}
 \sigma(f,f')\coloneqq -\int\textrm{d}^{4}x\sqrt{|g|}\, f(x)(\mathds{E}f')(x) \eqqcolon -E(f,f') \; .
 \label{symplform}
\end{equation}
This antisymmetric form is degenerate, because if $f$ and $f'$, both elements of $\cC_{0}^{\infty}(\M,\mathbb{K})$, are related by $f=Pf'$, then $\forall f'' \in \cC_{0}^{\infty}(\M,\mathbb{K})$ we have
\[\sigma(f'',f)=0 \; .\]
Therefore the domain of the antisymmetric form must be replaced by the quotient space\footnote{$\textrm{Ran}P$ is the {\it range} of the operator $P$, that is, the elements $f\in\cC_{0}^{\infty}(\M,\mathbb{K})$ such that $f=Ph$ for some $h\in\cC_{0}^{\infty}(\M,\mathbb{K})$. Moreover, $\mathrm{Ker}\mathds{E}=\textrm{Ran}P$.} $\cC_{0}^{\infty}(\M,\mathbb{K})/\textrm{Ran}P\eqqcolon \K(\M)$.

We now construct the algebra of fields. To each $f\in\K(\M)$ we assign the abstract symbol $\Phi(f)$ and construct the universal tensor algebra:
\[\sA \coloneqq \bigoplus_{n=0}^{\infty}\K(\M)^{\otimes n} \; ,\]
where $\K(\M)^{(0)} \equiv \mathds{1}$. We then endow this algebra with a complex conjugation as a $\ast$-operation and take its quotient with the closed two-sided $\ast$-ideal generated by:
\begin{enumerate}[label=(\roman{*})]
 \item $\Phi(\overline{f})=\Phi(f)^{*}$;
 \item $\Phi(\alpha f+\beta f')=\alpha\Phi(f)+\beta\Phi(f')$, where $\alpha\, ,\, \beta \in\mathbb{C}$;
 \item $\Phi(Pf)=0$;
 \item $\comm{\Phi(f)}{\Phi(f')}=-i\sigma(f,f')\mathds{1}$, where $\comm{\cdot}{\cdot}$ is the commutator and $\mathds{1}$ is the unit element. \label{comm}
\end{enumerate}
The algebra then becomes a unital $\ast$-algebra, the $CCR$-algebra, which will be called $\F$.

The quantum field $\Phi$ is then a linear map from the space of test sections $\mathcal{C}_{0}^{\infty}(\M,\mathbb{K})$ to a unital $\ast$-algebra weakly satisfying the equation of motion (see item (iii) above). Hence $\Phi$, formally written as 
\[\Phi(f)=\int\textrm{d}^{4}x\sqrt{|g|}\, \phi(x)f(x) \; ,\]
may be understood as an algebra valued distributional solution of the Klein Gordon equation. Moreover, a topology can be assigned to this algebra \citep{ArakiYamagami82}, turning $\F$ into a topological, unital $\ast$-algebra. This algebra is unique up to isomorphism.

Dimock \citep{Dimock80} showed that the $CCR$-algebra can be equivalently constructed using the initial-value fields, by setting $\phi = \mathds{E}f$ and $\psi = \mathds{E}f'$, where $f$ and $f'$ are test sections: one defines the restriction operators $\rho_{0}: \phi \mapsto \phi_{\upharpoonright \Sigma} \eqqcolon \phi_{0}$ and $\rho_{1}: \phi \mapsto (\partial_{n}\phi)_{\upharpoonright \Sigma} \eqqcolon \phi_{1}$ (and similarly for $\psi$), where $\partial_{n}$ is the derivative in the direction of the vector $n$, normal to $\Sigma$. The new space of functions is given by
\begin{equation}
 L(\Sigma)=\left\{(\phi_{0},\phi_{1})\in C_{0}^{\infty}(\Sigma,\mathbb{K})\times C_{0}^{\infty}(\Sigma,\mathbb{K})\right\} \; ,
 \label{Cauchysymplspace}
\end{equation}
and the symplectic form, by
\begin{equation}
 \sigma(f,f')=-\int_{\Sigma}\textrm{d}^{3}x\sqrt{|g_{\upharpoonright \Sigma}|}\, \left(\overline{\phi}_{0}(x)\psi_{1}(x)-\psi_{0}(x)\overline{\phi}_{1}(x)\right) \; .
 \label{Cauchysymplform}
\end{equation}
The symplectic form defined above does not dependend on the Cauchy hypersurface on which it is calculated and it is preserved by the isomorphic mapping $\beta:K\rightarrow L\textrm{\ ,\ }\phi\mapsto (\phi_{1},\phi_{2})$.

The symbols $\Phi(f)$ are unbounded. In order to obtain an algebra of bounded operators, we will construct the so-called {\it Weyl algebra} as follows. Operating on the quocient space $\mathcal{C}_{0}^{\infty}(\M,\mathbb{K})/\mathrm{Ker}\mathds{E} \times \mathcal{C}_{0}^{\infty}(\M,\mathbb{K})/\mathrm{Ker}\mathds{E}$, the anti-symmetric form $\sigma$ becomes nondegenerate. We thus define the real vector space $L\coloneqq \mathrm{Re}\left(\mathcal{C}_{0}^{\infty}(\M,\mathbb{R})/\mathrm{Ker}\mathds{E}\right)$ and hence $(L,\sigma)$ is a real symplectic space where $\sigma$ is the symplectic form. From the elements of this real symplectic space one can define the symbols $W(f)$, $f\in L$, satisfying
\begin{enumerate}[label=(\Roman{*})]
 \item $W(0)=\mathds{1}$;
 \item $W(-f)=W(f)^{*}$;
 \item For $f,g\in L$, $W(f)W(g)=e^{-i\frac{\sigma(f,g)}{2}}W(f+g)$.
\end{enumerate}
The relations (II) and (III) are known as {\it Weyl relations}. The algebra constructed from the formal finite sums
\[\W(L,\sigma)\coloneqq\sum_{i}a_{i}W(f_{i})\]
admits a unique $C^{\ast}$-norm \citep{BraRob-II}. The completion of this algebra with respect to this norm is the so-called {\it Weyl algebra}. From the nondegenerateness of the symplectic form one sees that $W(f)=W(g)$ iff $f=g$.

\subsection{Quasifree states and the Hadamard condition}\label{secstates}

On quantum field theory in Minkowski spacetime there is a unique state which is invariant under the action of the Poincar\'e group, the {\it vacuum} state. In a curved spacetime there is, in general, no group of symmetries and, hence, no state can be singled out by invariance arguments. In the particular case of stationary spacetimes, i.e., spacetimes possessing a timelike Killing vector field, one can indeed construct a vacuum state \citep{Kay78,Wald94}.

Before we proceed, let us give some important information about the $n$-point functions corresponding to states. Throughout this thesis we will focus on states which are completely described by their two-point functions, the so-called {\it quasifree states}. All odd-point functions vanish identically and the higher even-point functions can be written as
\[w^{(2n)}_{\omega}(f_{1} \otimes \ldots \otimes f_{2n}) = \sum_{p}\prod_{k=1}^{n} w^{(2)}_{\omega}(f_{p(k)},f_{p(k+n)})\; .\]
Here, $w^{(2n)}_{\omega}$ is the $2n$-point function associated to the state $\omega$, $w^{(2)}_{\omega}(f_{p(k)},f_{p(k+n)})\equiv\omega(f_{p(k)},f_{p(k+n)})$ and the sum runs over all permutations of $\{1,\ldots ,n\}$ which satisfy $p(1)< \ldots <p(n)$ and $p(k)<p(k+n)$.

The two-point function of a state can be decomposed in its symmetric and anti-symmetric parts\footnote{Throughout this thesis, all test functions will be real valued.} ($f_{1},f_{2}\in L$)
\begin{equation}
w_{\omega}^{(2)}(f_{1},f_{2})=\mu(f_{1},f_{2})+\frac{i}{2}\sigma(f_{1},f_{2}) \; ,
\label{quasifree-2ptfcn}
\end{equation}
where $\mu(\cdot,\cdot)$ is a real linear symmetric product which majorizes the symplectic product, i.e.
\begin{equation}
|\sigma(f_{1},f_{2})|^{2}\leq 4\mu(f_{1},f_{1})\mu(f_{2},f_{2}) \; .
\label{mu_geq_sigma}
\end{equation}
The state is pure if and only if the inequality above is saturated, i.e., $\forall f_{1}\in L$,
\begin{equation}
\mu(f_{1},f_{1})=\frac{1}{4}\underset{f_{2}\neq 0}{\textrm{l.u.b.}}\frac{|\sigma(f_{1},f_{2})|^{2}}{\mu(f_{2},f_{2})} \; ,
\label{pure-state}
\end{equation}
where l.u.b. is the {\it least upper bound} (in infinite dimensions, a maximum over $f_{2}\neq 0$ will possibly not be attained). Since the symplectic form is uniquely determined, the characterization of the quasifree state amounts to the choice of the real linear symmetric product $\mu$. Thus, the choice of a Hilbert space is equivalent to the choice of $\mu$. This equivalence is reinforced by the following proposition, whose proof can be found in the appendix A of \citep{KayWald91}:
\begin{prop}\label{one-particle_structure}
Let $L$ be a real vector space on which are defined both a bilinear symplectic form, $\sigma$, and a bilinear positive symmetric form, $\mu$, satisfying \eqref{mu_geq_sigma}. Then, one can always find a complex Hilbert space $\mathscr{H}$, with scalar product $\langle\cdot|\cdot\rangle_{\mathscr{H}}$, together with a real linear map $K:L\rightarrow\mathscr{H}$ such that
\begin{tabbing}
\= (iii) \=  $\sigma(f_{1},f_{2})=2\textrm{Im}\langle Kf_{1}|Kf_{2} \rangle_{\mathscr{H}}$, $\forall f_{1},f_{2}\in L$. \kill
\> (i) \>  the complexified range of $K$, $KL+iKL$, is dense in $\mathscr{H}$; \\
\> (ii) \>  $\mu(f_{1},f_{2})=\textrm{Re}\langle Kf_{1}|Kf_{2} \rangle_{\mathscr{H}}$, $\forall f_{1},f_{2}\in L$; \\
\> (iii) \>  $\sigma(f_{1},f_{2})=2\textrm{Im}\langle Kf_{1}|Kf_{2} \rangle_{\mathscr{H}}$, $\forall f_{1},f_{2}\in L$.
\end{tabbing}
\end{prop}
The pair $(K,\mathscr{H})$ is uniquely determined up to an isomorphism, and it is called the {\it one-particle structure}. Moreover, we have $w_{\omega}^{(2)}(f_{1},f_{2})=\langle Kf_{1}|Kf_{2} \rangle_{\mathscr{H}}$ and the quasifree state with this two-point function is pure if and only if $KL$ alone is dense in $\mathscr{H}$.


In spite of the obvious simplification introduced by the restriction to quasifree states, this is a little bit superfluous for our present purpose, the construction of Hadamard states for quantum field theory in curved spacetimes. As we will see below, the condition for a state to be called a Hadamard state is a restriction on the form of its singularity structure. Recently Sanders \citep{Sanders10} proved that, if the two-point function of the state is of Hadamard form, is a bissolution of the field equations and the commutation relations, with the commutator given by the advanced-minus-retarded operator of the wave operator, then the state os a Hadamard state. This result shows that the analysis of the two-point function of a state suffices to determine whether the state is Hadamard. Since we are interested in concrete examples, we will maintain the restriction of quasifree states. We will now explain the concept of Hadamard states.

The concept of Hadamard states is a local remnant of the spectral condition in Minkowski spacetime, together with the fact that the two-point function of this state is a bissolution of the field equations. There the spectral condition provides sufficient control on the singularities of the $n$-point functions, opening the possibility of extending the states to correlation functions of nonlinear functions of the field as, e.g., the energy momentum tensor. These nonlinear functions are incorporated, in Minkowski spacetime, by means of normal ordering and the Wick product \citep{StreaterWightman64}. The first rigorous form of the two-point function of a Hadamard state was written by Kay and Wald \citep{KayWald91} as
\begin{align}
w_{\omega}^{(2)}(x,y)^{T,n}_{\epsilon} &=\frac{\chi(x,y)}{(2\pi)^{2}}\left(\frac{\Delta^{2}}{\sigma+2i\epsilon(T(x)-T(y))+\epsilon^{2}}+v^{(N)}\ln(\sigma+2i\epsilon(T(x)-T(y))+\epsilon^{2})\right) \nonumber \\
&+H^{(N)}(x,y) \; .
\label{Had-2ptfcn}
\end{align}
Below we will present the modern definition of Hadamard states, which is equivalent to this one and will be used in the following chapters. Here $T$ is a global time function, $x$ and $y$ are causally related points such that $J^{+}(x)\cap J^{-}(y)$ and $J^{+}(y)\cap J^{-}(x)$ are contained within a convex normal neighborhood, $n$ is an integer, $\epsilon$ is a strictly positive real number, $\sigma$ is the squared geodesic distance, $\Delta$ is the van Vleck-Morette determinant,
\[v^{(N)}(x,y)=\sum_{m=0}^{N}v_{m}(x,y)\sigma^{m} \; ,\]
and $H^{(N)}$ is a $\cC^{N}$ function. $\sigma$ is well defined and smooth within an open neighborhood $\mathscr{O}\in\M\times\M$ of $(x,y)$. For $\Sigma$ a Cauchy hypersurface of $\M$ and $N$ a causal normal neighborhood of $\Sigma$, let $\mathscr{O}'\in N\times N$ be an open neighborhood of the set of pairs of causally related points such that the closure of $\mathscr{O}'$ is contained in $\mathscr{O}$. Then $\chi(x,y)$ is such that $\chi(x,y)\equiv 0$ if $(x,y)\notin \mathscr{O}$ and $\chi(x,y)\equiv 1$ if $(x,y)\in \mathscr{O}'$.

It is remarkable that the singular part of this two-point function, the term between parentheses, is a purely geometrical term. The dependence on the state is contained in the $\cC^{N}$ function $H^{(N)}$, whence it is possible to define the renormalized quantum field theory for the whole class of Hadamard states at once. In the late 1970's it was proved that the expectation value of the energy-momentum tensor on such a state in a globally hyperbolic spacetime can be renormalized by the point-splitting procedure \citep{Wald77}.

The result on the last paragraph is based on the assumption that Hadamard states do exist in a general globally hyperbolic spacetime. Such an existence is not obvious, a priori, but it was long shown to be true. Fulling, Sweeny and Wald \citep{FullingSweenyWald78} proved that if the two-point function of a state is of Hadamard form in an open neighborhood of a Cauchy hypersurface, then it retains this form everywhere. A few years later Fulling, Narcowich and Wald \citep{FullingNarcowichWald81} proved that if a globally hyperbolic spacetime contains a static region (an asymptotically flat spacetime, for example), then the vacuum state defined in this region is a Hadamard state, and it retains its singularity structure throughout the spacetime. Furthermore, they showed that one can smoothly deform the geometry to the past of a Cauchy hypersurface $\Sigma_{1}$ so that it becomes ultrastatic\footnote{Ultrastatic spacetimes are static spacetimes such that the timelike Killing vector is everywhere unitary. In the coordinate system of \eqref{metric-timedecomp-static}, $\alpha\equiv 1$.} to the past (say, to the past of another Cauchy hypersurface $\Sigma_{2}$) while retaining global hyperbolicity. Therefore the vacuum state defined in the ultrastatic region will give rise to a Hadamard state in the undeformed region (the future of $\Sigma_{1}$), and this will define a Hadamard state in the whole (undeformed) spacetime. Although rather indirect, for the sake of well-definiteness of the theory, this existence argument suffices, but for practical purposes, a more explicit construction is desired. We will provide some explicit examples of Hadamard states in the next chapters.

The Hadamard condition as presented in equation \eqref{Had-2ptfcn}, although rigorous, makes explicit use of the global time function, which is not uniquely defined. A purely geometrical characterization of Hadamard states was only achieved in the seminal works of Radzikowski (with the collaboration of Verch) \citep{Radzikowski96,RadzikowskiVerch96}, where the Hadamard condition was written in terms of the wave front set of the two-point function corresponding to the state:
\begin{mydef}\label{Hadamard-wf}
A quasifree state $\omega$ is said to be a {\it Hadamard state} if its two-point distribution $w_{\omega}^{(2)}$ has the following wave front set:
\begin{equation}
 WF(w_{\omega}^{(2)})=\left\{\left(x_{1},k_{1};x_{2},-k_{2}\right) | \left(x_{1},k_{1};x_{2},k_{2}\right)\in {\mathcal T}^{*}\left(\M\times\M\right) \diagdown \{0\} ; (x_{1},k_{1})\sim (x_{2},k_{2}) ; k_{1}\in \overline{V}_{+}\right\}
 \label{Wfcond}
\end{equation}
where $(x_{1},k_{1})\sim (x_{2},k_{2})$ means that there exists a null geodesic connecting $x_{1}$ and $x_{2}$, $k_{1}$ is the cotangent vector to this geodesic at $x_{1}$ and $k_{2}$, its parallel transport, along this geodesic, at $x_{2}$. $\overline{V}_{+}$ is the closed forward light cone of $\mathcal{T}^{*}_{x_{1}}\M$.
\end{mydef}
To facilitate the writing, we will call this set $C^{+}$ and say that a quasifree state is Hadamard if its two-point function has this wave front set:
\begin{equation}
WF(w_{\omega}^{(2)})=C^{+}\; .
\label{HadWF-C+}
\end{equation}

Since the antisymmetric part of a Hadamard two-point function is the commutator function $\mathds{E}$, the difference between the two-point functions of different Hadamard states is symmetric. But the symmetric part of the wave front set of a Hadamard function is empty, hence it is a smooth function. This fact will play a fundamental role when we come to the proof that the states which will be constructed below are Hadamard states.

This definition of Hadamard states in terms of the $WF$ set allows the incorporation of interacting field theories at the perturbative level in the algebraic approach \citep{BruFreKoe96,BruFre00,HoWa02}. In this setting, the renormalization of the expectation value of the energy-momentum tensor was done by \citep{Moretti03} (see the recent review \citep{BeniniDappiaggiHack13} for further references). Conversely to the just cited papers, which showed that the Hadamard condition is sufficient to obtain a coherently renormalizable quantum field theory, the authors of \citep{FewsterVerch13} proved that, in ultrastatic slab spacetimes\footnote{Spacetimes possessing a timelike Killing vector which is everywhere orthogonal to the Cauchy hypersurfaces and normalized.} with compact Cauchy hypersurfaces, the requirement that certain Wick polynomials have finite fluctuations enforces the Hadamard condition.

As a concluding remark, we mention that, although the Hadamard condition is physically well-motivated and mathematically well defined, we have no general criterion to select one state among the class of Hadamard states --- actually, such a nonexistence has the status of a no-go theorem \citep{FewsterVerch12,FewsterVerch_sf12} (conversely, a general assignment of the Weyl algebra to any globally hyperbolic spacetime by considering the algebras defined on open subsets isometrically embedded in the manifold was achieved in \citep{BFV03}). However it is known that the set of Hadamard states comprises a {\it locally quasiequivalence class}. We will now present this concept.

In finite dimensional spaces one can start from the Stone-von Neumann Theorem (see Theorem 2.2.1 in \citep{Wald94} and Theorem 7.5 of \citep{Simon72} for a proof):
\begin{thm}
Any two strongly-continuous, irreducible, unitary representations of the Weyl relations over a finite dimensional symplectic space are unitarily equivalent\footnote{By {\it unitarily equivalent} we mean that there is a unitary transformation mapping the operators of one representation to the operators of the other representation.}.
\end{thm}
This theorem has no analog in an infinite-dimensional space. Nonetheless there are other general concepts about representations which will allow us to find classes of states which are physically equivalent.
\begin{mydef}
Two representations $\pi_{1}$ and $\pi_{2}$ of the Weyl algebra $\W$ defined on a globally hyperbolic spacetime $\M$ are called {\it quasiequivalent} if every subrepresentation of the first contains a representation which is unitarily equivalent to a subrepresentation of the second. $\pi_{1}$ and $\pi_{2}$ are called {\it locally quasiequivalent} if the ${\pi_{j}}_{\upharpoonright\W(\mathscr{O})}$ , $j=1,2$, are quasiequivalent for arbitrary open regions $\mathscr{O}$ with compact closure in $\M$.
\end{mydef}

In \citep{ArakiYamagami82} it is proved that
\begin{prop}
If two quasifree states $\omega_{1}$ and $\omega_{2}$ have quasiequivalent GNS representations $\pi_{1}$ and $\pi_{2}$, then the scalar products $\mu_{1}$ and $\mu_{2}$ induce the same topology on $L$, i.e., there are constants $C$ and $C'$ such that, $\forall f\in L$,
\[C\mu_{1}(f,f)\leq\mu_{2}(f,f)\leq C'\mu_{1}(f,f) \; .\]
\end{prop}
Two quasifree states are called quasiequivalent if their corresponding GNS representations are quasiequivalent.

Because of this last Proposition, quasiequivalent states are considered to be physically equivalent states. This result is reinforced by the next theorem, which states that, in some particular subsets of spacetime, any two Hadamard states are locally quasiequivalent (advancing ideas originally presented by Haag, Narnhofer and Stein \citep{HNS84}). The proof of the theorem can be found in \citep{Verch94}.
\begin{thm}
Let $\omega_{1}$ and $\omega_{2}$ be two quasifree Hadamard states on the Weyl algebra $\W(L,\sigma)$ of the Klein-Gordon field in the globally hyperbolic spacetime $(\M,g)$ and let $\pi_{1}$ and $\pi_{2}$ be their associated GNS representations. Then ${\pi_{1}}_{\upharpoonright\W(\mathscr{O})}$ and ${\pi_{2}}_{\upharpoonright\W(\mathscr{O})}$ are quasiequivalent for every open subset $\mathscr{O}\subset\M$ with compact closure.
\end{thm}
This means that any two quasifree Hadamard states are locally quasiequivalent.

\subsection{Quantum Energy Inequalities}\label{sec_QEI}

We have already seen that Hadamard states lead to a sensible renormalization of the expectation value of the energy-momentum tensor, even in curved spacetimes. However, if one evaluates the renormalized expectation value of the energy-momentum tensor at a point of spacetime and tries to minimize this value by changing the Hadamard state, one finds that this value has no lower bound. In other words, the renormalized expectation value of the energy-momentum tensor on a Hadamard state evaluated at a point of spacetime can be arbitrarily negative \citep{EpsGlaJaffe65}. This result is valid both for Minkowski and curved spacetimes.

A negative energy flux would cause the violation of the second law of thermodynamics, as proved by Ford \citep{Ford78}. In this same paper it was shown that such a violation could be avoided if the magnitude and duration of the negative energy flux are limited by an inequality of the form
\begin{equation}
\lvert F \rvert \lesssim\tau^{-2} \; ,
\label{negFlux-bound}
\end{equation}
where $F$ is the flux of negative energy and $\tau$ is the duration of the flux. Some examples where this inequality is satisfied were worked out in this same paper, for certain classes of states in two-dimensional Minkowski spacetime. The same author showed later \citep{Ford91} that inequalities of this form occur also in four dimensional Minkowski spacetime and for all classes of states (in four-dimensional Minkowski spacetime, the negative flux is bounded by $\tau^{-4}$). Instead of calculating the renormalized expectation value at a point of spacetime, now one must multiply it with a peaked function whose integral in time is unity and calculate the flux of this ``smeared'' energy-density.

In curved spacetimes the situation becomes more dramatic. The energy density measured by a causal observer with 4-velocity $u$ is evaluated from the energy momentum tensor by
\begin{equation}
\rho=T_{\mu\nu}u^{\mu}u^{\nu} \; .
\label{energy_density}
\end{equation}
If the energy-momentum tensor is derived from classical fields, it will automatically satisfy the so-called {\it Weak Energy Condition} (WEC),
\begin{equation}
T_{\mu\nu}u^{\mu}u^{\nu}\geq 0 \; .
\label{WEC}
\end{equation}
By continuity, this inequality is also satisfied if the timelike vector $u$ is substituted by a null vector $K$, thus giving rise to the {\it Null Energy Condition} (NEC),
\begin{equation}
T_{\mu\nu}K^{\mu}K^{\nu}\geq 0 \; .
\label{NEC}
\end{equation}
The WEC finds application in the singularity theorems \citep{HawkingEllis73}. The absence of a lower bound to the energy density in curved spacetimes could give rise to pathological spacetimes, with formation of naked singularities \citep{FordRoman90,FordRoman92}, traversable wormholes \citep{Visser95} and allowing faster-than-light travel \citep{Alcubierre94}.

As stated at the beginning of this section, quantum fields do not satisfy the WEC so, in principle, they could give rise to the exotic phenomena described in the above paragraph. Bounds such as the one in equation \eqref{negFlux-bound} largely limit the possibility of occurrence of such phenomena \citep{Ford09}. It is then important to know how general these bounds are.

Fewster \citep{Fewster00} proved such an inequality true for an arbitrary globally hyperbolic spacetime, the so-called {\it Quantum Energy Inequality} (QEI). We will now present a very brief outline of the proof.

Let $(\M,g)$ be an $N$-dimensional lorentzian manifold and consider the $CCR$-algebra $\F$ constructed in subsection \ref{sec-CCR-Weyl}. In order to calculate the expectation value of the energy density in a certain state, we have to take into account that the product of two distributions at the same point of spacetime is an ill-defined object. Hence we start by the tensor product $\phi(x)\phi(x')$, which is an well-defined bidistribution. In this way, the expectation value of the energy density in any such "point-split" state is a well-define quantity. Moreover, as we saw in the previous section, any two Hadamard states differ by only a smooth function. Therefore, let $\gamma$ be a timelike curve in $\M$, parametrized by proper time $\tau$ and with unit tangent vector $u(\tau)$ at $\gamma(\tau)$. We will call $u^{\mu}u^{\nu}T_{\mu\nu}\eqqcolon T$. Let $\omega_{H}$ and $\omega_{H'}$ be two different Hadamard states. The quantity
\begin{equation}
\langle T\rangle_{\omega_{H}}(\tau,\tau')-\langle T\rangle_{\omega_{H'}}(\tau,\tau')
\label{point-split-difference}
\end{equation}
is smooth. We will call each of these terms the expectation value of the regularized energy density, in the respective Hadamard state (denoted below by $\langle T^{reg}\rangle_{\omega}$ when no reference to the spacetime points is made). Since this quantity is smooth, the limit $\tau'\rightarrow \tau$ defines a smooth distribution, which we will call the expectation value of the renormalized energy density in the state $\omega_{H}$ with respect to the state $\omega_{H'}$ (we will omit the reference to the second state, always taking it as an arbitrary, fixed Hadamard state). We will denote it by the symbol $:T:$,
\[\langle :T: \rangle_{\omega}(\tau)=\lim_{\tau'\rightarrow \tau}\left(\langle T \rangle_{\omega_{H}}(\tau,\tau')-\langle T \rangle_{\omega_{H'}}(\tau,\tau')\right) \; .\]
For any Hadamard state $\omega$ on $\F$, we define the expectation value of the renormalized energy density as
\begin{equation}
\rho_{\omega}(\tau)\coloneqq\langle u^{\mu}(\tau)u^{\nu}(\tau):T_{\mu\nu}(\gamma(\tau)): \rangle_{\omega} \; .
\end{equation}
$\rho_{\omega}$ is the restriction to the diagonal $\tau'=\tau$ of the smooth quantity given in equation \eqref{point-split-difference}. This procedure for renormalization is known as the {\it point-splitting procedure} \citep{HackPhD,Moretti03,Wald94}.

Let now $f$ be any smooth compactly supported real-valued function. It follows that
\begin{align*}
\int d\tau(f(\tau))^{2}\rho_{\omega}(\tau)=&\int_{0}^{\infty}\frac{d\alpha}{\pi}\int d\tau d\tau' f(\tau)f(\tau')e^{-i\alpha(\tau-\tau')}\langle :T: \rangle_{\omega}(\tau,\tau') \\
&\int_{0}^{\infty}\frac{d\alpha}{\pi}\int d\tau d\tau' \overline{f}_{\alpha}(\tau)f_{\alpha}(\tau')\langle :T: \rangle_{\omega}(\tau,\tau') \; ,
\end{align*}
where $f_{\alpha}(\tau)\coloneqq f(\tau)e^{i\alpha\tau}$. The restriction of the integration to $(0,\infty)$ is possible because $\langle :T: \rangle_{\omega}$ is symmetric in $\tau$, $\tau'$.

The first observation we make, in order to find the lower bound, is that the distribution $\langle T \rangle_{\omega}(\tau,\tau')$ is a distribution of positive type, i.e.,
\begin{equation}
\int d\tau d\tau' \overline{f}_{\alpha}(\tau)f_{\alpha}(\tau')\langle T \rangle_{\omega}(\tau,\tau')\geq 0 \; .
\label{pos_type}
\end{equation}
This steems from the fact that the set of normals of the map $\varphi:(\tau,\tau')\mapsto(\gamma(\tau),\gamma(\tau'))$, given by
\[N_{\varphi}=\{(\gamma(\tau),k;\gamma(\tau'),k')\,|\,k_{a}u^{a}(\tau)=k'_{b'}u^{b'}(\tau')=0\} \; ,\]
contains only the zero covector (because no nonzero null covector $k$ can annihilate a timelike vector $u$). Now, since $\omega$ is a Hadamard state, $WF(\omega)\cap N_{\varphi}=\emptyset$, therefore the pullback $\varphi^{\ast}\omega$ mapping a distribution defined in $\M\times\M\ni(\gamma(\tau),\gamma(\tau'))$ into a distribution defined in $\mathbb{R}^{2}\ni(\tau,\tau')$ is well defined. That $\langle T \rangle_{\omega}(\tau,\tau')$ is of positive type is a direct consequence of the fact that the original distribution defined in $\M\times\M$ is, by construction, of positive type. Further analysis of $WF(\langle T \rangle_{\omega})$ shows that the integral in equation \eqref{pos_type} decreases rapidly as $\alpha\rightarrow\infty$. Since $\omega_{0}$ is also a Hadamard state, these results are valid to $\langle T \rangle_{\omega_{0}}$ as well.

We are now very close to the result. Since $\omega$ and $\omega_{0}$ are Hadamard states, $\langle :T: \rangle_{\omega}$ is smooth. The fact that $f$ is real-valued allows us to write $\overline{f}_{\alpha}(\tau)$ as $f_{-\alpha}(\tau)$. Therefore
\begin{equation}
\int d\tau d\tau' f_{-\alpha}(\tau)f_{\alpha}(\tau')\langle T \rangle_{\omega_{0}}(\tau,\tau')=[f\otimes f\langle T \rangle_{\omega_{0}}]^{\wedge}(-\alpha,\alpha) \; ,
\end{equation}
where $[\cdot]^{\wedge}(-\alpha,\alpha)$ denotes the Fourier transform of the term between brackets. The result, now, is obvious
\begin{equation}
\int d\tau(f(\tau))^{2}\langle :T: \rangle_{\omega}(\tau,\tau) \geq \int_{0}^{\infty}\frac{d\alpha}{\pi}[f\otimes f\langle T \rangle_{\omega_{0}}]^{\wedge}(-\alpha,\alpha) \; .
\label{QEI}
\end{equation}
The result is valid for any real-valued compactly supported smooth function $f$.

Such an inequality is called {\it difference QEI}. The value of the bound depends explicitly on the reference Hadamard state $\omega_{0}$ and, as pointed in \citep{Fewster00}, depends also on a choice of orthonormal frame. This type of inequality will be recalled in chapter \ref{chap_sle} in the construction of States of Low Energy. In chapter \ref{chap_vac-like} the idea of smearing with real-valued compactly supported smooth functions will be recalled.

Difference QEIs have also been derived for free spin-$1/2$ \citep{FewsterVerch02} and spin-$1$ fields \citep{FewsterPfenning03} and for free spin-$3/2$ fields in Minkowski spacetime \citep{HuLingZhang06,YuWu04}. The QEI was also derived for the nonminimally coupled scalar field \citep{FewsterOsterbrink08}, for interacting fields, from the Operator Product Expansion \citep{BostelmannFewster09} and for the massive Ising model \citep{BostelmannCadamuroFewster13}. A different type of inequality was derived in \citep{FewsterSmith08} for the minimally coupled free scalar field, the so-called {\it Absolute QEI}. In this case, the renormalized energy density was smeared over a timelike submanifold of the spacetime. In this case, the bound does not depend on any reference Hadamard state.

\subsection{Static and Expanding Spacetimes}\label{subsec-static_expand}

On a general globally hyperbolic spacetime, one can always choose a coordinate system on which the metric takes the form \citep{BarGinouxPfaffle07,Wald84}
\begin{equation}
 ds^{2}=\Gamma dt^{2}-h_{t} \; ,
 \label{metric-GH}
\end{equation}
where $\Gamma$ is a positive smooth function and $h_{t}$ is a Riemannian metric on $\Sigma$ depending smoothly on $t\in\mathbb{R}$. But the Klein-Gordon equation arising from such a metric,
\begin{equation}
(\Box +m^{2})\phi=0 \; ,
\label{KG}
\end{equation}
is not, in general, separable. We will analyse two classes of spacetimes for which this equation is separable: the so-called {\it static spacetimes} and the class of {\it expanding spacetimes}. Moreover, in these spacetimes, we will be able to write explicit expressions for the advanced-minus-retarded operator. For simplicity, we will only consider spacetimes with compact Cauchy hypersurfaces.

Static spacetimes are spacetimes with a timelike Killing vector $k$ and with Cauchy hypersurfaces which are orthogonal to the Killing vector; these spacetimes admit a coordinate system in which the metric assumes the form
\begin{equation}
ds^{2}=\alpha^{2}(\uline{x})\left(dt^{2}-h_{ij}(\uline{x})d\uline{x}^{i}d\uline{x}^{j}\right) \; ,
\label{metric-timedecomp-static}
\end{equation}  
where all coefficients are smooth functions defined on a Cauchy hypersurface $\Sigma$ \citep{Fulling89}. Throughout this thesis, $\uline{x}$ will denote the spatial coordinates of a point on the manifold. We will also use $x=(t,\uline{x})$ for the spacetime coordinates.

In static spacetimes, the Klein-Gordon equation \eqref{KG} becomes
\begin{equation}
 \alpha^{-2}\left(\frac{\partial^{2}}{\partial t^{2}}+K\right)\phi=0 \; ,
  \label{KG-static}
\end{equation}
where
\[K=-\alpha^{2}\left[\frac{1}{\sqrt{|g|}}\partial_{j}(\sqrt{|g|}h^{jk}\partial_{k})+m^{2}\right] \; ,\]
($|g|=|\textrm{det}(g_{\mu\nu})|$). On the Hilbert space $L^2(\Sigma,\alpha^{-2}\sqrt{|g|})$, the operator $K$ is symmetric and positive. According to Kay \citep{Kay78}, if the spacetime is uniformly static, i.e., $\alpha$ is bounded fom above and from below away from zero, the operator $K$ is even essentially self-adjoint on the domain $\mathcal{C}^{\infty}_0(\Sigma,\mathbb{C})$. We give in Theorem \ref{theorem_essential_self-adjointness} below a proof of essential self-adjointness of $K$ without any assumptions on $\alpha$. The self-adjoint closure of $K$ will be denoted by the same symbol. Due to the compactness of $\Sigma$, it has a discrete spectrum with an orthonormal system of smooth eigenfunctions $\psi_{j}$ and positive eigenvalues $\lambda_j$, $j\in\mathbb{N}$ with $\lambda_{j}\geq\lambda_{k}$ for $j>k$. Moreover, due to Weyl's estimate, the sums $\sum_{j}\lambda_{j}^{-p}$ converge for $p>\frac{d}{2}$ where $d$ is the dimension of $\Sigma$ \citep{Jost11}.

The advanced-minus-retarded-operator, in this case, has the integral kernel
\begin{equation}
\mathds{E}(t,\uline{x};t',\uline{x'})=-\sum_{j}\frac{1}{\omega_{j}}\sin((t-t')\omega_{j})\psi_{j}(\uline{x})\overline{\psi}_{j}(\uline{x'})\; ,
\label{E-static}
\end{equation}
with $\omega_{j}=\sqrt{\lambda_j}$ \citep{Fulling89}. This sum converges in the sense of distributions, i.e., if $f,f'\in\cC_{0}^{\infty}(\M,\mathbb{C})$, then
\[\left|\int f(x)\mathds{E}(x;x')f'(x')\sqrt{|g(x)|}\sqrt{|g(x')|}dxdx'\right|\]
gives a finite number. This sort of convergence is also usually called {\it weak convergence}. We will now present the proof of essential self-adjointness of the operator $K$.

\begin{thm}\label{theorem_essential_self-adjointness}
Let $\M=\mathbb{R}\times \Sigma$ be a globally hyperbolic spacetime with metric $g=\alpha^2(dt^2-h)$ where $h$ is a Riemannian metric on the manifold $\Sigma$ and $\alpha$ a smooth nowhere vanishing function on $\Sigma$. Then
\begin{enumerate}
\item \label{a} $(\Sigma,h)$ is a complete metric space.
\item \label{b} The d'Alembertian  on $\M$  is of the form
\[\square_g=\alpha^{-2}(\partial_{t}^{2}+K)\]
where $K$ is an elliptic differential operator on $\Sigma$ which is positive and self-adjoint on $L^2(\Sigma,\alpha^{-2}\sqrt{|\mathrm{det}g|})$ with core $\mathcal{C}^{\infty}_0(\Sigma)$. 
\end{enumerate}
\end{thm}
\begin{proof}
We follow the papers of Chernoff \citep{Chernoff73} and Kay \citep{Kay78}.

\eqref{a} The spacetime $\M$ is conformally equivalent to an ultrastatic spacetime with metric $dt^2-h$. The latter is globally hyperbolic if and only if $(\Sigma,h)$ is complete \citep{Kay78}. Hence the same holds true for $\M$.

\eqref{b} In local coordinates, $K$ assumes the form
\[K=-\alpha^2\left(\gamma^{-1}\partial_j\gamma \alpha^{-2}h^{jk}\partial_k+m^{2}\right)\]
with $\gamma=\sqrt{|\mathrm{det}g|}$. The principal symbol of $K$ is $\sigma_K=h^{jk}\xi_j\xi_k$, hence $K$ is elliptic. Moreover, on $L^2(\Sigma,\gamma \alpha^{-2})=:\mathcal{H}$ we have, for $\phi,\psi\in\mathcal{C}^{\infty}_0(\Sigma)$
\begin{align*}
\langle \phi,K\psi\rangle=-\int dx\,\overline{\phi(x)}\partial_j\gamma \alpha^{-2}h^{jk}\partial_k\psi(x)=\int dx\, \gamma \alpha^{-2}h^{kj}\overline{\partial_j\phi(x)}\partial_k\psi(x)=\langle K\phi,\psi\rangle \ ,
\end{align*}
hence $K$ is symmetric and positive, thus one can form the Friedrichs extension and obtains a self-adjoint positive operator.

It remains to prove that $K$ is essentially self-adjoint on $\mathcal{C}^{\infty}_0(\Sigma)$. For this purpose we use a variation of the method of Chernoff and exploit the fact that the Cauchy problem for normally hyperbolic differential equations is well posed on globally hyperbolic spacetimes.

Let $V(t)$ denote the operator on $\mathcal{D}:=\mathcal{C}^{\infty}_0(\Sigma)\oplus \mathcal{C}^{\infty}_0(\Sigma) $ defined by
\[V(t)\left(\begin{array}{c}\phi_1\\ \phi_2\end{array}\right)=\left(\begin{array}{c}\phi(t)\\ \dot{\phi}(t)\end{array}\right)\] 
where $t\mapsto \phi(t)$ is the solution of the Cauchy problem with initial conditions $\phi(0)=\phi_1$ and $\dot{\phi}(0)=\phi_2$.
The 1-parameter group $t\mapsto V(t)$ satisfies the differential equation
\[\frac{d}{dt}V(t)=iAV(t)\]
with
\[iA=\left(\begin{array}{cc}0 & 1\\-K & 0\end{array}\right)\ . \]
We now equip $\mathcal{D}$ with a positive semidefinite scalar product such that $V(t)$ becomes unitary. We set
\[\langle \phi,\psi\rangle=\langle \phi_1,K\psi_1\rangle+\langle\phi_2,\psi_2\rangle\]
where on the right hand side we use the scalar product of $\mathcal{H}=L^2(\Sigma,\gamma \alpha^{-2})$. (The first component of the scalar product vanishes for $m^2=0$ on functions which are constant on every connected component of $\Sigma$, hence if $\Sigma$ has compact connected components the scalar product is not definite). Then
\[\frac{d}{dt}\langle V(t)\phi,V(t)\psi\rangle=0\]
which implies that $V(t)$ is unitary.

We proceed similarly to the proof of Chernoff's Theorem \citep{Chernoff73}. Let $\psi\in \mathcal{H}=L^2(\Sigma,\gamma\alpha^{-2})$ such that for all $\phi\in\mathcal{C}^{\infty}_0(\Sigma)$ 
\[\langle K\phi,\psi\rangle=i\langle \phi,\psi\rangle \ .\]
We consider the function
\[f:t\mapsto \left\langle V(t)\left(\begin{array}{c}0\\ \phi\end{array}\right),\left(\begin{array}{c}0\\ \psi\end{array}\right)\right\rangle\]
Due to the unitarity of $V(t)$, this function is bounded. By the assumption on $\psi$, it satisfies the differential equation
\[\frac{d^2}{dt^2}f(t)=-\left\langle K\left(V(t)\left(\begin{array}{c}0\\ \phi\end{array}\right)\right)_2,\psi\right\rangle=-i\left\langle \left(V(t)\left(\begin{array}{c}0\\ \phi\end{array}\right)\right)_2,\psi\right\rangle=-if(t)\ .\]
But the only bounded solution of this equation vanishes, hence $\psi$ is orthogonal to $\mathcal{C}^{\infty}_0(\Sigma)$ on the Hilbert space $\mathcal{H}$, hence $\psi=0$. The same argument holds if we replace $i$ by $-i$ in the defining condition on $\psi$.
This proves that $K$ is essentially self-adjoint on $\mathcal{H}$.

\end{proof}

We will return to the subject of quantum fields in static spacetimes in chapter \ref{chap_vac-like} when we will construct Hadamard states in these spacetimes. Now we will present the other class of spacetimes which will appear in this thesis.

The class of expanding spacetimes comes from the general case after we make the assumptions that $\Gamma\equiv 1$ and that the metric on the spatial hypersurfaces can be written in the following form:
\begin{equation*}
h_{t}=c(t)^{2}h_{ij}(\uline{x})d\uline{x}^{i}d\uline{x}^{j}\; .
\end{equation*}
Here, $c(t)$ is a smooth positive function of time, the so-called {\it scale factor}, and $h_{ij}(\uline{x})$ is again the metric on the Riemannian hypersurfaces. The metric assumes the form
\begin{equation}
ds^{2}=dt^{2}-c(t)^{2}h_{ij}(\uline{x})d\uline{x}^{i}d\uline{x}^{j}\; ,
\label{metric-timedecomp-expanding}
\end{equation}
and the Klein-Gordon equation for a scalar field $\phi(t,\uline{x})$ assumes, then, the form
\begin{equation}
\left(\partial_{t}^{2}+3\frac{\dot{c}(t)}{c(t)}\partial_{t}-\frac{\Delta_{h}}{c(t)^{2}}+m^{2}\right)\phi(t,\uline{x})=0 \; .
\label{KG-expanding}
\end{equation}

The Laplace operator $-\Delta_{h}$ is essentially self-adjoint on the compact Riemannian space $(\Sigma,h)$. Its unique self-adjoint extension (denoted by the same symbol) is an operator on $L^{2}(\Sigma,\sqrt{\lvert h\rvert})$ with discrete spectrum \citep{Jost11}. Again we use the orthonormal basis of eigenfunctions $\psi_{j}$ and the associated nondecreasing sequence of positive eigenvalues $\lambda_j$ of $-\Delta_h$. An ansatz for a solution is
\begin{equation}
 \phi(t,\uline{x})=T_{j}(t)\overline{\psi}_{j}(\uline{x}) \; ,
 \label{KGfield}
\end{equation}
and $T_{j}$ must satisfy
\begin{equation}
 \left(\partial_{t}^{2}+3\frac{\dot{c}(t)}{c(t)}\partial_{t}+\omega_{j}^{2}(t)\right)T_{j}(t)=0 \; ,
 \label{timeKG}
\end{equation}
where
\begin{equation}
\omega_{j}^{2}(t)\coloneqq\frac{\lambda_{j}}{c(t)^{2}}+m^{2} \; .
\label{frequency}
\end{equation}

The two linearly independent real valued solutions of equation \eqref{timeKG} can be combined in a complex valued solution which satisfies the normalization condition
\begin{equation}
\dot{T}_{j}(t)\overline{T}_{j}(t)-T_{j}(t)\dot{\overline{T}}_{j}(t)=\frac{i}{c(t)^{3}} \; .
\label{timeSymplprod}
\end{equation}
Since the left-hand side is the Wronskian $W[T_{j},\overline{T}_{j}]$, $T_{j}(t)$ and $\overline{T}_{j}(t)$ are linearly independent. From this linear independence, if $S_{j}(t)$ and $\overline{S}_{j}(t)$ are also linearly independent solutions of \eqref{timeKG}, we can write
\begin{equation}
 T_{j}(t)=\alpha_{j}S_{j}(t)+\beta_{j}\overline{S}_{j}(t)\; .
 \label{Bogtransf}
\end{equation}
Since $S_{j}(t)$ must also satisfy \eqref{timeSymplprod}, the parameters $\alpha_{j}$ and $\beta_{j}$ are then subject to
\begin{equation}
 |\alpha_{j}|^{2}-|\beta_{j}|^{2}=1\; .
 \label{Bogcoef}
\end{equation}
The solutions to equation \eqref{timeKG} have only two free parameters. On the other hand, taking into account the absolute values and phases of $\alpha_{j}$ and $\beta_{j}$, subject to \eqref{Bogcoef}, we would have three free parameters. One of these is then a free parameter. In chapter \ref{chap_sle} we will choose $\beta_{j}$ to be a real parameter. In chapter \ref{chap_vac-like} the relative phase will be chosen differently.

The advanced-minus-retarded operator now has the integral kernel
\begin{equation}
 \mathds{E}(t,\uline{x};t',\uline{x'})=\sum_{j}\frac{(\overline{T}_{j}(t)T_{j}(t')-T_{j}(t)\overline{T}_{j}(t'))}{2i}\psi_{j}(\uline{x})\overline{\psi}_{j}(\uline{x'}) \; .
 \label{E-prop}
\end{equation}
This sum converges in the sense of distributions, as the sum in equation \eqref{E-static}.

\chapter{States of Low Energy}\label{chap_sle}

In this chapter we will construct explicit examples of Hadamard states, the {\it States of Low Energy} (SLE). These will be the states that minimize the smeared renormalized expectation value of the energy density in expanding spacetimes. The existence of these states fills a gap left open by the QEIs in the sense that those inequalities prove that the smeared renormalized expectation value of the energy density has a lower bound (where the expectation value is taken on a Hadamard state), but they do not show that it is possible to construct a Hadamard state on which this quantity takes its minimum value. Moreover, as will be seen later, these states reduce to the vacuum state in static spacetimes. Therefore, this is a direct generalization of the concept of vacuum state to any expanding spacetime. Furthermore, we generalize the construction given in \citep{Olbermann07} of SLE in Robertson-Walker spacetimes.

On section \ref{sec_2ptfcn-expand-hom} we will construct two-point functions from the solutions of the Klein-Gordon equation in general expanding spacetimes and in expanding spacetimes with homogeneous spatial hypersurfaces, both with compact Cauchy hypersurfaces (the definition of homogeneity will be presented below). On section \ref{SLE-est} we will construct the SLE. On section \ref{SLE-Hadamard} we will show that the SLE are Hadamard states.

\section{Quasifree states in Expanding and Homogeneous Spacetimes}\label{sec_2ptfcn-expand-hom}

\subsection{Expanding Spacetimes}

In this subsection, we are going to construct quasifree states in the spacetimes which we denoted {\it expanding spacetimes} (see section \ref{subsec-static_expand}). We will show how the two-point functions corresponding to quasifree states arise directly from the GNS construction.

In the expanding spacetimes the KG operator separates as a laplacian operator on the Cauchy hypersurfaces and an ordinary differential operator (see the discussion between equations \eqref{KG-expanding} and \eqref{frequency}). Given a state $\omega$ and the corresponding cyclic vector $\Omega_{\omega}\in\mathscr{H}_{\omega}$, we will expand the representation of the field in terms of the operator $a$ and its adjoint, $a^{\dagger}$, such that
\[a|\Omega_{\omega}\rangle=0\; .\]
The representation of the field $\phi$ (taken as an element of the $CCR$-algebra) on the Hilbert space generated by the GNS construction will be denoted by the symbol $\hat{\phi}$. It can be expressed as
\begin{equation}
 \hat{\phi}(t,\uline{x})=\frac{1}{\sqrt{2}}\sum_{j}\left[a_{j}\overline{T}_{j}(t)\psi_{j}(\uline{x})+a^{\dagger}_{j}T_{j}(t)\overline {\psi}_{j}(\uline{x})\right]\; .
 \label{fieldrepr}
\end{equation}
The proof, to be given in section \ref{subsecfulfillhadamard}, that the arising state is a Hadamard state, will also entail that this sum converges in the sense of distributions. The operator $a$ and its adjoint also follow the mode decomposition. The initial-value expression of the commutation relation satisfied by the fields (equation \eqref{Cauchysymplform}), the orthonormality of the eigenfunctions of the laplacian and the normalization satisfied by the functions $T_{j}(t)$ (equation \eqref{timeSymplprod}) imply that the operators $a$ and $a^{\dagger}$ satisfy the usual commutation relations: \[\Big[a_{j},a_{j'}\Big]=\Big[a^{\dagger}_{j},a^{\dagger}_{j'}\Big]=0\textrm{\quad} \Big[a_{j},a^{\dagger}_{j'}\Big]=\delta_{jj'}\; ,\]
where $\delta_{jj'}$ is the Kronecker delta.

Evaluated on the state $\omega$, the two-point function of this field operator is ($f$ and $f'$ are test functions of compact support)
\begin{align}
 w^{(2)}_{\omega}\left(\phi(f)\phi(f')\right) &=\langle \Omega_{\omega}|\hat{\phi}(f)\hat{\phi}(f')|\Omega_{\omega} \rangle \nonumber \\
 &=\int\textrm{d}^{4}x\sqrt{|g(x)|}\textrm{d}^{4}x'\sqrt{|g(x')|}\, f(t,\uline{x})f'(s,\uline{x'})\sum_{j}\overline{T}_{j}(t)T_{j}(s)\psi_{j}(\uline{x})\overline{\psi}_{j}(\uline{x'}) \nonumber \\
 &\eqqcolon\int\textrm{d}^{4}x\sqrt{|g(x)|}\textrm{d}^{4}x'\sqrt{|g(x')|}\, f(t,\uline{x})f'(s,\uline{x'})w^{(2)}_{\omega}(\uline{x},\uline{x'}) \; .
 \label{GNS2ptfcn}
\end{align}
Regarding the convergence of this sum, we remark that on section \ref{subsecfulfillhadamard} we will show that the state whose two-point function is given by \eqref{GNS2ptfcn} is a Hadamard state.

This simple expression is valid for any expanding spacetime with compact Cauchy hypersurfaces. In the case the spacetime possesses symmetries, one can proceed differently. The analysis we will present now is valid for spacetimes whose compact spatial hypersurfaces possess a group of symmetry acting on them. For a construction of quasifree Hadamard states on symmetric spacetimes whose Cauchy hypersurfaces are not necessarily compact, see \citep{AvetisyanVerch12}.

\subsection{Homogeneous Spacetimes}

The spatial hypersurfaces are Riemannian submanifolds, and we will now present the definition of homogeneity on them \citep{Jost11}.

\begin{mydef}
 Let $G$ be a group of isometries from the Riemannian manifold $\Sigma$ to itself, i.e., $g\in G$ is a diffeomorphism from $\Sigma$ to itself and $\forall g\in G$, $g^{*}h=h$, where $h$ is the metric on $\Sigma$. If for every pair of points $p,q\in\Sigma\textrm{, }\exists g'\in G$ such that $g'p=q$, then the group $G$ is said to act transitively on $\Sigma$. A Riemannian manifold with a transitive group of isometries is called {\it homogeneous}.
\end{mydef}

The action of the group $G$ as a group of isometries at a point $p\in\Sigma_{t}$, where $\Sigma_{t}$ is a Cauchy hypersurface labelled by the value of the time parameter, is $G:\Sigma_{t}\ni p \mapsto gp\in\Sigma_{t}$. The homogeneous spaces can be classified according to their Lie-group structure \citep{Osinovsky73} and are designated as Bianchi I-IX spaces. From this classification, one can construct globally hyperbolic spacetimes whose Cauchy hypersurfaces are isometric to one of those homogeneous spaces. Such spacetimes are called {\it Bianchi spacetimes}.

We are interested in spacetimes with compact Riemannian hypersurfaces without boundary because, in this case, each eigenvalue of the laplacian has finite multiplicity \citep{Berard86}, thus simplifying both the mode decomposition presented earlier and the construction of symmetric states\footnote{The problem of mode decomposition for spacetimes with noncompact homogeneous Riemannian hypersurfaces was treated in \citep{Avetisyan12}. We are indebted to the author of that paper for stressing the validity of our treatment.}.  The Bianchi spaces I-VIII are topologically equivalent to $\mathbb{R}^{3}$, therefore they are noncompact spaces. The symmetry structure of the Bianchi IX space is given by the SU(2) group, which is already compact. Among the noncompact ones, the simplest is Bianchi I, which has a commutative group structure. One can form a compact space from the Bianchi I space by taking the quotient between this group and the group $\mathbb{Z}$ of integer numbers. The resulting space is the 3-torus, a compact space without boundary. We remark that there exists more than one method of compactification (see \citep{Tanimoto04} and references therein).

$G$ has a unitary representation $U\oplus U$ on $L^2(\Sigma,\sqrt{|h|})\oplus L^2(\Sigma,\sqrt{|h|})$, given by $U(g)f=f\circ g^{-1}$, such that
\begin{equation}
 \alpha_{g}\left(\phi(f)\right) = \phi(f\circ g^{-1}) \; .
 \label{automorph}
\end{equation}
A quasifree state $\omega$ is said to be symmetric if $\forall g\in G$, $\omega \circ \alpha_{g} =\omega$. The quasifree state symmetric under the action of the group $G$ will be denoted by $\omega_{G}$.

The Riemannian metric on $\Sigma$ induces a scalar product on $L^2(\Sigma)\oplus L^2(\Sigma)$. Working with the initial-value fields (see section \ref{sec-CCR-Weyl}) $F=\left(F_{0},F_{1}\right)=\left(\rho_{0}\mathds{E}f,\rho_{1}\mathds{E}f\right)$ and $F'=\left(F'_{0},F'_{1}\right)=\left(\rho_{0}\mathds{E}f',\rho_{1}\mathds{E}f'\right)$, we have
\begin{equation}
 \left(F,F'\right)_{L^{2}}=\int\textrm{d}^{3}\uline{x}\sqrt{|h|}\, \left(\overline{F}_{0}F'_{0}+\overline{F}_{1}F'_{1}\right) \; .
\end{equation}

Schwarz's nuclear theorem \citep{ReedSimon-I} states that to any two-point function $S$ in the space of initial-value fields is associated an element of the dual to $L^2(\Sigma)$:
\[L^2(\Sigma)\ni F' \, \mapsto \, S(\cdot ,F')\in \left(L^2(\Sigma)\right)^{*} \; .\]
Now, from Riesz's representation theorem, to the element $S(\cdot ,F')\in \left(L^2(\Sigma)\right)^{*}$ there exists associated an element $\hat{S}(F')\in L^2(\Sigma)$ such that, $\forall F\in L^2(\Sigma)$,
\begin{equation}
 S(F,F')=\left(F,\hat{S}(F')\right) \; .
 \label{twoptriesz}
\end{equation}

Using the eigenfunctions $\psi_{j}$ of the laplacian, a generalized Fourier transform can be defined:
\begin{equation}
 \widetilde{F}_{j}:=\left(\psi_{j},F\right)_{L^2}=\int \textrm{d}^{3}\uline{x}\sqrt{|h|}\, \overline{\psi}_{j}(\uline{x})F(\uline{x}) \; .
 \label{Fouriertransf}
\end{equation}
The fact that the $\psi_{j}$ form a complete basis of orthonormal eigenfunctions of the laplacian operator allows us to write
\begin{equation}
 S\left(F,F'\right)=\sum_{j}\langle \widetilde{\overline{F}}_{j},\widetilde{\hat{S} (F')}_{j} \rangle \; ,
\end{equation}
where
\begin{equation}
 \langle \widetilde{\overline{F}}_{j},\widetilde{\hat{S} (F')}_{j} \rangle =\sum_{l=0}^{1}\widetilde{\overline{F}_{l}}_{j}\left(\widetilde{\hat{S} (F')_{l}}\right)_{j} \; .
\end{equation}
This last sum is over $l$ ranging from $0$ to $1$ because $F$ and $F'$ represent the initial-value fields $\left(F_{0},F_{1}\right)$ and $\left(F'_{0},F'_{1}\right)$.

We have thus constructed a homogeneous two-point function. Kolja Them (\citep{Them10}, in German) showed that this two-point function gives rise to a quasifree homogeneous state. This proof follows from \citep{LuRo90}, the only difference being that, in RW spacetimes, the commutant of each symmetry group consists of diagonalizable operators (see Appendix A of that reference), i.e., operators $T$ such that
\[\left(\widetilde{Tf}\right)_{j}=t_{j}\widetilde{f}_{j} \; ,\]
and this will not be generally true in our case. There, this fact led to the conclusion that the operation of $\hat{S}$ on a test function, evaluated in Fourier space, simply amounted to a multiplication by a function of the mode, i.e., $\widetilde{\hat{S} (F')}_{j}=\widetilde{\hat{S}}_{j}\widetilde{F'}_{j}$, but this is not true here. Nevertheless the same analysis made there for the operator $\widetilde{\hat{S}}_{j}$ can be made here for $\widetilde{\hat{S} (F')}_{j}$. In our case, this results in the construction of a quasifree homogeneous state, while there the state was also isotropic.

The two-point function of the homogeneous state has integral kernel
\begin{equation}
 w^{(2)}_{\omega_{G}}(x,x')=\sum_{j}\overline{T}_{j}(t)T_{j}(t')\psi_{j}(\uline{x})\overline{\psi}_{j}(\uline{x'})
 \label{2ptfcnHom}
\end{equation}
and $T_{j}$ has initial conditions at time $t_{0}$ given by
\[T_{j}(t_{0})=q_{j} \; \; , \; \; \dot{T}_{j}(t_{0})=c^{-3}(t_{0})p_{j} \; ,\]
where $q_{j}$ and $p_{j}$ are polynomially bounded functions. Elliptic regularity guarantees the boundedness of $\psi_{j}(\uline{x})\overline{\psi}_{j}(\uline{x'})$.

We note that this two-point function has the same form as \eqref{GNS2ptfcn}, but there it was formulated for general expanding spacetimes, therefore it has a wider range of aplicability than the one here.

\section{States of Low Energy in Expanding Spacetimes}\label{SLE-est}

We will now construct the States of Low Energy in expanding spacetimes without spatial symmetries but with compact Cauchy hypersurfaces without boundary. We will show that the construction in homogeneous spacetimes is a particular case of the one presented in this section. The SLE will be chosen by a simple minimization process as those states on which the expectation value of the renormalized energy density is minimized. Besides, we will point out the differences between our construction and the original one, given in \citep{Olbermann07}.

The renormalization of the energy density will be attained by means of the point-splitting method. In the absence of spatial symmetries, the energy density must be dependent on position. Therefore we will need to smear its renormalized expectation value over a spatially extended spacelike submanifold. Since the Cauchy hypersurfaces are compact, we can perform the smearing with test functions which do not depend on the spatial position. We want to stress here that in the RW spacetimes with negative or null spatial curvature there was no need to integrate in space, but this was necessary in the case of positive spatial curvature \citep{Olbermann07}.

We will choose as observers a congruence of geodesic curves which are everywhere orthogonal to the Cauchy hypersurfaces. We also require that their four-velocity is future-pointing. This means that for every such observer, its four-velocity $\dot{\gamma}$ is orthogonal to every vector $X\in\mathcal{T}_{p}\Sigma$, for every point $p$ in $\Sigma$. In the coordinate system we have chosen, with metric of the form \eqref{metric-timedecomp-expanding}, the ortogonality condition becomes
\begin{equation}
 g(\dot{\gamma},X)=-c^{2}(t)h_{kl}(\uline{x})\dot{\gamma}^{k}(t,\uline{x})X^{l}(t,\uline{x})=0\, \therefore \dot{\gamma}^{k}\equiv 0 \; .
\end{equation}
We also require that the four velocity is normalized, thus
\begin{equation}
 g(\dot{\gamma},\dot{\gamma})=(\dot{\gamma}^{0})^{2}=1\, \therefore \dot{\gamma}^{0}=1 \; .
\end{equation}

The energy density measured by the chosen observers is evaluated from the energy-momentum tensor $T_{\mu\nu}(x)$ as
\begin{equation}
 \rho(x)=T_{\mu\nu}(x)\dot{\gamma}^{\mu}\dot{\gamma}^{\nu}=T_{00}(x)\; .
\end{equation}

The expectation value of the regularized energy density on a state $\omega$ is
\begin{equation}
 \langle \hat{T}^{reg}\rangle_{\omega}(x,x')=\left[\frac{1}{2}\nabla_{0}|_{x}\nabla_{0}|_{x'}+\frac{1}{2}\nabla^{c}|_{x}\nabla_{c}|_{x'}+\frac{1}{2}m^{2}\right]w^{(2)}_{\omega}(x,x')\; .
 \label{Treg_u}
\end{equation}
The expectation value of the renormalized energy density is encountered by subtracting from the above expression the expectation value of the regularized energy density on a reference Hadamard state $\omega_{0}$,
\begin{equation}
 \langle \hat{T}^{ren}\rangle_{\omega}=\langle \hat{T}^{reg}\rangle_{\omega}-\langle \hat{T}^{reg}\rangle_{\omega_{0}} \; ,
\label{pointsplit-general}
\end{equation}
and then taking the coincidence limit $x'\rightarrow x$. This limit is taken by parallel transporting the point $x'$ to the point $x$ along the unique geodesic connecting them. This process introduces additional terms to the expectation value of the renormalized energy density, which are purely geometrical terms \citep{HackPhD,Moretti03,Wald94}. Since these are independent of the state $\omega$ they turn out to be irrelevant for the determination of the SLE. Therefore, they will not be written in the following.

The expectation value of the renormalized energy density will be calculated on a quasifree state with two-point function given by \eqref{GNS2ptfcn}. Let us now take a closer look at this two-point function. By performing the Bogolubov transformation \eqref{Bogtransf}, $w^{(2)}_{\omega}(x,x')$ becomes
\begin{align}
 w^{(2)}_{\omega}(x,x')=\sum_{j} &\left[(1+\beta_{j}^{2})\overline{S}_{j}(t)S_{j}(s)+\beta_{j}^{2}S_{j}(t)\overline{S}_{j}(s) \right. \nonumber \\
 &\left. +2\beta_{j}\sqrt{1+\beta_{j}^{2}}\textrm{Re}\left(e^{i\theta_{j}}S_{j}(t)S_{j}(s)\right)\right]\psi_{j}(\uline{x})\overline{\psi}_{j}(\uline{x'}) \nonumber \\
 &\eqqcolon \sum_{j}w^{(2)}_{\omega_{j}}(x,x') \; .
 \label{Bogstate}
\end{align}
where $\alpha_{j}=e^{i\theta_{j}}\sqrt{1+\beta_{j}^{2}}$ ($\beta_{j}$ was chosen to be real --- see remarks after equation \eqref{Bogcoef}). The last equality in \eqref{Bogstate} shows that the minimization amounts to finding the Bogolubov parameters $\beta_{j}$ and $\theta_{j}$ which minimize the contribution of each mode to the expectation value of the renormalized energy density. Since the last term in equation \eqref{pointsplit-general} is independent of the state $\omega$, and therefore independent of the Bogolubov parameters, this term becomes irrelevant for the present purposes. Regarding the convergence of the sum in equation \eqref{Bogstate}, see the remarks after equation \eqref{GNS2ptfcn}.

The definition of $w^{(2)}_{\omega_{j}}$ allows us to make a mode decomposition of the expectation value of the regularized energy density on the state $\omega$. We thus define
\begin{equation}
 \langle \hat{T}^{reg}\rangle_{\omega_{j}}(x,x')\coloneqq\left[\frac{1}{2}\nabla_{0}|_{x}\nabla_{0}|_{x'}+\frac{1}{2}\nabla^{c}|_{x}\nabla_{c}|_{x'}+\frac{1}{2}m^{2}\right]w^{(2)}_{\omega_{j}}(x,x') \; .
 \label{Treg-mode}
\end{equation}
Inserting \eqref{Bogstate} into \eqref{Treg-mode} and taking the coincidence limit $x'\rightarrow x$, we find
\begin{align}
 &\langle \hat{T}\rangle_{\omega_{j}}(t,\uline{x}) \coloneqq \lim_{x'\rightarrow x}\langle \hat{T}^{reg}\rangle_{\omega_{j}}(x,x')= \nonumber \\
 &\frac{1}{2}(1+2\beta_{j}^{2})\left\{\lvert\dot{S}_{j}(t)\rvert^{2}\lvert \psi_{j}(\uline{x})\rvert^{2}+\lvert S_{j}(t)\rvert^{2}\left(c(t)^{-2}h^{kl}(\uline{x})\nabla_{k}\psi_{j}(\uline{x})\nabla_{l}\overline{\psi}_{j}(\uline{x})+m^{2}\lvert \psi_{j}(\uline{x})\rvert^{2}\right)\right\} \nonumber \\
 &+\frac{1}{2}2\beta_{j}\sqrt{1+\beta_{j}^{2}}\textrm{Re}\left\{e^{i\theta_{j}}\left[\left(\dot{S}_{j}(t)\right)^{2}\lvert \psi_{j}(\uline{x})\rvert^{2} \right. \right. \nonumber \\
 &\left. \left. +\left(S_{j}(t)\right)^{2}\left(c(t)^{-2}h^{kl}(\uline{x})\nabla_{k}\psi_{j}(\uline{x})\nabla_{l}\overline{\psi}_{j}(\uline{x})+m^{2}\lvert \psi_{j}(\uline{x})\rvert^{2}\right)\right]\right\}\; .
 \label{EnergyGNS}
\end{align}

It is clear from equation \eqref{EnergyGNS} that, if $\langle \hat{T} \rangle_{\omega_{j}}$ is not smeared also in space, we will not be able to choose constant parameters $\beta_{j}$ and $\theta_{j}$ which minimize the energy density. The smeared energy density is now
\begin{equation}
 E_{j} \coloneqq \int_{\mathbb{R}}\textrm{d}t\, f^{2}(t)\int_{\Sigma}\textrm{d}^{3}x\sqrt{|h|}\, \langle \hat{T}\rangle_{\omega_{j}}(t,\uline{x}) \; .
\end{equation}
This should be interpreted as a heuristic formula, since this is just the coincidence limit of the expectation value of the regularized energy density, not the renormalized one. However, as stated above, this is the term which must be analyzed in order to construct the SLE.

Since the spatial hypersurfaces are compact without boundary, we calculate
\begin{align}
 \int_{\Sigma}\textrm{d}^{3}x\sqrt{|h|}\, h^{kl}(\uline{x})\nabla_{k}\psi_{j}(\uline{x})\nabla_{l}\overline{\psi}_{j}(\uline{x}) &=-\int_{\Sigma}\textrm{d}^{3}x\sqrt{|h|}\, \psi_{j}(\uline{x})\Delta_{h}\overline{\psi}_{j}(\uline{x}) \nonumber \\
 &=\lambda_{j}\int_{\Sigma}\textrm{d}^{3}x\sqrt{|h|}\, \lvert \psi_{j}(\uline{x})\rvert^{2}=\lambda_{j}\; .
 \label{laplcompsurfc}
\end{align}

Therefore,
\begin{align}
 E_{j} &=(1+2\beta_{j}^{2})\frac{1}{2}\int\textrm{d}t\, f^{2}(t)\left(|\dot{S}_{j}(t)|^{2}+\omega_{j}^{2}(t)|S_{j}(t)|^{2}\right) \nonumber \\
 &+2\beta_{j}\sqrt{1+\beta_{j}^{2}}\frac{1}{2}\textrm{Re}\left\{e^{i\theta_{j}}\int\textrm{d}t\, f^{2}(t)\left((\dot{S}_{j}(t))^{2}+\omega_{j}^{2}(t)S_{j}(t)^{2}\right)\right\} \nonumber \\
 &=(1+2\beta_{j}^{2})c_{1j}+2\beta_{j}\sqrt{1+\beta_{j}^{2}}\textrm{Re}(e^{i\theta_{j}}c_{2j})\; ,
 \label{EnergyGNSmin}
\end{align}
where
\begin{align}
&c_{1j}=\frac{1}{2}\int\textrm{d}t\, f^{2}(t)\left(|\dot{S}_{j}(t)|^{2}+\omega_{j}^{2}(t)|S_{j}(t)|^{2}\right) \label{c1j} \\
&c_{2j}=\frac{1}{2}\int\textrm{d}t\, f^{2}(t)\left((\dot{S}_{j}(t))^{2}+\omega_{j}^{2}(t)S_{j}(t)^{2}\right) \; . \label{c2j}
\end{align}
It is easy to see that, by choosing
\begin{equation}
 \beta_{j}=\sqrt{\frac{c_{1j}}{2\sqrt{c_{1j}^{2}-|c_{2j}|^{2}}}-\frac{1}{2}}\textrm{\quad and \quad}\alpha_{j}=e^{i\theta_{j}}\sqrt{\frac{c_{1j}}{2\sqrt{c_{1j}^{2}-|c_{2j}|^{2}}}+\frac{1}{2}}
 \label{Bogmin}
\end{equation}
and
\begin{equation}
\theta_{j}=-\textrm{Arg}c_{2j}+\pi\; ,
\label{thetamin}
\end{equation}
we minimize \eqref{EnergyGNSmin}. We will refer to these states of low energy as $\omega_{SLE}$, and their two-point functions will be referred to as $w^{(2)}_{\omega_{SLE}}$. The proof that these states are of the Hadamard form will be left for the next section.

We remark that the SLE are dependent on the test function used in the smearing. In spite of that, Degner \citep{Degner09} calculated the particle production process on such states in RW spacetimes and showed that the rate of production is not strongly dependent on the test function chosen. This dependence would only be dropped if the terms between parentheses in equations \eqref{c1j} and \eqref{c2j} could be taken out of the integrals. This would be the case if and only if $c(t)=$ constant, and in such a case we would have $c_{2j}\equiv 0$ and the SLE would reduce to the static vacuum.

The SLE constructed here are different from the ones constructed by Olbermann because here we smear the expectation value of the energy density over a spacelike slab of spacetime (containing entirely a Cauchy hypersurface and extended in time), while there the integration over a spatially extended region was not in general necessary. Besides, the treatment given here does not depend on the occurrence of spatial symmetries. We also note that if we had chosen an arbitrary causal observer, the energy density would contain terms of the form $\dot{\gamma}^{0}\dot{\gamma}^{l}\nabla_{0}|_{x}\nabla_{l}|_{x'}w^{(2)}_{\omega}(x,x')$, which could spoil the positivity of \eqref{c1j}, thus compromising the minimization of the energy density. Such a problem would also occur in the homogeneous, but anisotropic case.

\section{Fulfillment of the Hadamard condition by the SLE}\label{SLE-Hadamard}

We will now show that the SLE are Hadamard states. As stated in section \ref{secstates}, Hadamard states are completely characterized by the singularity structure of their two-point function, which means that the difference between the two-point functions corresponding to different Hadamard states must be a smooth function. Therefore we will compare the two-point function corresponding to the SLE to another one, corresponding to a given Hadamard state, and check that this difference is smooth. For this purpose, we will use the concept of adiabatic states, which are known to be, under certain conditions, Hadamard states \citep{JunSchrohe02}. Moreover, this will give us an explicit ansatz for $T_{j}(t)$.

The definition of adiabatic states given in \citep{JunSchrohe02} and the proof that these are Hadamard states involves a refinement of the Hadamard condition which uses the Sobolev wave front sets. On the following, we will first present the definition of Hadamard states and of adiabatic states in terms of the Sobolev wave front set. In the sequel, we will present the iteration procedure, formalized in \citep{LuRo90}, which provides the explicit ansatz for $T_{j}(t)$. After that we will use these as tools to show that the SLE constructed in the former section satisfy the Hadamard condition.

\subsection{Adiabatic States} \label{adiabaticstates}$ $

Recalling the definition of Sobolev wave front set (equation \eqref{Sobolev-space} and definition \eqref{Sobolev-wf}), Junker and Schrohe \citep{JunSchrohe02} proved that 
\begin{lem}\label{Hadamard-Sobolev}
 For every Hadamard state $\omega_{H}$ we have
\begin{equation}
 WF^{s}(w^{(2)}_{\omega_{H}})=\left\{\begin{array}{cc}
                           \emptyset \textrm{\ ,\ } &s< -1/2 \\
			   C^{+} \textrm{\ ,\ } &s\geqslant -1/2
                          \end{array}
\right. \; ,
\end{equation}
\end{lem}
where $C^{+}$ is the set of points which composes the smooth wave front set of a Hadamard state (see definition \eqref{Hadamard-wf} and equation \eqref{HadWF-C+}).

The adiabatic states are formulated iteratively (see below). For the $N$-th order of iteration, the adiabatic states $\omega_{N}$ are defined by the singularity structure of their two-point function:
\begin{mydef}\label{Adiabatic-def}
 A quasifree state $\omega_{N}$ on the $CCR$-algebra $\F$ is an {\it Adiabatic State} of order $N$ if, $\forall s<N+3/2$,
\begin{equation}
 WF^{s}(w^{(2)}_{\omega_{N}})=C^{+} \; . 
\end{equation}
\end{mydef}
When we compare the definition \eqref{smooth-wf} of smooth wave front set with the definition \eqref{Sobolev-wf} of Sobolev wave front set, we see that, while the former one only indicates the directions, in cotangent space, where the singularities of a distribution are located, the latter one indicates also the degree of this singularity.

Comparing the above definition with Lemma \eqref{Hadamard-Sobolev}, we have:
\begin{equation}
 WF^{s}(w^{(2)}_{\omega_{H}}-w^{(2)}_{\omega_{N}})=\emptyset \textrm{\quad ,\quad} \forall s<N+3/2 \; .
 \label{AdiabaticHadamard}
\end{equation}
Junker and Schrohe showed that the explicit construction given in \citep{LuRo90} satisfies the $WF^{s}$ condition. Furthermore, they defined adiabatic states on general globally hyperbolic spacetimes with compact Cauchy hypersurface, hence the definition is also valid on the expanding spacetimes considered here. We will present this construction now.

The adiabatic ansatz determines the initial conditions of the solutions to the field equation \eqref{timeKG}. A solution to this equation, $S_{j}(t)$, assumes, as initial values,
\begin{equation}
 S_{j}(t_{0})=W_{j}(t_{0}) \textrm{\qquad ;\qquad} \dot{S}_{j}(t_{0})=\dot{W}_{j}(t_{0}) \; .
\end{equation}

$W_{j}$ is of the WKB form:
\begin{equation}
 W_{j}(t)=\frac{1}{\sqrt{2\Omega_{j}(t)c(t)^{3}}}\exp\left(i\int_{t_{0}}^{t}\textrm{d}t'\, \Omega_{j}(t')\right) \; .
 \label{WKB}
\end{equation}
This form automatically satisfies the normalization, equation \eqref{timeSymplprod}. For it to satisfy the Klein-Gordon equation \eqref{timeKG}, $\Omega_{j}(t)$ must satisfy
\begin{equation}
(\Omega_{j})^{2}=\omega_{j}^{2}-\frac{3(\dot{c})^{2}}{4c^{2}}-\frac{3\ddot{c}}{2c}+\frac{3(\dot{\Omega}_{j})^{2}}{4(\Omega_{j})^{2}}-\frac{\ddot{\Omega}_{j}}{2\Omega_{j}} \; .
\end{equation}
A solution to this equation could be attempted iteratively:
\begin{align}
 \Omega_{j}^{(0)} &=\omega_{j} \label{WKBiteration_initial} \\
 (\Omega_{j}^{(N+1)})^{2} &=\omega_{j}^{2}-\frac{3(\dot{c})^{2}}{4c^{2}}-\frac{3\ddot{c}}{2c}+\frac{3(\dot{\Omega}_{j}^{(N)})^{2}}{4(\Omega_{j}^{(N)})^{2}}-\frac{\ddot{\Omega}_{j}^{(N)}}{2\Omega_{j}^{(N)}} \; .
 \label{WKBiteration}
\end{align}
Clearly, there could be values of $t$ and $j$ for which the right hand side of \eqref{WKBiteration} is negative and the iteration breaks down. L\"uders and Roberts \citep{LuRo90} showed that, within a certain interval of time $\mathscr{I}$ and for large values of $\lambda_{j}$, $\Omega_{j}^{(N)}(t)$ is strictly positive and
\begin{equation}
\Omega_{j}^{(N)}(t)=\mathcal{O}((1+\lambda_{j})^{1/2}) \; .
\label{Omega-adiabatic}
\end{equation}
All derivatives of $\Omega_{j}^{(N)}$ have the same asymptotic behavior.

At a generic instant of time,
\begin{equation}
 S_{j}(t)=\varsigma_{j}(t)W_{j}(t)+\xi_{j}(t)\overline{W}_{j}(t) \; .
 \label{Adansatz}
\end{equation}
The variables $\varsigma$ and $\xi$ are determined by the following coupled integral equations:
\begin{align}
\varsigma_{j}(t)&=1-i\int_{t_{1}}^{t}R_{j}(t')\left[\varsigma_{j}(t')+\xi_{j}(t')\exp\left(-2i\int_{t_{0}}^{t'}dt''\, \Omega_{j}(t'')\right)\right]dt' \label{Ad_varsigma} \\
\xi_{j}(t)&=i\int_{t_{1}}^{t}R_{j}(t')\left[\xi_{j}(t')+\varsigma_{j}(t')\exp\left(2i\int_{t_{0}}^{t'}dt''\, \Omega_{j}(t'')\right)\right]dt' \; , \label{Ad_xi}
\end{align}
where $R_{j}(t)$ is determined by
\begin{equation}
2R_{j}\Omega_{j}=(\Omega_{j})^{2}-\frac{3(\dot{\Omega}_{j})^{2}}{4(\Omega_{j})^{2}}+\frac{\ddot{\Omega}_{j}}{2\Omega_{j}}-\omega_{j}^{2}+\frac{3(\dot{c})^{2}}{4c^{2}}+\frac{3\ddot{c}}{2c} \; .
\end{equation}
Standard techniques of integral equations, which can be found in appendix A of \citep{LuRo90}, lead to the conclusion that there exist constants $C_{\xi} \, ,\, C_{\varsigma}>0$ such that
\begin{equation}
\lvert  \xi_{j}^{(N)}(t) \rvert \leqslant C_{\xi}(1+\lambda_{j})^{-N-1/2} \; \; , \; \; \lvert  1-\varsigma_{j}^{(N)}(t) \rvert \leqslant C_{\varsigma}(1+\lambda_{j})^{-N-1/2}
\label{WKBpar-j}
\end{equation}
(the same estimates above are valid for $\lvert \dot{\xi}_{j}^{(N)}(t) \rvert$ and $\lvert \dot{\varsigma}_{j}^{(N)}(t) \rvert$). Finally,
\begin{equation}
 \lvert W_{j}^{(N)}(t) \rvert =\mathcal{O}((1+\lambda_{j})^{-1/4}) \; \; \textrm{and} \; \; \lvert \dot{W}_{j}^{(N)}(t) \rvert =\mathcal{O}((1+\lambda_{j})^{1/4}) \; .
 \label{WKBsol-j}
\end{equation}

Now we will show that $WF^{s}(w^{(2)}_{\omega_{SLE}}-w^{(2)}_{\omega_{N}})=\emptyset$ and, by property \eqref{WFsum}, we will have $WF^{s}(w^{(2)}_{\omega_{SLE}}-w^{(2)}_{\omega_{H}})=\emptyset$.

\subsection{Fulfillment of conditions} \label{subsecfulfillhadamard}$ $

From \eqref{Adansatz}, \eqref{WKBpar-j} and \eqref{WKBsol-j},
\begin{equation}
 \partial_{t}^{k}S_{j}^{(N)}(t)=\mathcal{O}((1+\lambda_{j})^{k/2-1/4}) \; .
 \label{Ad-J}
\end{equation}

The two-point function corresponding to the SLE is given by \eqref{Bogstate}, where $c_{1j}$, $c_{2j}$ and $\beta_{j}$ are given by \eqref{c1j}, \eqref{c2j} and \eqref{Bogmin}, respectively. Since $c_{1j} > \lvert c_{2j} \rvert $,
\begin{equation}
 2\beta_{j}^{2} \approx \frac{1}{2}\frac{\lvert c_{2j} \rvert^{2}}{c_{1j}^{2}}+\frac{1}{4}\frac{\lvert c_{2j} \rvert^{4}}{c_{1j}^{4}}+\ldots 
\end{equation}

From \eqref{c1j} and \eqref{Ad-J}, it is immediate to see that
\begin{equation}
 c_{1j}=\mathcal{O}((1+\lambda_{j})^{1/2}) \; .
 \label{c1j-J}
\end{equation}
The analysis of the behavior of $\lvert c_{2j} \rvert $ is more involved. For this we need to estimate the scalar products of the WKB functions. The first one already appeared in equation \eqref{WKBsol-j}:
\begin{equation}
\left(W_{j}^{(N)},W_{j}^{(N)}\right)=\int_{I} dt\frac{1}{2c(t)\Omega_{j}^{(N)}(t)}=\mathcal{O}((1+\lambda_{j})^{-1/4}) \; .
\end{equation}
On the other hand, the scalar product
\begin{equation}
\left(\overline{W}_{j}^{(N)},W_{j}^{(N)}\right)=\int_{I}dt\frac{1}{2c(t)\Omega_{j}^{(N)}(t)}\exp{2i\int_{t_0}^t\Omega_{j}^{(N)}(t')dt'}
\end{equation}
is rapidly decaying in $\lambda_j$. This follows from the stationary phase approximation. It can be directly seen by exploiting the identity
\[\exp{2i\int_{t_0}^t\Omega_j^{(N)}(t')dt'}=\frac{1}{2i\Omega_j^{(N)}(t)}\frac{\partial}{\partial t}\exp{2i\int_{t_0}^t\Omega_j^{(N)}(t')dt'}\]
several times and subsequent partial integration. The estimates on $\Omega_j^{(N)}$ and its derivatives, together with the smoothness of $c(t)$, then imply the rapid decay. With these estimates at hand, we write (we will omit the index $^{(N)}$ to simplify the notation)
\begin{align}
S_{j}(t) &= |W_{j}|\left(\varsigma_{j} e^{i\int\Omega_{j}(t')dt'}+\xi_{j} e^{-i\int\Omega_{j}(t')dt'}\right) \; \therefore \nonumber \\
\left(S_{j}(t)\right)^{2} &=|W_{j}|^{2}\left(\varsigma_{j}^{2} e^{2i\int\Omega_{j}(t')dt'}+\xi_{j}^{2} e^{-2i\int\Omega_{j}(t')dt'}+2\varsigma_{j}\xi_{j}\right) \label{S^2-WKB}
\end{align}
and
\begin{align}
\left(\dot{S}_{j}(t)\right)^{2} &=|\dot{W}_{j}|^{2}\left(\varsigma_{j}^{2} e^{2i\int\Omega_{j}(t')dt'}+\xi_{j}^{2} e^{-i\int\Omega_{j}(t')dt'}+2\varsigma_{j}\xi_{j}\right) \nonumber \\
&+2|W_{j}||\dot{W}_{j}|\left(\varsigma_{j} e^{i\int\Omega_{j}(t')dt'}+\xi_{j} e^{-i\int\Omega_{j}(t')dt'}\right)\left[(\dot{\varsigma}+i\Omega_{j}\varsigma)e^{i\int\Omega_{j}(t')dt'}+(\dot{\xi}-i\Omega_{j}\xi)e^{-i\int\Omega_{j}(t')dt'}\right] \nonumber \\
&+|W_{j}|^{2}\left[(\dot{\varsigma}+i\Omega_{j}\varsigma)^{2}e^{2i\int\Omega_{j}(t')dt'}+(\dot{\xi}-i\Omega_{j}\xi)^{2}e^{-2i\int\Omega_{j}(t')dt'}+(\dot{\varsigma}+i\Omega_{j}\varsigma)(\dot{\xi}-i\Omega_{j}\xi)\right] \; . \label{dotS^{2}-WKB}
\end{align}
From \eqref{WKBpar-j} and \eqref{WKBsol-j}, we have
\begin{align}
\omega_{j}^{2}\left(S_{j}(t)^{(N)}\right)^{2} &=\mathcal{O}(\lambda_{j}^{-N}) \; ; \label{S^2_j} \\
\left(\dot{S}_{j}^{(N)}(t)\right)^{2} &=\mathcal{O}(\lambda_{j}^{-N}) \; . \label{dotS^2_j}
\end{align}

Using these results, we have
\begin{equation}
\lvert c_{2j}^{(N)} \rvert =\mathcal{O}(\lambda_{j}^{-N}) \; .
\label{c2j-J}
\end{equation}
Therefore,
\begin{equation}
 \beta_{j}^{(N)}=\mathcal{O}(\lambda_{j}^{-N-1/2}) \; .
 \label{Bog-J}
\end{equation}

Now, we need similar estimates for the eigenfunctions and eigenvalues of the laplacian. The asymptotic behavior of the eigenvalues is directly given by Weyl's estimate \citep{Jost11}:
\begin{equation}
\lambda_{j}=\mathcal{O}(j^{2/m}) \; ,
\label{lambda-J}
\end{equation}
where $m$ is the dimension of the Riemannian manifold.

For the estimate on $\psi_{j}$, we start by defining the spectral function of the Laplace operator as the kernel of the projection operator on the subspace of all its eigenfunctions whose corresponding eigenvalues are smaller than a certain value $\lambda$:
\begin{equation}
 e(\uline{x},\uline{y},\lambda) \coloneqq \sum_{\lambda_{j}\leqslant\lambda}\psi_{j}(\uline{x})\overline{\psi}_{j}(\uline{y}) \; .
 \label{spectralfunction}
\end{equation}
Elliptic regularity guarantees that this sum is bounded. H\"ormander \citep{Hormander-IV} proved that, for any differential operator $Q_{\uline{x},\uline{y}}$ of order $\mu$, the following inequality is valid:
\begin{equation}
 \lvert Q_{\uline{x},\uline{y}}(e(\uline{x},\uline{y},\lambda)) \rvert \leqslant C_{Q}\lambda^{m+\mu} \; ,
\end{equation}
where $m$ is the dimension of $\Sigma$. Combining this result with the Weyl's estimate and restricting to $m=3$, we obtain
\begin{equation}
 \lvert \partial^{\lvert k\rvert}\psi_{j}(\uline{x}) \rvert^{2} \leqslant C_{3,k}\lambda_{j}^{3+2\lvert k \rvert} \; \therefore \; \lvert \partial^{\lvert k\rvert}\psi_{j} \rvert = \mathcal{O}(j^{1+2\lvert k\rvert /3}) \; .
 \label{eigen-lapl-J}
\end{equation}

Now, we proceed to the proof that the SLE are Hadamard states. As stated at the beginning of this section, adiabatic states $\omega_{N}$ in spacetimes with metric \eqref{metric-timedecomp-expanding} and compact Cauchy hypersurface are Hadamard states. To show that the SLE are Hadamard, it suffices to show that
\[w^{(2)}_{\omega_{SLE}}-w^{(2)}_{\omega_{N}} \in H^{s}(\M \times \M) \; ,\]
for $s<N+3/2$. Moreover, since \citep{ReedSimon-II}
\[\forall s<k-\frac{1}{2}\dim(\M \times \M) \; , \; C^{k}(\M \times \M) \subset H^{s}(\M \times \M) \; ,\]
all that is needed is to show that $\exists \, k>0$ such that
\begin{equation}
 w^{(2)}_{\omega_{SLE}}-w^{(2)}_{\omega_{N}} \in C^{k}(\M \times \M) \; .
 \label{SLE-Hadamard-cond}
\end{equation}

The difference between the two-point functions is given by
\begin{equation}
 (w^{(2)}_{\omega_{SLE}}-w^{(2)}_{\omega_{N}})(t,\uline{x};t',\uline{x'})=\sum_{j}\left(\overline{T}_{j}(t)T_{j}(t')-\overline{S}_{j}^{(N)}(t)S_{j}^{(N)}(t')\right)\psi_{j}(\uline{x})\overline{\psi}_{j}(\uline{x'}) \; .
 \label{SLE-N}
\end{equation}
We will verify the convergence of this sum by estimating the asymptotic behavior of its derivatives:
\begin{equation}
 \partial_{x,x'}^{\lvert k\rvert}(w^{(2)}_{\omega_{SLE}}-w^{(2)}_{\omega_{N}})(x;x')=\sum_{j}\partial_{x,x'}^{\lvert k\rvert}\left[\left(\overline{T}_{j}(t)T_{j}(t')-\overline{S}_{j}^{(N)}(t)S_{j}^{(N)}(t')\right)\psi_{j}(\uline{x})\overline{\psi}_{j}(\uline{x'})\right] \; .
 \label{partial-SLE-N}
\end{equation}

Since $T_{j}(t)$ is obtained from the adiabatic ansatz, it should be viewed as $T_{j}^{(N)}(t)$. Performing the Bogolubov transformation to write the two-point function of the SLE as \eqref{Bogstate}, we find
\begin{align}
 \overline{T}_{j}^{(N)}(t)T_{j}^{(N)}(t')-\overline{S}_{j}^{(N)}(t)S_{j}^{(N)}(t') &=(\beta_{j}^{(N)})^{2}\left[S_{j}^{(N)}(t)\overline{S}_{j}^{(N)}(t')+S_{j}^{(N)}(t')\overline{S}_{j}^{(N)}(t)\right] \nonumber \\
 &+2\beta_{j}^{(N)}\sqrt{1+(\beta_{j}^{(N)})^{2}}\textrm{Re}\left[e^{i\theta_{j}}S_{j}^{(N)}(t)S_{j}^{(N)}(t')\right] \; .
 \label{time-SLE-N}
\end{align}
From \eqref{Ad-J} and \eqref{Bog-J}, the last term on the rhs of \eqref{time-SLE-N} has the largest order in $j$. For that reason, this is the only term which we will take into account in the verification of the convergence of the sum.

Rewriting the estimates \eqref{Ad-J} and \eqref{Bog-J} in terms of $j$, we have
\begin{align}
 \partial_{t}^{k}S_{j}^{(N)}(t) &=\mathcal{O}(j^{k/3-1/6}) \label{S-J} \; , \\
 \beta_{j}^{(N)} &=\mathcal{O}(j^{-2N/3-1/3}) \; . \label{beta-J}
\end{align}

It is then easy to see that the derivative of largest order in \eqref{partial-SLE-N} is $\partial_{\uline{x},\uline{x'}}^{\lvert k\rvert}$:
\begin{equation}
 \partial_{\uline{x},\uline{x'}}^{\lvert k\rvert}\left[\left(\overline{T}_{j}(t)T_{j}(t')-\overline{S}_{j}^{(N)}(t)S_{j}^{(N)}(t')\right)\psi_{j}(\uline{x})\overline{\psi}_{j}(\uline{x'})\right]=\mathcal{O}\left(j^{\frac{4\lvert k\rvert}{3}-\frac{2N}{3}+\frac{4}{3}}\right) \; .
\end{equation}

The sum in \eqref{partial-SLE-N} will be absolutely convergent if
\begin{equation}
 \frac{4\lvert k\rvert}{3}-\frac{2N}{3}+\frac{4}{3}<-1 \; \therefore \; \lvert k\rvert<\frac{N}{2}-\frac{7}{4} \; .
\end{equation}
This means that
\begin{equation}
 w^{(2)}_{\omega_{SLE}}-w^{(2)}_{\omega_{N}} \in C^{\lfloor \frac{N}{2}-\frac{7}{4} \rfloor}(\M \times \M) \; ,
\end{equation}
where
\begin{equation}
 \lfloor x \rfloor \coloneqq \left\{\begin{array}{cc}
                           \max\{m \in \mathbb{Z}|m \leqslant x\} \textrm{\ ,\ } & x>0 \\
			   0 \textrm{\ ,\ } &x \leqslant 0
                          \end{array}
\right. \; .
\end{equation}

Finally,
\begin{equation}
 WF^{s}(w^{(2)}_{\omega_{SLE}}-w^{(2)}_{\omega_{N}})=\emptyset \textrm{\quad for \quad}s<\frac{N}{2}-\frac{23}{4} \; .
 \label{WF-SLE-N}
\end{equation}
Since $N+\frac{3}{2}>\frac{N}{2}-\frac{23}{4}$, the above equality means that $\forall s>-1/2\, ,\, \exists N\in \mathbb{Z_{+}}$ such that the adiabatic states are Hadamard states and, at the same time, satisfy \eqref{WF-SLE-N}\footnote{The equality \eqref{WF-SLE-N} is valid $\forall s\in \mathbb{R}$, but for $s<-1/2$ the Sobolev wave front set of a Hadamard state is itself empty.} . Therefore,
\begin{equation}
 WF(w^{(2)}_{\omega_{SLE}}-w^{(2)}_{\omega_{H}})=\emptyset \; .
 \label{SLE-Hadamard-fulfilled}
\end{equation}

This proves that the States of Low Energy constructed on globally hyperbolic spacetimes with metric of the form \eqref{metric-timedecomp-expanding} and compact Cauchy hypersurface are Hadamard states. We remark that this proof is also valid for the SLE constructed on homogeneous spacetimes above.

\chapter{``Vacuum-like'' Hadamard states}\label{chap_vac-like}

In this chapter we will provide another example of Hadamard states, the {\it ``Vacuum-like'' states}. These were introduced by Afshordi, Aslanbeigi and Sorkin \citep{AfshordiAslanbeigiSorkin12}, after previous works of Johnston \citep{Johnston09} and Sorkin \citep{Sorkin11}, under the name ``Sorkin-Johnston states'' (S-J states), in an attempt to find a prescription to construct states which would be valid in any globally hyperbolic spacetime and would be singled out by the field dynamics. They also aimed at applications in cosmological problems. As proved by Fewster and Verch, such a unique general prescription is not possible \citep{FewsterVerch12,FewsterVerch_sf12}. The same authors proved that the construction given in \citep{AfshordiAslanbeigiSorkin12} does not lead to a Hadamard state in ultrastatic spacetimes \citep{FewsterVerch-SJ12}.

We will consider a relatively compact globally hyperbolic spacetime $\M$ isometrically embedded in another, larger globally hyperbolic spacetime $\N$. We will construct quasifree Hadamard states such that the kernel of their two-point functions, in the interior of $\M$, coincides with the positive spectral part of the advanced-minus-retarded operator of $\N$ restricted to $\M$ (such a proposal was put forward, for the case of the Dirac field, by Finster \citep{Finster11}). The difference to the construction of the S-J states \citep{AfshordiAslanbeigiSorkin12} is that, instead of assuming that the kernel of the two-point function goes abruptly to zero at the border of (the closure of) $\M$, we consider that the kernel consists of the positive spectral part of the advanced-minus-retarded operator of $\N$ multiplied by a smooth compactly supported function which is identically equal to 1 in the interior of $\M$. In the particular cases of static and expanding spacetimes, we will show that this is sufficient for the arising states to be Hadamard states. The states we construct here will be called ``modified S-J states''.

In subsection \ref{S-J_states} we will present the construction of the modified S-J states. In subsection \ref{subsec_static-Had} we will prove that, in static spacetimes, the states constructed are Hadamard states. In subsection \ref{subsec_expand-Had} we will prove the Hadamard property in expanding spacetimes.

\section{``Vacuum-like'' Hadamard states}\label{S-J_states}

The construction of the modified S-J states starts from the observation that the advanced-minus-retarded operator, operating on square-integrable functions in a globally hyperbolic spacetime $\M$, isometrically embedded in another globally hyperbolic manifold $\N$, with relatively compact image, is a bounded operator. We give a sketch of the proof in the following\footnote{See also the appendix of \citep{FewsterVerch-SJ12}.}
\begin{thm}
Let $\M$ and $\N$ be globally hyperbolic spacetimes such that there exists an isometry $\Psi :\M \rightarrow \N$ which embeds $\M$ isometrically in $\N$. Let also $\Psi(\M)$ be a causally convex and relatively compact subset of $\N$. Then the advanced-minus-retarded operator of $\M$, $\mathds{E}_{\M}$, is a bounded operator.
\end{thm}
\begin{proof}
Consider $\M=\mathds{I}\times\Sigma$, where $\mathds{I}=(-t_{0},t_{0})$ is a limited interval of time, and $\N=\mathbb{R}\times\Sigma$. The energy estimate found in Appendix III of \citep{Choquet09} gives, for any test function satisfying equation \eqref{KGfund} in an open, relatively compact subset $\mathcal{O}\in \mathfrak{I}\times\Sigma$, where $\mathfrak{I}\in\mathbb{R}$,
\begin{equation}
 \lVert \mathds{E}^{\pm}_{\N}f \rVert_{L^{2}(\mathfrak{I}\times\Sigma)}\leqslant C_{\mathfrak{I},\mathcal{O}}\lVert f \rVert_{L^{2}(\mathfrak{I}\times\Sigma)}
\end{equation}
where $C_{\mathfrak{I},\mathcal{O}}$ is a constant which only depends on $\mathfrak{I}$ and $\mathcal{O}$, and the norms are defined with respect to the volume measure on $\N$.

From the uniqueness of the advanced and retarded fundamental solutions, we have
\begin{equation}
\mathds{E}_{\M}f=\Psi^{*}\mathds{E}_{\N}\Psi_{*}f \; .
\end{equation}
for every $f\in C_{0}^{\infty}(\M,\mathbb{R})$, where $\Psi^{*}$, $\Psi_{*}$ are, respectively, the pull-back and push-forward associated to $\Psi$. Choosing $\mathfrak{I}$ such that $\Psi(\M)\in\mathfrak{I}\times\Sigma$,
\begin{equation}
 \lVert \mathds{E}^{\pm}_{\M}f \rVert_{L^{2}(\M)}\leqslant \lVert \mathds{E}^{\pm}_{\N}\Psi_{*}f \rVert_{L^{2}(\mathfrak{I}\times\Sigma)} \leqslant C_{\mathfrak{I},\Psi(\M)}\lVert \Psi_{*}f \rVert_{L^{2}(\mathfrak{I}\times\Sigma)} = C_{\mathfrak{I},\Psi(\M)}\lVert f \rVert_{L^{2}(\M)} \; .
 \label{norm_E_Mf-f}
\end{equation}
$L^{2}(\M)$ is a Banach space with norm
\begin{equation}
 \lVert \mathds{E}^{\pm}_{\M} \rVert \coloneqq \sup_{f \neq 0}\frac{\lVert \mathds{E}^{\pm}_{\M}f \rVert_{L^{2}(\M)}}{\lVert f \rVert_{L^{2}(\M)}} \; .
\end{equation}
Using \eqref{norm_E_Mf-f}, we have, finally,
\begin{equation}
 \lVert \mathds{E}_{\M} \rVert \leqslant 2C_{\mathfrak{I},\Psi(\M)} \; 
 \label{norm_E_M}
\end{equation}
thus concluding that $\mathds{E}_{\M}$ is a bounded operator on $L^{2}(\M)$.
\end{proof}

Let $f\in \mathcal{C}_{0}^{\infty}(\N,\mathbb{R})$ be a real-valued test function such that $f\equiv 1 $ on $\Psi(\M)$. We define the bounded self-adjoint operator $A$
\begin{equation}
 {A} \coloneqq if\mathds{E}_{\N} f \; ,
 \label{A-iE}
\end{equation}
where $f$ acts by multiplication on $L^2(\N)$. This is the point where our construction departs from the one done by \citep{AfshordiAslanbeigiSorkin12}. If we replace $f$ by the characteristic function of $\M$ we obtain the operator introduced by them and carefully analyzed in \citep{FewsterVerch-SJ12}.

Since $A$ is a bounded self-adjoint operator on the Hilbert space $L^2(\N)$ (with volume measure), then there exists a unique {\it projection valued measure} which allows us to write
\begin{equation}
 A=\int\lambda \, dP_{\lambda} \; ,
\end{equation}
where $dP_{\lambda}$ is called the {\it spectral measure}. The numbers $\lambda$ are the eigenvalues of $A$, and $\lambda\in[-\lVert A\rVert ; \lVert A\rVert]$ \citep{ReedSimon-I}. The {\it positive part} of $A$ is defined as
\begin{equation}
 A^{+}=\int_{[0,\lVert A\rVert]}\lambda \, dP_{\lambda} \; .
 \label{pos_A}
\end{equation}

The modified S-J state $\omega_{SJ_{f}}$ is now defined as the quasifree state on the spacetime $\M$ whose two-point function is given by
\begin{equation}
W_{SJ_{f}}(q,r) \coloneqq (  q , A^{+}r  ) \; ,
\label{W-SJ}
\end{equation}
for real-valued test functions $q,r$ on $\M$. Considering the adjoint of $A$, $A^{\ast}$, we define $|A|\coloneqq(A^{\ast}A)^{1/2}$. We can then write the positive part of $A$ as
\begin{equation}
A^{+}=\frac{|A|+A}{2} \; .
\end{equation}
The operator $|A|$ is a symmetric operator and it gives rise to the symmetric product
\[\mu_{SJ_{f}}(q,r)=(q,|A|r) \; .\]
Note that the antisymmetric part of  the two-point function coincides with $i\mathds{E}_\M$. This is due to the fact that the intersection of the kernel of $A$ with $L^2(\M)$ coincides with the kernel of $\mathds{E}_{\M}$. In particular, the integral kernel of $A^{+}$, restricted to $\M$, is a bisolution of the Klein-Gordon equation. This bisolution can be uniquely extended to the domain of dependence of $\M$ (which coincides with $\N$ in the case considered here). The state $\omega_{SJ_{f}}$ is a pure state, as can be seen in the following

\begin{thm}
Let $\M$ be a globally hyperbolic subspacetime of another globally hyperbolic spacetime $\N$, and let $\Sigma\subset\M$ be a Cauchy hypersurface of $\N$. Then for every real-valued $f\in\mathcal{C}^{\infty}_0(\N,\mathbb{R})$ with $f\equiv1$ on $\M$, the modified S-J state
\[\omega_{SJ_f}(W(\phi))=e^{-\frac{1}{2}(\phi,|f\mathds{E}_\N f|\phi)}\]
with $\phi\in\mathcal{C}^{\infty}_{0}(\M,\mathbb{R})$, is pure. Here $\mathds{E}_\N$ is the commutator function on $\N$ and $|\cdot|$ denotes the modulus of the operator.
\end{thm}
\begin{proof}
We consider the Weyl algebra over the symplectic space $(L,\sigma)$ with, now, 
\[L=\mathrm{Re}\left(\mathcal{C}^{\infty}_{0}(\N,\mathbb{R})/\mathrm{Ker}\mathds{E}_\N f\right)\]
and
\[\sigma([\phi_1],[\phi_2])=(\phi_1,f\mathds{E}_\N f\phi_2) \; .\]
Due to the compactness of the support of $f$, the operator $f\mathds{E}_\N f$ is bounded on this Hilbert space, and, according to the results of Manuceau and Verbeure \citep{ManuceauVerbeure68} mentioned in the Introduction, we can define a pure state on the Weyl algebra by setting 
\[\omega(W(\phi))=e^{-\frac12(\phi,|f\mathds{E}_\N f|\phi)}\]
where $|f\mathds{E}_\N f|=\sqrt{-f\mathds{E}_\N f^2\mathds{E}_\N f}$.

It remains to prove that the Weyl algebra above coincides with the Weyl algebra over $\M$ with the symplectic form defined by the commutator function $\mathds{E}_\M$ on $\M$. For this purpose we prove that the corresponding symplectic spaces are equal. Since the restriction of $\mathds{E}_\N$ to $\M$ coincides with $\mathds{E}_\M$ and since $f\equiv1$ on $\M$, the symplectic space associated to $\M$ is a symplectic subspace of $(L,\sigma)$. We now show that this subspace is actually equal to $(L,\sigma)$. This amounts to prove that every rest class $[\phi]\in L$ with $\phi\in\mathcal{C}^{\infty}_0(\N,\mathbb{R})$ contains an element $\phi_0$ with $\mathrm{supp}\phi_0\subset \M$.

Here we proceed similarly to Fulling, Sweeny and Wald \citep{FullingSweenyWald78}. Let $\phi\in\mathcal{C}^{\infty}_0(\N,\mathbb{R})$. We may decompose $\phi=\phi_{+}+\psi+\phi_{-}$ with $\mathrm{supp}\phi_\pm\subset J_\pm(\Sigma)$ and $\mathrm{supp}\psi\subset\M$. Let $\chi\in\mathcal{C}^{\infty}(\N,\mathbb{R})$ such that $\chi\equiv1$ on $J_+(\Sigma)$ and $\mathrm{supp}\chi\subset J_+(\Sigma_-)$ for a Cauchy hypersurface $\Sigma_-$ of $\M$ in the past of $\Sigma$. Set
\[\psi_+=P(1-\chi)\mathds{E}^-_\N f\phi_+\]
where $P$ is the Klein Gordon operator and $\mathds{E}_\N^-$ the advanced propagator. By the required properties of $\chi$, $\psi_{+}$ vanishes where $\chi$ is constant, hence $\mathrm{supp}\psi_+\subset \M$. In particular $f\psi_+=\psi_+$. We are left with showing that $\phi_+-\psi_+\in\mathrm{Ker}\mathds{E}_\N f$,
\[\mathds{E}_\N f(\phi_+-\psi_+)=\mathds{E}_\N(f\phi_+-\psi_+)=\mathds{E}_\N P\chi \mathds{E}^-_\N f\phi_+=0\ ,\]
where in the last step we used the fact that $\chi \mathds{E}^-_N f\phi_+$ has compact support. For $\phi_-$ an analogous argument works and yields an element $\psi_-\in[\phi_-]$ with $\mathrm{supp}\psi_-\subset \M$. Thus we find that $\phi_0=\psi_++\psi+\psi_-$ has the properties required above.
\end{proof}

The question now arises whether the modified S-J states are Hadamard states. We will prove this to be true in two situations, static spacetimes and expanding spacetimes. We remark that the proofs rely only upon the fact that $f\in \mathcal{C}_{0}^{\infty}(\N,\mathbb{R})$ is a real-valued test function such that $f_{\upharpoonright \M}\equiv 1 $. 
If we change $f$ we will, in general, obtain a different Hadamard state. Thus the states we construct here are not uniquely singled out by the spacetime geometry.

In the spacetimes considered below, the operator $\mathds{E}$ will be decomposed into a sum over the eigenprojections $\psi_{j}\overline{\psi}_{j}$ 
of the spatial part of the Klein-Gordon operator (since it possesses, in both cases, a complete base of orthonormalized eigenfunctions; see section \ref{subsec-static_expand}). We will choose our smearing function $f$ to depend only on time. Then, what we will have to do is to analyze, for each $j$, the operators $A_{j}$ defined as
\begin{equation}
A(t,\uline{x};t',\uline{x'}) \eqqcolon \sum_{j}A_{j}(t',t)\psi_{j}(\uline{x})\overline{\psi}_{j}(\uline{x'})\; .
\label{A_j}
\end{equation}

\subsection{Static spacetimes}\label{subsec_static-Had}

Taken as an operator on $L^{2}(\mathbb{R})$, $A_{j}$ has the integral kernel (see \eqref{E-static})
\begin{equation}
 A_{j}(t',t)=\frac{i}{\omega_{j}}f(t')\left(\sin (\omega_{j}t'-\theta_j)\cos (\omega_{j}t-\theta_j)-\cos (\omega_{j}t'-\theta_j)\sin (\omega_{j}t-\theta_j)\right)f(t).
\end{equation}
This expression does not depend on the phase $\theta_j$ due to the addition theorem of trigonometric functions.
We choose 
$\theta_j$ such that
\[\int dtf(t)^2\cos(\omega_jt-\theta_j)\sin(\omega_jt-\theta_j)=0\ .\]
Such a choice is possible since the integrand changes its sign if $\theta_j$ is shifted by $\pi/2$.

Since $A_{j}^{*}(t',t)\equiv \overline{A_{j}(t,t')}$, we find
\begin{equation}
|A_j|(t',t)=\frac{1}{\omega_j^2}\left(||S_j||^2C_j(t)C_j(t')+||C_j||^2S_j(t)S_j(t')\right)
\end{equation}  
with
\[S_j(t)=f(t)\sin(\omega_jt-\theta_j)\ ,\ C_j(t)=f(t)\cos(\omega_jt-\theta_j)\ ,\]
\[||S_j||^2 \coloneqq \int dtS_{j}(t)^{2}\; ,\]
and similarly for $||C_j||^2$. Hence the positive part of $A_{j}$ has the integral kernel
\[A_{j}^{+}(t',t)=\frac{1}{2\omega_j||C_j||||S_j||}\left(||S_j||C_j(t)-i||C_j||S_j(t)\right)\left(||S_j||C_j(t')+i||C_j||S_j(t')\right)\ .\]

Setting
\begin{equation}
 \delta_{j}\coloneqq 1-\frac{||C_{j}||}{||S_{j}||}\; ,
 \label{delta}
\end{equation}
we write
\begin{equation}
 A_{j}^{+}(t,t')=
 \frac{1}{2\omega_j}\left(\frac{1}{1-\delta_j}C_j(t)-iS_j(t)\right)\left(C_j(t')+i(1-\delta_j)S_j(t')\right)\; .
\end{equation}
Therefore, the two-point function on $\M$ is
\begin{align}
 W_{SJ_{f}}(t,\uline{x};t',\uline{x}')=\sum_{j}\frac{1}{2\omega_{j}}&\left(\frac{1}{1-\delta_{j}}C_j(t)-iS_j(t)\right)\left(C_j(t')+i(1-\delta_j)S_j(t')\right)\psi_j(\uline{x}) \overline{\psi}_j(\uline{x}') \; .
 \label{W-SJ-static}
\end{align}

A practical way to verify that this state is a Hadamard state is to compare it with another Hadamard state and check whether the difference $w$ of the two-point functions is smooth. For this comparison, we use the two-point function of the static ground state, restricted to $\M$:
\begin{equation}
 W_{0}(t,\uline{x};t',\uline{x}')=\sum_{j}\, \frac{e^{-i\omega_{j}(t-t')}}{2\omega_{j}}\psi_{j}(\uline{x})\overline{\psi}_{j}(\uline{x}') \; .
 \label{W-H-static}
\end{equation}
These two-point functions would coincide if $\delta_j=0$. Further we note that multiplying the latter by $f(t)f(t')$ gives the same function, since $f_{\upharpoonright \M}\equiv 1$.

We state our result as a theorem:

\begin{thm}\label{Had-static}
Let $\N=\mathbb{R}\times \Sigma$ be a static spacetime with metric $g=a^{2}dt^{2}-h$, where $h$ is a Riemannian metric on the compact manifold $\Sigma$ and $a$ is a smooth everywhere positive function on $\Sigma$. Let $I$ be a finite interval and $f$ a smooth real-valued function on $\mathbb{R}$ with compact support which is identical to 1 in $I$. Then the modified S-J state $\omega_{SJ_f}$ as constructed above on $\M=I\times\Sigma$ is a Hadamard state.  
\end{thm}

\begin{proof}
The difference $\colon W_{SJ_{f}}\colon$ between $ W_{SJ_{f}}$ and $W_{0}$ is
\begin{equation}
 \colon W_{SJ_{f}}\colon (t,\uline{x};t',\uline{x}')=\sum_{j}\, \frac{\delta_{j}}{2\omega_{j}}\left[\frac{1}{1-\delta_{j}}C_{j}(t')C_{j}(t)-S_{j}(t')S_{j}(t)\right]\psi_{j}(\uline{x})\overline{\psi}_{j}(\uline{x}') \; .
 \label{W-SJH-static}
\end{equation}
To prove that $\omega_{SJ_{f}}$ is a Hadamard state it suffices to show that $\colon W_{SJ_{f}}\colon$ is smooth.
Since the eigenfunctions $\psi_j$ of the elliptic operator $K$ are smooth, each term in the expansion above is smooth, and it suffices to prove that the sum converges in the sense of smooth functions. This can be done by proving that, for all derivatives, the sum converges in $L^2(\M\times\M)$.

For this purpose we first exploit the fact that the $L^2$-norms of derivatives of functions on $\Sigma$ can be estimated in terms of the operator $K$. Namely, for every differential operator $D$ of order $n$ on $\Sigma$ there exists a constant $c_D>0$ such that 
\[||D\psi||_2\leq c_D||K^{m}\psi||_2\]  
with $m$ the smallest integer larger than or equal to $n/2$ \citep{Hormander-III}. Hence spatial derivatives of the functions $\psi_j$ can be absorbed by multiplication with the corresponding eigenvalues of $K$. Similarly, time derivatives amount to multiplication with factors $\omega_j$ and exchanges between the functions $S_j$ and $C_j$. Since their $L^2$-norms are uniformly bounded in $j$, it remains to show that
\[\sum_j\omega_j^n\delta_j<\infty\ \forall\ n\in\mathbb{N}_0 \ .\]
We first observe that $||C_{j}||^{2}$ and $||S_{j}||^{2}$ can be expressed in terms of the Fourier transform of the square of the test function $f$:
\begin{equation}
||C_{j}||^{2}=\int\textrm{d}t\, f(t)^{2}\left(\frac{e^{2i(\omega_{j}t-\theta)}+e^{-2i(\omega_{j}t-\theta)}+2}{4}\right)=\frac{1}{2}+\frac{\widetilde{f^{2}}(2\omega_{j})e^{-2i\theta}+\widetilde{f^{2}}(-2\omega_{j})e^{2i\theta}}{4}
\label{Norm_C}
\end{equation}
and
\begin{equation}
||S_{j}||^{2}=\frac{1}{2}-\frac{\widetilde{f^{2}}(2\omega_{j})e^{-2i\theta}+\widetilde{f^{2}}(-2\omega_{j})e^{2i\theta}}{4} \; .
\label{Norm_S}
\end{equation}
Since $f$ is a smooth test function, so is $f^{2}$, and $\forall n\in\mathds{R}$,
\begin{equation}
\lim_{\omega\rightarrow \infty}\omega^{n}\widetilde{f^{2}}(2\omega)=0 \; .
\label{omega-f-smooth}
\end{equation}
It follows immediately that
\begin{equation}
\lim_{j\rightarrow \infty}\omega_{j}^{n}\delta_{j}=0 \; .
\label{omega-delta-smooth}
\end{equation}
 
The last information we need concerns the behavior of the eigenvalues of $K$. In order to analyse this behavior, we need a couple of definitions: for a positive operator $B$ with eigenvalues $\mu_{j}$, the {\it p-Schatten norm} \citep{Simon05} is defined as
\[|B|_{p}\coloneqq\Big(\sum_{j}\mu_{j}^{p}\Big)^{1/p} \; .\]
If this norm is finite, it is said that the operator $B$ is in the Schatten class $L_{p}$. Taking $B=(K+\mathds{1})^{-1}$ the {\it resolvent of $K$} and $\mu_{j}=(\lambda_{j}+1)^{-1}$, since $K$ is a self-adjoint elliptic positive operator of order 2 on a $d$-dimensional compact space, its resolvent is in the Schatten classes $L_{d/2+\epsilon}$, for $\epsilon>0$ \citep{Shubin01}. Hence, for $p>d/2+\epsilon$ ($p\in\mathbb{N}$), $\sum_j\omega_j^{-p}<\infty$ and we finally obtain the estimate
\[\sum_j\omega_j^n\delta_j\le (\sum_j\omega_j^{-p})(\sup_k\omega_k^{n+p}\delta_k)\le \infty\ .\]
\end{proof}

Before we proceed to the case of expanding spacetimes, we remark that the smoothness of the function $f$ was crucial for getting a Hadamard state. The state depends via the expansion coefficients $\delta_j$ and the phases $\theta_j$ on the values of the Fourier transform of $f^2$ at the points $2\omega_j$, and it is the fast decrease of these values as $j$ tends to infinity that implies the Hadamard property. Hence, if $f\notin \mathcal{C}_{0}^{\infty}(\mathbb{R})$ then, in general,  \eqref{omega-f-smooth} and \eqref{omega-delta-smooth} would not be satisfied, and the state would not be a Hadamard state.

\subsection{Expanding spacetimes}\label{subsec_expand-Had}

The advanced-minus-retarded-operator is now
\begin{equation}
 \mathds{E}(t,\uline{x};t',\uline{x'})=\sum_{j}\frac{(\overline{T}_{j}(t)T_{j}(t')-T_{j}(t)\overline{T}_{j}(t'))}{2i}\psi_{j}(\uline{x})\overline{\psi}_{j}(\uline{x'}) \; .
 \label{E-prop-expanding-st}
\end{equation}
We decompose $fT_j$ into its real and imaginary parts, $fT_j=B_{j}-iD_j$, and obtain for the integral kernel of the operator $A_j$
\begin{equation}
 A_j(t',t)=i\left(D_{j}(t')B_{j}(t)-B_{j}(t')D_{j}(t)\right) \; .
\end{equation}
$A_j$ is a self-adjoint antisymmetric rank 2 operator.

We can choose the phase of $T_{j}$ such that
\begin{equation}
 \int B_{j}(t)D_{j}(t)dt \equiv 0 \; .
 \label{B-D-0}
\end{equation}

Analogous to the static case we obtain
\begin{equation}
 A_{j}^{+}(t',t)=
 \frac{1}{2||B_{j}||||D_{j}||}\left(||D_{j}||B_{j}(t')-i||B_{j}||D_{j}(t')\right)\left(||D_{j}||B_{j}(t)+i||B_{j}||D_{j}(t)\right) \; .
 \label{A-pos}
\end{equation}
Setting again
\begin{equation}
\delta_j=1-\frac{||B_{j}||}{||D_{j}||} \; ,
\end{equation}
we find for the two-point function of the modified S-J state on $\M$
\begin{equation}
 W_{SJ_{f}}(t,\uline{x};t',\uline{x}')=\sum_j \frac{1}{2}\left(\frac{1}{1-\delta_j}B_{j}(t')-iD_{j}(t')\right)\left(B_{j}(t)+i(1-\delta_j)D_{j}(t)\right)\psi_{j}(\uline{x})\overline{\psi}_{j}(\uline{x}') \; .
 \label{W-SJ-expand}
\end{equation}

We now investigate the wave front set of this two-point function. We proceed as in the proof of the Hadamard condition for states of low energy (see section \ref{SLE-Hadamard}) by comparing \eqref{W-SJ-expand} with the two-point functions of adiabatic states of finite order. It suffices to prove that for all adiabatic orders $N$ the two-point functions \eqref{W-SJ-expand} and the one corresponding to an adiabatic state differ only by a function which is in the local Sobolev space of order $s$ satisfying $s<N+3/2$ (see the definition of adiabatic states in section \ref{adiabaticstates}).


As in the last chapter, the adiabatic ansatz will provide initial conditions for the solution $T_{j}$ of the Klein-Gordon equation:  
\begin{equation}
 T_{j}^{(N)}(t)=\left( \varsigma_{j}^{(N)}(t)W_{j}^{(N)}(t)+\xi_{j}^{(N)}(t)\overline{W}_{j}^{(N)}(t) \right)e^{i\theta_{j}} \; ,
\end{equation}
where $\theta_{j}$ is the phase factor introduced so that \eqref{B-D-0} is satisfied, $W_{j}^{(N)}$ is given by
\[W_{j}^{(N)}(t)=\frac{1}{\sqrt{2\Omega_{j}^{(N)}c(t)^{3}}}\exp\left(i\int_{t_{0}}^{t}dt'\Omega_{j}^{(N)}(t')\right) \; ,\]
where $\Omega_j^{(N)}$ is bounded from below by a constant times $\sqrt{\lambda_j}$, and together with its derivatives, bounded from above by constants times $\sqrt{\lambda_j}$. The functions $\varsigma_j^{(N)}$ and $\xi_j^{(N)}$ satisfy the estimates given in \eqref{WKBpar-j}. Thus, as in \eqref{WKBsol-j} and \eqref{Ad-J},
\begin{equation}
 \lvert W_{j}^{(N)}(t) \rvert =\mathcal{O}((1+\lambda_{j})^{-1/4}) \; \; \textrm{and} \; \; \lvert T_{j}^{(N)}(t) \rvert =\mathcal{O}((1+\lambda_{j})^{-1/4}) \; .
  \label{adiabatic-lambda}
\end{equation}

The proof that the two-point function \eqref{W-SJ-expand} has the Hadamard property will be presented in the theorem below. For clarity of the argument, we will repeat some statements and results which were presented in the proof that the SLE are Hadamard states (section \ref{subsecfulfillhadamard}).

\begin{thm}
Let $\N=J\times \Sigma$ be an expanding spacetime with $\Sigma$ compact and $J$ an open interval on the real axis. Let $I$ be a finite open interval with closure contained in $J$, and let $f\in\mathcal{C}^{\infty}_{0}(J)$ such that $f$ is equal to 1 on $I$. Then the modified Sorkin-Johnston state $\omega_{SJ_{f}}$, as defined above, is a Hadamard state on the expanding spacetime $\M=I\times\Sigma$.
\end{thm}

\begin{proof}
We want to show that for each $s>0$ there is an $N\in\mathbb{N}$ such that the difference of the two-point functions of the state $\omega_{SJ_f}$ and the adiabatic state of $N$-th order is an element of the Sobolev space of order $s$. As in the static case, we use the fact that spatial derivatives can be estimated in terms of the elliptic operator and amount to multiplication with powers of the corresponding eigenvalues $\lambda_j$. For the time derivatives we exploit the fact that the functions $T_j$ are solutions of a second order differential equation, which again will allow us to replace derivatives by multiplication with powers of $\lambda_j$.

Therefore, In order to verify the Hadamard property of $W_{SJ_{f}}$, we investigate for which index $s\in\mathbb{R}$ the operator
\begin{equation}
R_s=\sum_j \lambda_j^{s}(A^+_j-\frac{1}{2}|fT_j\rangle\langle fT_j|)\otimes |\psi_{j}\rangle\langle \psi_{j}|
\end{equation}  
is Hilbert-Schmidt. For this purpose we have to estimate the $L^{2}$ scalar products of the WKB functions. We have
\begin{equation}
\left(fW_j^{(N)},fW_j^{(N)}\right)_{L^{2}}=\int dtf(t)^2\frac{1}{2c(t)\Omega_j^{(N)}(t)}
\label{fW.fW}
\end{equation} 
which can be bounded from above and from below by a constant times $(1+\lambda_j)^{-\frac12}$. On the other hand, the scalar product
\begin{equation}
\left(\overline{fW}_j^{(N)},fW_j^{(N)}\right)_{L^{2}}=\int dtf(t)^{2}\frac{1}{2c(t)\Omega_j^{(N)}(t)}\exp{2i\int_{t_0}^t\Omega_j^{(N)}(t')dt'}
\label{fW*.fW}
\end{equation}
is rapidly decaying in $\lambda_j$. This follows from the stationary phase approximation. It can be directly seen by exploiting the identity
\[\exp{2i\int_{t_0}^t\Omega_j^{(N)}(t')dt'}=\frac{1}{2i\Omega_j^{(N)}(t)}\frac{\partial}{\partial t}\exp{2i\int_{t_0}^t\Omega_j^{(N)}(t')dt'}\]
several times and subsequent partial integration. The estimates on $\Omega_j^{(N)}$ and its derivatives together with the smoothness of $c(t)$ and $f(t)$ then imply the claim.

Now, the term $(A^+_j-\frac{1}{2}|fT_j\rangle\langle fT_j|)$ reads
\begin{align}
 A_{j}^{+}(t',t)-\frac{f(t')T_{j}(t')f(t)\overline{T}_{j}(t)}{2} &=\frac{1}{8(1-\delta_{j})}f(t')\left\{(\delta_{j})^{2}\left(\overline{T}_{j}(t')T_{j}(t)+T_{j}(t')\overline{T}_{j}(t)\right) \right. \nonumber \\
 &\left. +2\textrm{Re}\left[\delta_{j}(2-\delta_{j})T_{j}(t')T_{j}(t)\right]\right\}f(t) \; .
 \label{A-fT-expand}
\end{align}

On the static case, the terms $\lVert B_{j} \rVert$ and $\lVert D_{j} \rVert$ were written as combinations of the Fourier transform of a smooth function (see equations \eqref{Norm_C} and \eqref{Norm_S}). This is no longer valid in the expanding case. We now have
\begin{align*}
||B_{j}^{(N)}||^{2} &=\int\textrm{d}t\, B_{j}^{(N)}(t)^{2} \\
&=\frac{1}{2}\int\textrm{d}t\, f(t)^{2}\left[\textrm{Re}\left((\varsigma_{j}^{(N)}(t)+\overline{\xi}_{j}^{(N)}(t))^{2}W_{j}^{(N)}(t)^{2}\right)+\lvert\varsigma_{j}^{(N)}(t)+\overline{\xi}_{j}^{(N)}(t)\rvert^{2}\lvert W_{j}^{(N)}(t) \rvert^{2}\right]
\end{align*}
and
\begin{align*}
||D_{j}^{(N)}||^{2} &=\int\textrm{d}t\, D_{j}^{(N)}(t)^{2} \\
&=\frac{1}{2}\int\textrm{d}t\, f(t)^{2}\left[\lvert\varsigma_{j}^{(N)}(t)+\overline{\xi}_{j}^{(N)}(t)\rvert^{2}\lvert W_{j}^{(N)}(t) \rvert^{2}-\textrm{Re}\left((\varsigma_{j}^{(N)}(t)+\overline{\xi}_{j}^{(N)}(t))^{2}W_{j}^{(N)}(t)^{2}\right)\right]
\end{align*}
Taking into account \eqref{WKBpar-j} and the estimates below equations \eqref{fW.fW} and \eqref{fW*.fW}, we get
\begin{equation*}
 \delta_{j}=\mathcal{O}(\lambda_{j}^{-N-1/2}) \; .
\end{equation*}
The pre-factor of the first term in \eqref{A-fT-expand} is of order
\[(\delta_{j})^{2}=\mathcal{O}(\lambda_{j}^{-2N-1}) \; ,\]
while the one of the second term,
\[(\delta_{j})(2-\delta_{j})=\mathcal{O}(\lambda_{j}^{-N-1/2}) \; .\]
This last one imposes more stringent restrictions. Taking \eqref{adiabatic-lambda} into account, we obtain for the Hilbert-Schmidt norm of $R_s$
\begin{equation}
||R_s||^2_2\leq\sum_j (1+\lambda_j)^{2s-2N-2} \; .
\end{equation}  
For the Laplacian on a compact Riemannian space of dimension $m$ we know from Weyl's estimate \citep{Jost11} that $\lambda_j$ is bounded by some constant times $j^{\frac{2}{m}}$. Hence the Hilbert-Schmidt norm of $R_s$ is finite if 
\begin{equation}
s<N+1-\frac{m}{4} \; .
\label{SJ_f-Hadamard}
\end{equation}

The modified S-J states are independent of the order of the adiabatic approximation. They thus have the same Sobolev wave front sets as Hadamard states for every index $s$ and therefore fulfill the Hadamard condition.
\end{proof}

We note that, in comparison with the Hilbert-Schmidt norm appearing in the paper containing the results of this chapter, \citep{BrumFredenhagen14}, we have now corrected a factor of $1/4$ in \eqref{SJ_f-Hadamard}. This correction does not change the conclusion of the paper at all.

We remark further that for $f\notin \mathcal{C}_{0}^{\infty}(\M)$, the proof that the scalar product \eqref{fW*.fW} decays faster than any power of $\lambda_{j}$ breaks down, and thus there could be some $s$ for which the operator $R_s$ would not be Hilbert-Schmidt, thus $\omega_{SJ_{f}}$ would not be a Hadamard state.



\chapter{Hadamard state in Schwarzschild-de Sitter spacetime}\label{chap_Sch-dS}

\section{Introduction}

We will construct here a Hadamard state in the Schwarzschild-de Sitter spacetime. To our knowledge, this is the first explicit example of such a state in this spacetime (the general existence of Hadamard states in a globally hyperbolic spacetime was proved in \citep{FullingNarcowichWald81}. See the discussion in section \ref{secstates}). The state will be defined in some regions of the Penrose diagram, not in its full Kruskal extension, and we will show that it is invariant under the action of the group of symmetries of this spacetime. Hence, the state constructed here will not be the Hartle-Hawking state for this spacetime, whose nonexistence was proven in \citep{KayWald91}. As we will make clear later, neither can this state be interpreted as the Unruh state.

One point of the nonexistence theorems mentioned above which must be already emphasized is that the state which will arise in our construction does not satisfy the KMS condition at any point of the region where it is defined. The impossibility of having a thermal state, invariant under the action of the group of symmetries of the Schwarzschild-de Sitter spacetime, was already pointed in \citep{GibbonsHawking77}, where the authors showed the existence of a temperature in the de Sitter spacetime, related to the surface gravity of the event horizon, as in the classical Hawking effect \citep{Hawking74}.

Concerning the thermal nature of these states, it is important to mention, as stressed in \citep{BuchholzSolveen13,Solveen12}, temperature here is meant only in the sense of an ``absolute temperature'', as appears in the efficiency of a Carnot cycle. The ``calorimetric temperature'', the parameter used to define thermal equilibrium among thermodynamic states, on the other hand, is a local concept, defined as the expectation value of an observable. As they proved, in the Schwarzschild spacetime, the absolute temperature corresponding to the surface gravity of the black hole horizon, and uniquely determining a ``thermal state'', is equal to the calorimetric temperature, evaluated in this state. On the other hand, in the de Sitter spacetime, while the absolute temperature is given by the surface gravity of the cosmological horizon, the calorimetric temperature evaluated in the corresponding thermal state is zero. In this work, since we will be concerned with the KMS property of states, we will be talking about absolute temperatures.

Our construction is based on the bulk-to-boundary technique developed in \citep{DappiaggiMorettiPinamonti06,DappiaggiMorettiPinamonti09b,DappiaggiMorettiPinamonti09c} and applied in the construction of the Unruh state in the Schwarzschild spacetime in \citep{DappiaggiMorettiPinamonti09}. As in \citep{DappiaggiMorettiPinamonti09}, we apply that technique to the construct a Hadamard state in the Schwarzschild-de Sitter spacetime.

\section{Schwarzschild-de Sitter Spacetime}

The Schwarzschild-de Sitter (SdS) spacetime is a spherically symmetric solution of the Einstein equations in the presence of a positive cosmological constant. Its metric, in the coordinates $(t,r,\theta,\varphi)$, has the form \citep{GriffithsPodolsky09}
\begin{equation}
ds^{2}=\left(1-\frac{2M}{r}-\frac{\Lambda}{3}r^{2}\right)dt^{2}-\left(1-\frac{2M}{r}-\frac{\Lambda}{3}r^{2}\right)^{-1}dr^{2}-r^{2}(d\theta^{2}+\sin^{2}(\theta)d\varphi^{2}) \; ,
\label{SdS-metric}
\end{equation}
where $M>0$ is the black hole mass and $\Lambda$ is the cosmological constant (we will consider only $\Lambda>0$, the other case being the so-called {\it Anti-de Sitter spacetime}). The coordinates $(\theta,\varphi)$ have the usual interpretation of polar angles. The interpretation of the coordinates $t$ and $r$ will be given below. The first question we ask ourselves is about the location of horizons, i.e., positive roots of $F(r)\coloneqq\left(1-\frac{2M}{r}-\frac{\Lambda}{3}r^{2}\right)$. One can see that, if $3M\sqrt{\Lambda}>1$, $F(r)$ has only one negative real root,
\[r_{0}=\left[\frac{-3M}{\Lambda}-\sqrt{\frac{9M^{2}}{\Lambda^{2}}-\frac{1}{\Lambda^{3}}}\right]^{1/3}+\left[\frac{-3M}{\Lambda}-\sqrt{\frac{9M^{2}}{\Lambda^{2}}-\frac{1}{\Lambda^{3}}}\right]^{-1/3} \; ,\]
the other two roots being complex,
\[z_{\pm}=e^{\mp i\pi/3}\left[\frac{-3M}{\Lambda}-\sqrt{\frac{9M^{2}}{\Lambda^{2}}-\frac{1}{\Lambda^{3}}}\right]^{1/3}+e^{\pm i\pi/3}\left[\frac{-3M}{\Lambda}-\sqrt{\frac{9M^{2}}{\Lambda^{2}}-\frac{1}{\Lambda^{3}}}\right]^{-1/3} \; .\]
If $3M\sqrt{\Lambda}=1$, $F(r)$ has two coincident positive real roots,
\[r_{1}=\frac{1}{\sqrt{\Lambda}} \; ,\]
and one negative real root,
\[r_{2}=\frac{-2}{\sqrt{\Lambda}} \; .\]
The case on which we will be interested, $3M\sqrt{\Lambda}<1$, is when $F(r)$ has two distinct positive real roots, corresponding to the horizons. Defining $\xi=\arccos(-3M\sqrt{\Lambda})$ ($\pi<\xi<3\pi/2$), the positive roots are located at
\begin{align}
r_{b}&=\frac{2}{\sqrt{\Lambda}}\cos\left(\frac{\xi}{3}\right) \nonumber \\
r_{c}&=\frac{2}{\sqrt{\Lambda}}\cos\left(\frac{\xi}{3}+\frac{4\pi}{3}\right) \; .
\label{event_hor}
\end{align}
The negative real root is located at
\[r_{-}=\frac{2}{\sqrt{\Lambda}}\cos\left(\frac{\xi}{3}+\frac{2\pi}{3}\right)=-(r_{b}+r_{c}) \; .\]
One can easily see that $2M<r_{b}<3M<r_{c}$ \citep{LakeRoeder77}. The horizon located at $r_{b}$ is a black hole horizon. One can see that $\lim_{\Lambda\rightarrow 0}r_{b}=2M$ and $\lim_{M\rightarrow 0}r_{b}=0$. On the other hand, the horizon located at $r_{c}$ is a cosmological horizon, $\lim_{\Lambda\rightarrow 0}r_{c}=\infty$ and $\lim_{M\rightarrow 0}r_{c}=\sqrt{3/\Lambda}$. Besides, $F(r)$ attains a maximum at $\overline{r}=\sqrt[3]{\frac{3M}{\Lambda}}$, $3M<\overline{r}<r_{c}$.

One can see from equation \eqref{SdS-metric} that the character of the coordinates $t$ and $r$ changes as one crosses the horizons. For $r_{b}<r<r_{c}$, $F(r)>0$ and $t$ is a timelike coordinate, $r$ being spacelike. If either $r<r_{b}$ or $r>r_{c}$, $F(r)<0$, $t$ becomes a spacelike coordinate and $r$, a timelike coordinate. Besides, it is immediate to see that the vector $X=\frac{\partial}{\partial_{t}}$ is a Killing vector. For $r_{b}<r<r_{c}$, $F(r)>0$ and the Killing vector is a timelike vector, thus this region of spacetime is a static region. If either $r<r_{b}$ or $r>r_{c}$, $F(r)<0$ and this vector becomes spacelike. Thus these are not static regions. On the horizons $r=r_{b}$ or $r=r_{c}$, $X$ is a null vector, $X^{a}X_{a}=0$. $X$ is normal to the horizon, and so is $\nabla^{a}(X^{b}X_{b})$. Therefore there exists a constant $\kappa$, the {\it surface gravity}, defined on the horizon such that
\begin{equation}
\nabla^{a}(X^{b}X_{b})=-2\kappa X^{a} \; .
\label{surface_grav_def}
\end{equation}
One can show that the surface gravity is a constant along the orbits of $X$ and also a constant over the horizon. Besides, one can also prove that \citep{Wald84}
\[\kappa^{2}=-\frac{1}{2}(\nabla^{a}X^{b})(\nabla_{a}X_{b}) \; .\]
For a spherically symmetric spacetime whose metric is of the form \eqref{SdS-metric}, but with a general function $F(r)$, one can prove that
\begin{equation}
\kappa=\frac{1}{2}F'(r)\vert_{r=r_{H}} \; ,
\label{kappa_horizon}
\end{equation}
where $r_{H}$ is a positive root of $F(r)$, i.e., a horizon. The surface gravities, in the particular case of the SdS spacetime, are given by
\begin{align}
\kappa_{b}&=\frac{1}{2}F'\vert_{r=r_{b}}=(r_{c}-r_{b})(r_{c}+2r_{b})\frac{\Lambda}{6r_{b}} \nonumber \\
\kappa_{c}&=\frac{1}{2}F'\vert_{r=r_{c}}=(r_{c}-r_{b})(2r_{c}+r_{b})\frac{\Lambda}{6r_{c}} \; .
\label{surface_grav}
\end{align}
It is immediate to see that $\kappa_{b}>\kappa_{c}$.

The metric \eqref{SdS-metric} is not regular at the horizons. 
As shown in \citep{BazanskiFerrari86}, one cannot obtain a coordinate system on which the metric is regular at both horizons, but we can construct a pair of coordinate systems such that each one renders the metric regular at one of the horizons. First, we define the usual tortoise coordinate $r_{\ast}$:
\begin{equation}
r_{\ast}=\int\frac{dr}{F(r)}=\frac{1}{2\kappa_{b}}\log\left(\frac{r}{r_{b}}-1\right)-\frac{1}{2\kappa_{c}}\log\left(1-\frac{r}{r_{c}}\right)-\frac{1}{2}\left(\frac{1}{\kappa_{b}}-\frac{1}{\kappa_{c}}\right)\log\left(\frac{r}{r_{b}+r_{c}}+1\right) \; .
\label{tortoise}
\end{equation}
It maps the region $r\in(r_{b},r_{c})$ into $r^{\ast}\in(-\infty,+\infty)$. We define null coordinates as $u=t-r_{\ast}$, $v=t+r_{\ast}$. The coordinate system which renders the metric regular at $r=r_{b}$ is defined as \citep{ChoudhuryPadmanabhan07}
\begin{equation}
U_{b}\coloneqq \frac{-1}{\kappa_{b}}e^{-\kappa_{b}u} \quad ; \quad V_{b}\coloneqq \frac{1}{\kappa_{b}}e^{\kappa_{b}v} \; .
\label{coord_UbVb}
\end{equation}
Since $u,v\in(-\infty,+\infty)$, $U_{b}\in(-\infty,0)$ and $V_{b}\in(0,+\infty)$. In these coordinates, the metric becomes (neglecting the angular part)
\begin{equation}
ds^{2}=\frac{2M}{r}\left(1-\frac{r}{r_{c}}\right)^{1+\kappa_{b}/\kappa_{c}}\left(\frac{r}{r_{b}+r_{c}}+1\right)^{2-\kappa_{b}/\kappa_{c}}dU_{b}dV_{b} \; .
\label{metric_UbVb}
\end{equation}
This expression is regular at $r_{b}$, but not at $r_{c}$. Thus, in this coordinate system, the metric covers the whole region $(0,r_{c})$ regularly. Therefore, we can extend $U_{b}$ to positive values and $V_{b}$ to negative values across the horizon at $r_{b}$. The Kruskal extension of this region is similar to the corresponding extension of the Schwarzschild spacetime \citep{Wald84}:

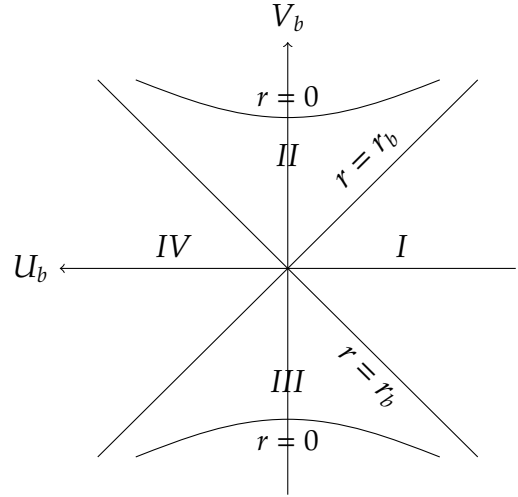
\begin{figure}[h]
\begin{center}
\begin{tikzpicture}
\draw[<-] (-3,0) -- (3,0);
\draw[->] (0,-3) -- (0,3);
\draw[thin] (-2.5,-2.5) -- node[near end,sloped,above] {$r=r_{b}$} (2.5,2.5);
\draw[thin] (2.5,-2.5) -- node[near start,sloped,below] {$r=r_{b}$} (-2.5,2.5);
\draw (-2,2.5) sin (0,2) cos (2,2.5); curve
\draw (-2,-2.5) sin (0,-2) cos (2,-2.5); curve
\node[above] at (1.5,0) {$I$};
\node[above] at (-1.5,0) {$IV$};
\node at (0,1.5) {$II$};
\node at (0,-1.5) {$III$};
\node[above] at (0,2) {\small $r=0$};
\node[below] at (0,-2) {\small $r=0$};
\node[left] at (-3,0) {$U_{b}$};
\node[above] at (0,3) {$V_{b}$};
\end{tikzpicture}
\end{center}
\caption{Conformal diagram of the Schwarzschild-de Sitter spacetime, extended only across the horizon at $r=r_{b}$.}
\label{Kruskal_rb}
\end{figure}

The region $I$ in figure \ref{Kruskal_rb} is the exterior region. Asymptotically, it tends to $r=r_{c}$. We call attention to the fact that $U_{b}$ increases to the left. Region $II$ is the black hole region. Any infalling observer initially at $I$ will fall inside this region and reach the singularity at $r=0$. Regions $III$ and $IV$ are copies of $II$ and $I$, the only difference being that, now, time runs in the opposite direction.

Similarly, we define the coordinate system which renders the metric regular at $r=r_{c}$:
\begin{equation}
U_{c}\coloneqq \frac{1}{\kappa_{c}}e^{\kappa_{c}u} \quad ; \quad V_{c}\coloneqq \frac{-1}{\kappa_{c}}e^{-\kappa_{c}v} \; ,
\label{coord_UcVc}
\end{equation}
where $U_{c}\in(0,+\infty)$ and $V_{c}\in(-\infty,0)$. In these coordinates,
\begin{equation}
ds^{2}=\frac{2M}{r}\left(\frac{r}{r_{b}}-1\right)^{1+\kappa_{c}/\kappa_{b}}\left(\frac{r}{r_{b}+r_{c}}+1\right)^{2-\kappa_{c}/\kappa_{b}}dU_{c}dV_{c} \; .
\label{metric_UcVc}
\end{equation}
This expression is regular at $r_{c}$, but not at $r_{b}$. Now, the metric covers the region $(r_{b},\infty)$ regularly. We can extend $U_{c}$ to negative values and $V_{c}$ to positive values across the horizon at $r_{c}$. The Kruskal extension of this region is similar to the corresponding extension of the de Sitter spacetime \citep{GibbonsHawking77}:

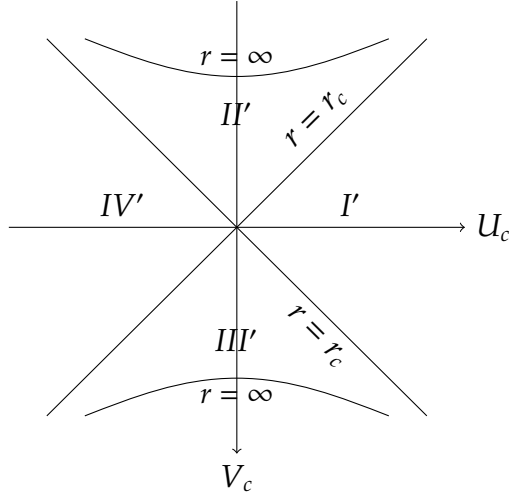
\begin{figure}[h]
\begin{center}
\begin{tikzpicture}
\draw[->] (-3,0) -- (3,0);
\draw[<-] (0,-3) -- (0,3);
\draw[thin] (-2.5,-2.5) -- node[near end,sloped,above] {$r=r_{c}$} (2.5,2.5);
\draw[thin] (2.5,-2.5) -- node[near start,sloped,below] {$r=r_{c}$} (-2.5,2.5);
\draw (-2,2.5) sin (0,2) cos (2,2.5); curve
\draw (-2,-2.5) sin (0,-2) cos (2,-2.5); curve
\node[above] at (1.5,0) {$I'$};
\node[above] at (-1.5,0) {$IV'$};
\node at (0,1.5) {$II'$};
\node at (0,-1.5) {$III'$};
\node[above] at (0,2) {\small $r=\infty$};
\node[below] at (0,-2) {\small $r=\infty$};
\node[right] at (3,0) {$U_{c}$};
\node[below] at (0,-3) {$V_{c}$};
\end{tikzpicture}
\end{center}
\caption{Conformal diagram of the Schwarzschild-de Sitter spacetime, extended only across the horizon at $r=r_{c}$.}
\label{Kruskal_rc}
\end{figure}

The region $I'$ in figure \ref{Kruskal_rc} is identical to region $I$ in figure \ref{Kruskal_rb}. Asymptotically, it tends to $r=r_{b}$. We call the attention to the fact that, now, $V_{c}$ increases downwards. Region $II'$ is the region exterior to the cosmological horizon. Any outwards directed observer initially at $I'$ will fall inside this region and reach the singularity at $r=\infty$. Regions $III'$ and $IV'$ are copies of $II'$ and $I'$, the only difference being that, now, times runs in the opposite direction.



The authors of \citep{BazanskiFerrari86} have also shown that transformations of coordinates of the form \eqref{coord_UbVb} and \eqref{coord_UcVc} are the only ones which give rise to a metric that is regular at one of the horizons. 

To obtain a maximally extended diagram, we first identify the regions $I$ and $I'$ of figures \ref{Kruskal_rb} and \ref{Kruskal_rc}, respectively. The wedges $IV$ and $IV'$ are also identical, hence we can combine new diagrams, identifying these wedges to the newly introduced wedges $IV'$ and $IV$, respectively. Now, the wedges $I$ and $I'$ can be combined with new wedges $I'$ and $I$, and this process is repeated indefinitely. Thus the maximally extended diagram is an infinite chain. In figure \ref{Sch-dS_conformal_diagram} below we depict part of the Penrose diagram of this maximally extended manifold (where we will also rename some of the regions):

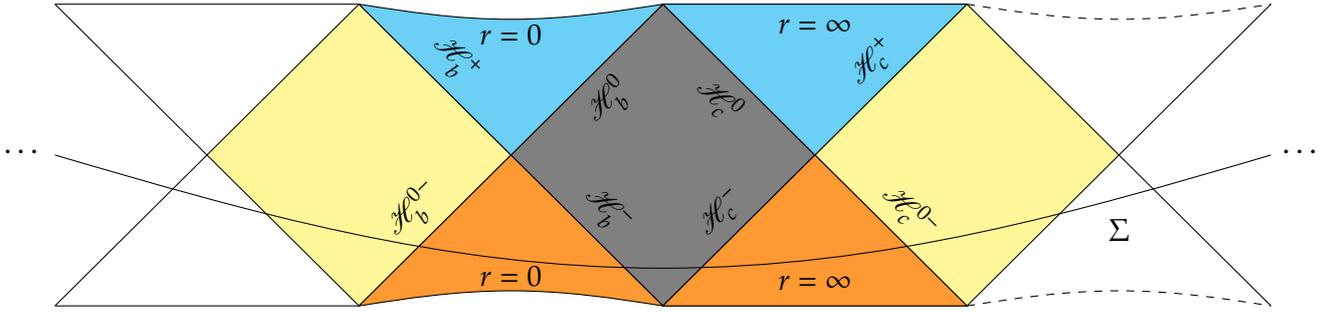
\begin{figure}[h]
\begin{center}
\begin{tikzpicture}
\draw[fill=yellow!50!] (0,0) -- (2,2) -- node[sloped,above] {\small $\Hor_{b}^{+}$} (4,0) -- (2,-2) -- (0,0); 
\draw[fill=gray] (4,0) -- node[sloped,below] {\small $\Hor_{b}^{0}$} (6,2) -- node[sloped,below] {\small $\Hor_{c}^{0}$} (8,0) -- node[sloped,above] {\small $\Hor_{c}^{-}$} (6,-2) -- node[sloped,above] {\small $\Hor_{b}^{-}$} (4,0); 
\draw[fill=yellow!50!] (8,0) -- node[sloped,above] {\small $\Hor_{c}^{+}$} (10,2) -- (12,0) -- (10,-2) -- (8,0); 
\draw (14,2) -- (12,0) -- (14,-2); 
\draw (-2,2) -- (0,0) -- (-2,-2); 
\draw (-2,2) -- (2,2);
\draw (-2,-2) -- (2,-2);
\draw (6,2) -- (10,2);
\draw (6,-2) -- (10,-2);
\draw[dashed] (10,2) sin (12,1.8) cos (14,2); curve
\draw[dashed] (10,-2) sin (12,-1.8) cos (14,-2); curve
\node[below=2pt] at (4,0) {\small $\cB_{b}$};
\node[below=2pt] at (8,0) {\small $\cB_{c}$};
\draw [fill=cyan!50!] (2,2) sin (4,1.8) cos (6,2) -- (4,0) -- node[sloped,above] {\small $\Hor_{b}^{+}$} (2,2); curve
\draw [fill=cyan!50!] (6,2) -- (10,2) -- node[sloped,above] {\small $\Hor_{c}^{+}$} (8,0) -- (6,2);
\draw [fill=orange!80!] (2,-2) sin (4,-1.8) cos (6,-2) -- (4,0) -- node[sloped,above] {\small $\Hor_{b}^{0-}$} (2,-2); curve
\draw [fill=orange!80!] (6,-2) -- (10,-2) -- node[sloped,above] {\small $\Hor_{c}^{0-}$} (8,0) -- (6,-2);
\draw (-2,0) sin (6,-1.5) cos (14,0); curve
\node at (12,-1) {$\Sigma$};
\node at (4,1.6) {\small $r=0$};
\node at (4,-1.6) {\small $r=0$};
\node at (8,1.7) {\small $r=\infty$};
\node at (8,-1.7) {\small $r=\infty$};
\node [left] at (-2,0) {$\cdots$};
\node [right] at (14,0) {$\cdots$};
\end{tikzpicture}
\end{center}
\caption{Maximally extended conformal diagram of the Schwarzschild-de Sitter spacetime.}
\label{Sch-dS_conformal_diagram}
\end{figure}

The region in gray is the region between the horizons. It will be denoted by $\sD$. The black-hole horizon is located at the surfaces denoted by $\Hor_{b}^{\pm (0,0-)}$, and the cosmological horizon, at $\Hor_{c}^{\pm (0,0-)}$.

In the coordinate systems $(U_{b(c)},V_{b(c)},\theta,\varphi)$, the horizons are defined as:
\[\Hor_{b}^{+}:\{(U_{b},V_{b},\theta,\varphi)\in\mathbb{R}^{2}\times\mathbb{S}^{2};\, U_{b}>0,V_{b}=0\}\; ; \; \Hor_{b}^{-}:\{(U_{b},V_{b},\theta,\varphi)\in\mathbb{R}^{2}\times\mathbb{S}^{2};\, U_{b}<0,V_{b}=0\}\; ;\]
\[\Hor_{b}^{0}:\{(U_{b},V_{b},\theta,\varphi)\in\mathbb{R}^{2}\times\mathbb{S}^{2};\, U_{b}=0,V_{b}>0\}\; ; \; \Hor_{b}^{0-}:\{(U_{b},V_{b},\theta,\varphi)\in\mathbb{R}^{2}\times\mathbb{S}^{2};\, U_{b}=0,V_{b}<0\}\; ;\]
\[\Hor_{c}^{+}:\{(U_{c},V_{c},\theta,\varphi)\in\mathbb{R}^{2}\times\mathbb{S}^{2};\, U_{c}=0,V_{c}>0\}\; ; \; \Hor_{c}^{-}:\{(U_{c},V_{c},\theta,\varphi)\in\mathbb{R}^{2}\times\mathbb{S}^{2};\, U_{c}=0,V_{c}<0\}\; ;\]
\[\Hor_{c}^{0}:\{(U_{c},V_{c},\theta,\varphi)\in\mathbb{R}^{2}\times\mathbb{S}^{2};\, U_{c}>0,V_{c}=0\}\; ; \Hor_{c}^{0-}:\{(U_{c},V_{c},\theta,\varphi)\in\mathbb{R}^{2}\times\mathbb{S}^{2};\, U_{c}<0,V_{c}=0\}\; .\]
The completely extended manifold will be denoted by $\sK$, and $\Sigma$ is a smooth Cauchy surface of $\sK$. In the region $\sD$, the Killing vector $X$ is timelike and future pointing. In the yellow regions, this Killing vector is also timelike, but past directed. The yellow region on the left will be called $IV$, and the one on the right, $IV'$. In the blue and orange regions, the Killing vector $X$ is spacelike. In the coordinate systems $(U_{b(c)},V_{b(c)},\theta,\varphi)$, these regions are defined as follows:
\[\sD\coloneqq \{(U_{b},V_{b},\theta,\varphi)\in\mathbb{R}^{2}\times\mathbb{S}^{2};U_{b}<0,V_{b}>0\}\; ;\]
\begin{alignat*}{3}
II\coloneqq \{(U_{b},V_{b},\theta,\varphi)\in\mathbb{R}^{2}\times\mathbb{S}^{2};U_{b}>0,V_{b}>0\} \quad && ; \quad && II'\coloneqq \{(U_{c},V_{c},\theta,\varphi)\in\mathbb{R}^{2}\times\mathbb{S}^{2};U_{c}>0,V_{c}>0\} \; ; \\
III\coloneqq \{(U_{b},V_{b},\theta,\varphi)\in\mathbb{R}^{2}\times\mathbb{S}^{2};U_{b}<0,V_{b}<0\} \quad && ; \quad && III'\coloneqq \{(U_{c},V_{c},\theta,\varphi)\in\mathbb{R}^{2}\times\mathbb{S}^{2};U_{c}<0,V_{c}<0\} \; .
\end{alignat*}
The region $\sD$ can be equivalently defined as
\[\sD=\{(U_{c},V_{c},\theta,\varphi)\in\mathbb{R}^{2}\times\mathbb{S}^{2};U_{c}>0,V_{c}<0\} \; .\]
The bifurcation spheres are defined as:
\[\cB_{b}:\{(U_{b},V_{b},\theta,\varphi)\in\mathbb{R}^{2}\times\mathbb{S}^{2};\, U_{b}=V_{b}=0\}\; ;\; \cB_{c}:\{(U_{c},V_{c},\theta,\varphi)\in\mathbb{R}^{2}\times\mathbb{S}^{2};\, U_{c}=V_{c}=0\}\; .\]
We note that the Killing vector $X$ vanishes on these spheres.

We will construct a Hadamard state in the region 
\[\M\coloneqq II\cup \sD\cup II'\]
(the gray and the blue regions in figure \ref{Sch-dS_conformal_diagram}). It is easy to see that any past inextensible causal curve passing through any point of $\M$ passes through $\sB\cup\sC$, where $\sB:\{(U_{b},V_{b},\theta,\varphi)\in\mathbb{R}^{2}\times\mathbb{S}^{2};\, V_{b}=0\}=\Hor_{b}^{+}\cup\cB_{b}\cup\Hor_{b}^{-}$, and $\sC:\{(U_{c},V_{c},\theta,\varphi)\in\mathbb{R}^{2}\times\mathbb{S}^{2};\, U_{c}=0\}=\Hor_{c}^{+}\cup\cB_{c}\cup\Hor_{c}^{-}$. Since this surface is achronal, $\sB\cup\sC$ is a Cauchy surface for $\M$. We will also show how the state can be restricted to the past horizons $\sB\cup\sC$ and will investigate the physical properties of this restriction.

\section{Algebras and State}

\subsection{Algebras}\label{sec-Sch-dS-algebra}

The algebra on which the state will operate as a functional is constructed from the solutions of the Klein-Gordon equation $\Box_{g}\phi=0$, where $g$ is the metric on the spacetime and $\Box_{g}$ is the d'Alembert operator corresponding to this metric. The symplectic space is given by the pair $\left(S(\M),\sigma_{\M}\right)$, where $S(\M)$ is the vector space of solutions of the Klein-Gordon equation having particular decaying properties (see complete definition below) and $\sigma_{\M}$ is the symplectic form built out of the {\it advanced-minus-retarded operator} (see section \ref{sec-CCR-Weyl}). Dafermos and Rodnianski \citep{DafermosRodnianski07} showed that, if the solutions of the Klein-Gordon equation have smooth initial data on $\Sigma$, then there exist, due to spherical symmetry, a constant $c$ which depends on $M$ and $\Lambda$, and another constant $C$ depending on $M$, $\Lambda$, the geometry of $\Sigma\cap J^{-}(\sD)$ and on the initial values of the field, such that
\[|\phi_{l}(u,v)|\leq C(e^{-cv_{+}/l^{2}}+e^{-cu_{+}/l^{2}})\]
and
\begin{equation}
|\phi_{0}(u,v)-\uline{\phi}|\leq C(e^{-cv_{+}/l^{2}}+e^{-cu_{+}/l^{2}})
\label{Sch-dS_decay}
\end{equation}
are valid in $J^{+}(\Sigma)\cap\sD$. Here,
\[|\uline{\phi}|\leq \underset{x\in\Sigma}{\textrm{inf}}\, |\phi_{0}(x)|+C \; ,\]
$u_{+}=\mathrm{max}\{u,1\}$, $v_{+}=\mathrm{max}\{v,1\}$ and $l$ is the spherical harmonic. These bounds are also valid on the horizons, and this feature will play a crucial role in the restriction of the algebra to the horizons and in the subsequent construction of the state. The regions $IV$ and $IV'$ are also static regions, with time running in the opposite direction. Therefore, this fast decay is also verified on $\Hor_{b}^{+}$ and on $\Hor_{c}^{+}$. Moreover, we make the further requirement that the solutions vanish at the bifurcation spheres $\cB_{b}$ and $\cB_{c}$ (as remarked in \citep{DafermosRodnianski07}, this requirement creates no additional complications).

The vector spaces of solutions in $\M$ and on the horizons are defined as
\begin{align}
S(\M):\Big\{&(\phi-\phi_{0})\, |\; \phi\in\cC^{\infty}(\M;\mathbb{R}), \Box_{g}\phi=0;\; \phi_{0}=\textrm{constant}\Big\} \; ,
\label{S_M} \\
S(\sB):\Big\{&(\phi-\phi_{0})\, |\; \phi\in\cC^{\infty}(\sB;\mathbb{R}), \Box_{g}\phi=0, \phi=0 \textrm{ at }\cB_{b};\; \phi_{0}=\textrm{constant;}\; \exists\exists C_{\phi}>0, C'>1 \nonumber \\
&\textrm{with } |\phi-\phi_{0}|<\frac{C_{\phi}}{U_{b}},\, |\partial_{U_{b}}\phi|<\frac{C_{\phi}}{U_{b}^{2}}\, \textrm{for } |U_{b}|>C'\Big\} \; ,
\label{S_sB} \\
S(\Hor_{b}^{-}):\Big\{&(\phi-\phi_{0})\, |\; \phi\in\cC^{\infty}(\Hor_{b}^{-};\mathbb{R}), \Box_{g}\phi=0;\; \phi_{0}=\textrm{constant;}\; \exists\exists C_{\phi}>0, C'>1 \nonumber \\
&\textrm{with } |\phi-\phi_{0}|<C_{\phi}e^{-u},\, |\partial_{u}\phi|<C_{\phi}e^{-u}\, \textrm{for } |u|>C'\Big\} \; ,
\label{S_Hor_b-} \\
S(\sC):\Big\{&(\phi-\phi_{0})\, |\; \phi\in\cC^{\infty}(\sC;\mathbb{R}), \Box_{g}\phi=0, \phi=0 \textrm{ at }\cB_{c};\; \phi_{0}=\textrm{constant;}\; \exists\exists C_{\phi}>0, C'>1 \nonumber\\
&\textrm{with } |\phi-\phi_{0}|<\frac{C_{\phi}}{V_{c}},\, |\partial_{V_{c}}\phi|<\frac{C_{\phi}}{V_{c}^{2}}\, \textrm{for } |V_{c}|>C'\Big\} \; ,
\label{S_sC} \\
S(\Hor_{c}^{-}):\Big\{&(\phi-\phi_{0})\, |\; \phi\in\cC^{\infty}(\Hor_{c}^{-};\mathbb{R}), \Box_{g}\phi=0;\; \phi_{0}=\textrm{constant;}\; \exists\exists C_{\phi}>0, C'>1 \nonumber \\
&\textrm{with } |\phi-\phi_{0}|<C_{\phi}e^{-v},\, |\partial_{v}\phi|<C_{\phi}e^{-v}\, \textrm{for } |v|>C'\Big\} \; .
\label{S_Hor_c-}
\end{align}

The Weyl algebras $\W(S(\M))$, $\W(S(\sB))$, $\W(S(\sC))$, $\W(S(\Hor_{b}^{-}))$, $\W(S(\Hor_{c}^{-}))$ (we will omit the $\sigma$'s to simplify the notation) are constructed from each of the symplectic spaces as was done at the end of section \ref{sec-CCR-Weyl}. We will now prove the existence of an injective isometric $\ast$-homomorphism between the Weyl algebra defined in $\M$ and the tensor product of the algebras defined on $\sB$ and on $\sC$. In the next subsection, we will show that the pullback of this map yields such a mapping of the corresponding state spaces (of course, in the opposite sense).

\begin{thm}\label{algebra-horizon-isomorphism}
For every $\phi\in S(\M)$, let us define
\[\phi_{\sB}\coloneqq \phi_{\upharpoonright\sB} \quad ; \quad \phi_{\sC}\coloneqq \phi_{\upharpoonright\sC} \; .\]
Then, the following holds:

(a)\ \ The linear map
\[\Gamma:S(\M) \ni \phi \mapsto \left(\phi_{\sB},\phi_{\sC}\right)\]
is an injective symplectomorphism of $S(\M)$ into $S(\sB)\oplus S(\sC)$ equipped with the symplectic form, s.t., for $\phi,\phi'\in S(\M)$:
\begin{equation}
\sigma_{\M}(\phi,\phi')\coloneqq\sigma_{\sB}(\phi_{\sB},\phi'_{\sB})+\sigma_{\sC}(\phi_{\sC},\phi'_{\sC}) \; .
\label{sigma-B;C}
\end{equation}

(b)\ \ There exists a corresponding injective isometric $^{\ast}$-homomorphism
\[\iota:\W(S(\M))\rightarrow \W(S(\sB))\otimes\W(S(\sC)) \; ,\]
which is uniquely individuated by
\begin{equation}
\iota(W_{\M}(\phi))=W_{\sB}(\phi_{\sB})W_{\sC}(\phi_{\sC}) \; .
\label{iota-B;C}
\end{equation}
\end{thm}
\begin{proof}
(a)\ \ Every $\phi\in S(\M)$ can be thought of as a restriction to $\M$ of a solution $\phi'$ of the Klein-Gordon equation in the whole Kruskal manifold. Similarly, the Cauchy surface of $\M$ can be seen as a subset of a Cauchy surface of $\sK$ \citep{BernalSanchez03,ONeill83}. Therefore $\phi_{\sB}\coloneqq \phi'_{\upharpoonright\sB}$ and $\phi_{\sC}\coloneqq \phi'_{\upharpoonright\sC}$ are well-defined and smooth.

Since the map $\Gamma$ is linear by construction, it remains to prove that (i) $\phi_{\sB}\in S(\sB)$ and $\phi_{\sC}\in S(\sC)$ and that (ii) $\Gamma$ preserves the symplectic form, i.e.,
\begin{equation}
\sigma_{\M}(\phi_{1},\phi_{2})=\sigma_{S(\sB)\oplus S(\sC)}(\Gamma\phi_{1},\Gamma\phi_{2}) \; .
\label{sigma-isomorphism}
\end{equation}
Since $\sigma_{\M}$ is nondegenerate, this identity implies that the linear map $\Gamma$ is injective. Moreover, from the bounds \eqref{Sch-dS_decay} and smoothness we have, automatically, that
\begin{equation}
\phi\in S(\M)\Rightarrow \phi_{\upharpoonright\Hor_{b}^{-}}\in S(\Hor_{b}^{-}) \; ,
\label{phi_restr-H_b-}
\end{equation}
and similarly to $\Hor_{c}^{-}$. Furthermore, using the corresponding bounds in regions $IV$ and $IV'$ of figure \ref{Sch-dS_conformal_diagram} and smoothness of the solutions in the whole $\sK$, we find that \eqref{phi_restr-H_b-} is valid also for $\Hor_{b}^{+}$ and $\Hor_{c}^{+}$. The restrictions of the solution to the bifurcation spheres $\cB_{b}$ and $\cB_{c}$ vanish, again because of the smoothness. Finally,
\begin{equation}
\phi\in S(\M)\Rightarrow \phi_{\upharpoonright\sB}\in S(\sB) \; ,
\label{phi-restr_sB}
\end{equation}
and similarly for the restriction to $\sC$.

Let us finally prove (ii), i.e., equation \eqref{sigma-isomorphism}. 
Consider $\phi,\phi'\in S(\M)$ and a spacelike Cauchy surface $\Sigma$ of $\M$ so that
\[\sigma_{\M}(\phi,\phi')=\int_{\Sigma}\left(\phi'\nabla_{n}\phi-\phi\nabla_{n}\phi'\right)d\mu_{g}\]
where $n$ is the unit normal to the surface $\Sigma$ and $\mu_{g}$ is the metric induced measure on $\Sigma$. Since both $\phi$ and $\phi'$ vanish for sufficiently large $U_{b}$ and $U_{c}$ we can use the surface $\Sigma$, defined as the locus $t=0$ in $\sD$, and, out of the Poincar\'e theorem, we can write
\begin{align}
\sigma_{\M}(\phi,\phi')=&\int_{\Sigma\cap\sD}\left(\phi' X(\phi)-\phi X(\phi')\right)d\mu_{g}+r_{b}^{2}\int_{\Hor_{b}^{+}}\left(\phi'\partial_{U_{b}}\phi-\phi\partial_{U_{b}}\phi'\right)dU_{b}\wedge d\mathbb{S}^{2} \nonumber \\
+&r_{c}^{2}\int_{\Hor_{c}^{+}}\left(\phi'\partial_{V_{c}}\phi-\phi\partial_{V_{c}}\phi'\right)dV_{c}\wedge d\mathbb{S}^{2}
\label{sigma-to-bulk-and-boundary}
\end{align}
where we have used the fact that $\cB_{b} \cap \Sigma$ and $\cB_{c} \cap \Sigma$ have measure zero. We shall prove that, if one restricts the integration to $\sD$,
\begin{align}
\int_{\Sigma\cap\sD}\left(\phi' X(\phi)-\phi X(\phi')\right)d\mu_{g}=&r_{b}^{2}\int_{\Hor_{b}^{-}}\left(\phi'\partial_{U_{b}}\phi-\phi\partial_{U_{b}}\phi'\right)dU_{b}\wedge d\mathbb{S}^{2} \nonumber \\
+&r_{c}^{2}\int_{\Hor_{c}^{-}}\left(\phi'\partial_{V_{c}}\phi-\phi\partial_{V_{c}}\phi'\right)dV_{c}\wedge d\mathbb{S}^{2} \; .
\label{sigmabulk-to-boundary}
\end{align}

In order to prove this last identity, we notice first that
\[\int_{\Sigma\cap\sD}\left(\phi' X(\phi)-\phi X(\phi')\right)d\mu_{g}=\int_{[r_{b},r_{c}]\times\mathbb{S}^{2}}\frac{r^{2}\left(\phi' X(\phi)-\phi X(\phi')\right)}{1-2m/r-\Lambda r^{2}/3}\Bigg\rvert_{t=0}dr\wedge d\mathbb{S}^{2}\]
Afterwards, we will break the integral on the rhs into two pieces with respect to the coordinate $r^{\ast}$:
\begin{align*}
&\int_{[r_{b},r_{c}]\times\mathbb{S}^{2}}\frac{r^{2}\left(\phi' X(\phi)-\phi X(\phi')\right)}{1-2m/r-\Lambda r^{2}/3}\Bigg\rvert_{t=0}dr\wedge d\mathbb{S}^{2}= \\
&\int_{[-\infty,\hat{R}^{\ast}]\times\mathbb{S}^{2}}r^{2}\left(\phi' X(\phi)-\phi X(\phi')\right)\big\rvert_{t=0}dr^{\ast}\wedge d\mathbb{S}^{2}+\int_{[\hat{R}^{\ast},+\infty]\times\mathbb{S}^{2}}r^{2}\left(\phi' X(\phi)-\phi X(\phi')\right)\big\rvert_{t=0}dr^{\ast}\wedge d\mathbb{S}^{2} \; .
\end{align*}
We assumed $t=0$, but the value of $t$ is immaterial, since we can work with a different surface $\Sigma_{t}$ obtained by evolving $\Sigma$ along the flux of the Killing vector $X$ ($X_{\upharpoonright\sD}=\partial_{t}$ and $X_{\upharpoonright\cB_{b}}=X_{\upharpoonright\cB_{c}}=0$). Furthermore, the symplectic form $\sigma_{\M}(\phi,\phi')$ is constructed in such a way that its value does not change with varying $t$.

Since $\cB_{b}$ and $\cB_{c}$ are fixed under the flux of $X$, per direct application of the Stokes-Poincar\'e theorem, one sees that this invariance holds also for the integration restricted to $\sD$. In other words, $\forall t>0$:
\begin{align*}
&\int_{\Sigma\cap\sD}\left(\phi' X(\phi)-\phi X(\phi')\right)d\mu_{g}=\int_{\Sigma_{t}\cap\sD}\left(\phi' X(\phi)-\phi X(\phi')\right)d\mu_{g}= \\
&\int_{(u_{0}(t),+\infty)\times\mathbb{S}^{2}}r^{2}\left(\phi' X(\phi)-\phi X(\phi')\right)\big\rvert_{(t,t-u,\theta,\varphi)}du\wedge d\mathbb{S}^{2}+ \\
&\int_{(v_{0}(t),+\infty)\times\mathbb{S}^{2}}r^{2}\left(\phi' X(\phi)-\phi X(\phi')\right)\big\rvert_{(t,v-t,\theta,\varphi)}dv\wedge d\mathbb{S}^{2}
\end{align*}
where we have also changed the variables of integration from $r^{\ast}$ to either $u=t-r^{\ast}$ or $v=t+r^{\ast}$, and both $u_{0}(t)\coloneqq t-\hat{R}^{\ast}$ and $v_{0}(t)\coloneqq t+\hat{R}^{\ast}$ are functions of $t$. Hence
\begin{align}
&\int_{\Sigma\cap\sD}\left(\phi' X(\phi)-\phi X(\phi')\right)d\mu_{g}=\lim_{t\rightarrow-\infty}\int_{(u_{0}(t),+\infty)\times\mathbb{S}^{2}}r^{2}\left(\phi' X(\phi)-\phi X(\phi')\right)du\wedge d\mathbb{S}^{2} \nonumber \\
&+\lim_{t\rightarrow-\infty}\int_{(v_{0}(t),+\infty)\times\mathbb{S}^{2}}r^{2}\left(\phi' X(\phi)-\phi X(\phi')\right)dv\wedge d\mathbb{S}^{2} \; .
\label{sigma-past_horizons}
\end{align}
The former limit should give rise to an integral over $\Hor_{b}^{-}$, whereas the latter to an analogous one over $\Hor_{c}^{-}$. Let us examine the former one (the latter being similar):

Fix $u_{1}-\hat{R}^{\ast}\in\mathbb{R}$ and the following decomposition
\begin{align*}
&\int_{(u_{0}(t),+\infty)\times\mathbb{S}^{2}}r^{2}\left(\phi' X(\phi)-\phi X(\phi')\right)\big\rvert_{(t,t-u,\theta,\varphi)}du\wedge d\mathbb{S}^{2}=\int_{\Sigma_{t}^{u_{1}}}\left(\phi' X(\phi)-\phi X(\phi')\right)d\mu_{g} \\
&+\int_{(u_{0}(t),u_{1})\times\mathbb{S}^{2}}r^{2}\left(\phi' X(\phi)-\phi X(\phi')\right)\big\rvert_{(t,t-u,\theta,\varphi)}du\wedge d\mathbb{S}^{2} \; .
\end{align*}
Here we have used the initial expression for the first integral, which is performed over the compact subregion $\Sigma_{t}^{u_{1}}$ of $\Sigma_{t}\cap\sD$ which contains the points with null coordinate $U_{b}$ included in $\left[-(1/\kappa_{b})e^{-\kappa_{b}u_{1}},0\right]$ (see Figure \ref{decomposition_Cauchy}). It is noteworthy that such integral is indeed the one of the smooth 3-form
\[\eta=\eta[\phi,\phi']\coloneqq\frac{1}{6}\left(\phi\nabla^{a}\phi'-\phi'\nabla^{a}\phi\right)\sqrt{-g}\epsilon_{abcd}dx^{b}\wedge dx^{c}\wedge dx^{d}\]
and, furthermore, in view of the Klein-Gordon equation, $d\eta\equiv 0$. Thus, by means of an appropriate use of the Stokes-Poincar\'e theorem, this integral can be rewritten as an integral of $\eta$ over two regions. One of them is a compact subregion of $\Hor_{b}^{-}$ which can be constructed as the points with coordinate $U_{b}\in\left[U_{b_{1}},0\right]$, where $U_{b_{1}}=-(1/\kappa_{b})e^{-\kappa_{b}u_{1}}$ (this region is named $\Hor_{b_{1}}^{-}$ in figure \ref{decomposition_Cauchy}); the other one is the compact null 3-surface $S_{t}^{(u_{1})}$ formed by the points in $\M$ with $U_{b}=U_{b_{1}}$ and lying between $\Sigma_{t}$ and $\Hor_{b}^{-}$.

\begin{figure}
\begin{center}
\begin{tikzpicture}
\draw (0,0) -- (2,2) -- (4,0) -- (2,-2) -- (0,0); 
\draw (0,0) sin (2,-0.8) cos (4,0); curve 
\path [name path=U] (1,-1) -- (2,0);
\path [name path=s] (0,0) sin (2,-0.8) cos (4,0); curve
\draw [name intersections={of=U and s, by=x}] (1,-1) -- (x); 
\draw [->] (3,-0.6) -- (3.5,-1);
\node [below right] at (3.5,-1) {$\Sigma_{t}$};
\draw [->] (0.5,-0.2) -- (0.5,1);
\node [above] at (0.5,1) {$\Sigma_{t}^{(u_{1})}$};
\draw [->] (1.2,-0.9) -- (1.4,-1.1);
\node [right] at (1.4,-1.1) {$S_{t}^{(u_{1})}$};
\draw [->] (0.5,-0.6) -- (0.3,-0.9);
\node [below] at (0.3,-0.9) {$\Hor_{b_{1}}^{-}$};
\end{tikzpicture}
\end{center}
\caption{Only the region $\sD$ of figure \ref{Sch-dS_conformal_diagram} is depicted here, and $\Hor_{b_{1}}^{-} \coloneqq \Hor_{b}^{-}\cap\{U_{b_{1}}\leq U_{b}\leq 0\}$.}
\label{decomposition_Cauchy}
\end{figure}


To summarise:
\[\int_{\Sigma_{t}^{u_{1}}}\left(\phi' X(\phi)-\phi X(\phi')\right)d\mu_{g}=\int_{\Hor_{b_{1}}^{-}}\eta +\int_{S_{t}^{(u_{1})}}\eta \; .\]
If we adopt coordinates $U_{b},V_{b},\theta,\varphi$, the direct evaluation of the first integral on the rhs yields:
\[\int_{\Sigma_{t}^{u_{1}}}\left(\phi' X(\phi)-\phi X(\phi')\right)d\mu_{g}=r_{b}^{2}\int_{\Hor_{b_{1}}^{-}}\left(\phi'\partial_{U_{b}}\phi-\phi\partial_{U_{b}}\phi'\right)dU_{b}\wedge d\mathbb{S}^{2}+\int_{S_{t}^{(u_{1})}}\eta \; .\]
We have obtained:
\begin{align}
&\lim_{t\rightarrow-\infty}\int_{(u_{1},+\infty)\times\mathbb{S}^{2}}r^{2}\left(\phi' X(\phi)-\phi X(\phi')\right)\big\rvert_{(t,t-u,\theta,\varphi)}du\wedge d\mathbb{S}^{2}= \nonumber \\
&r_{b}^{2}\int_{\Hor_{b_{1}}^{-}}\left(\phi'\partial_{U_{b}}\phi-\phi\partial_{U_{b}}\phi'\right)dU_{b}\wedge d\mathbb{S}^{2}+\lim_{t\rightarrow-\infty}\int_{S_{t}^{(u_{1})}}\eta+ \nonumber \\
&\lim_{t\rightarrow-\infty}\int_{[u_{0}(t),u_{1}]\times\mathbb{S}^{2}}r^{2}\left(\phi' X(\phi)-\phi X(\phi')\right)\big\rvert_{(t,t-u,\theta,\varphi)}du\wedge d\mathbb{S}^{2} \; .
\label{sigma-hor_b}
\end{align}

If we perform the limit as $t\rightarrow-\infty$, one has $\int_{S_{t}^{(u_{1})}}\eta\rightarrow0$, because it is the integral of a smooth form over a vanishing surface (as $t\rightarrow-\infty$), whereas
\begin{align*}
&\lim_{t\rightarrow-\infty}\int_{[u_{0}(t),u_{1}]\times\mathbb{S}^{2}}r^{2}\left(\phi' X(\phi)-\phi X(\phi')\right)\big\rvert_{(t,t-u,\theta,\varphi)}du\wedge d\mathbb{S}^{2} \\
&=r_{b}^{2}\int_{\Hor_{b}^{-}\cap\{u_{1}\ge u\}\times\mathbb{S}^{2}}\chi_{(u_{0}(t),+\infty)}(u)\left(\phi'\partial_{u}\phi-\phi\partial_{u}\phi'\right)du\wedge d\mathbb{S}^{2} \\
&=r_{b}^{2}\int_{\Hor_{b}^{-}\cap\{U_{b_{1}}\ge U_{b}\}\times\mathbb{S}^{2}}\left(\phi'\partial_{U_{b}}\phi-\phi\partial_{U_{b}}\phi'\right)dU_{b}\wedge d\mathbb{S}^{2} \; ,
\end{align*}
where $\chi_{(u_{0}(t),+\infty)}$ is the characteristic function of the interval $(u_{0}(t),+\infty)$. Here we stress that the final integrals are evaluated over $\Hor_{b}^{-}$ and we have used again Lebesgue's dominated convergence theorem. Inserting the achieved results in the rhs of \eqref{sigma-hor_b}, we find that:
\[\lim_{t\rightarrow-\infty}\int_{\mathbb{R}\times\mathbb{S}^{2}}r^{2}\left(\phi' X(\phi)-\phi X(\phi')\right)\big\rvert_{(t,t-u,\theta,\varphi)}du\wedge d\mathbb{S}^{2}=r_{b}^{2}\int_{\Hor_{b}^{-}\times\mathbb{S}^{2}}\left(\phi'\partial_{U_{b}}\phi-\phi\partial_{U_{b}}\phi'\right)dU_{b}\wedge d\mathbb{S}^{2} \; .\]
Such identity, brought in \eqref{sigma-past_horizons} yields, together with a similar one for the integral over $v$, \eqref{sigmabulk-to-boundary}. Together with \eqref{sigma-to-bulk-and-boundary} this concludes the proof of (a).

Item (b) can be proved as follows: In the following, $S\coloneqq S(\sB)\oplus S(\sC)$ and $\sigma$ is the natural symplectic form on such space. Let us consider the closure of the sub $\ast$-algebra generated by all the generators $W_{S}(\Gamma\phi)\in\W(S)$ for all $\phi\in S(\M)$. This is still a C$^{\ast}$-algebra which, in turn, defines a realization of $\W(S(\M))$ because $\Gamma$ is an isomorphism of the symplectic space $(S(\M),\sigma_{\M})$ onto the symplectic space $(\Gamma(S(\M)),\sigma_{\upharpoonright\Gamma(S(\M))\times\Gamma(S(\M))})$. As a consequence of theorem 5.2.8 in \citep{BraRob-II}, there is a $\ast$-isomorphism, hence isometric, between $\W(S(\M))$ and the other, just found, realization of the same Weyl algebra, unambiguously individuated by the requirement $\iota_{\M}(W_{\M})(\phi)\coloneqq W_{S}(\Gamma\phi)$. This isometric $\ast$-isomorphism individuates an injective $\ast$-homomorphism of $\W(S(\M))$ into $\W(S,\sigma)\equiv\W(S(\sB))\otimes\W(S(\sC))$.
\end{proof}

The proof of theorem \ref{algebra-horizon-isomorphism} establishes the following
\begin{thm}\label{state-homomorphism}
With the same definitions as in the theorem \ref{algebra-horizon-isomorphism} and defining, for $\phi\in S(\sD)$, $\phi_{\upharpoonright\Hor_{b}^{-}}=\lim_{\rightarrow\Hor_{b}^{-}}\phi$ and $\phi_{\upharpoonright\Hor_{b}^{0}}=\lim_{\rightarrow\Hor_{b}^{0}}\phi$ (similarly for $\Hor_{c}^{-}$ and $\Hor_{c}^{0}$), the linear maps
\begin{align*}
\Gamma_{-}:S(\sD)\ni\phi \mapsto \left(\phi_{\Hor_{b}^{-}},\phi_{\Hor_{c}^{-}}\right)\in S(\Hor_{b}^{-})\oplus S(\Hor_{c}^{-}) \\
\Gamma_{0}:S(\sD)\ni\phi \mapsto \left(\phi_{\Hor_{b}^{0}},\phi_{\Hor_{c}^{0}}\right)\in S(\Hor_{b}^{0})\oplus S(\Hor_{c}^{0}) \; ,
\end{align*}
are well-defined injective symplectomorphisms. As a consequence, there exists two corresponding injective isometric $\ast$-homomorphisms:
\begin{align*}
\iota^{-}:\W(S(\sD))\rightarrow \W(S(\Hor_{b}^{-}))\otimes\W(S(\Hor_{c}^{-})) \\
\iota^{0}:\W(S(\sD))\rightarrow \W(S(\Hor_{b}^{0}))\otimes\W(S(\Hor_{c}^{0})) \; .
\end{align*}
\end{thm}

As a prelude to the next subsection, we note that if the linear functional $\omega:\W(S(\sB))\otimes\W(S(\sC))\rightarrow\mathbb{C}$ is an algebraic state, the isometric $\ast$-homomorphism $\iota$ constructed in theorem \ref{algebra-horizon-isomorphism} gives rise to a state $\omega_{\M}:\W(S(\M))\rightarrow\mathbb{C}$ defined as
\[\omega_{\M}\coloneqq \iota^{\ast}(\omega),\; \textrm{where }\iota^{\ast}(\omega)(W)=\omega(\iota(W)),\; \forall W\in S(\M)\; .\]
Specializing to quasifree states, we know that the ``quasifree property'' is preserved under pull-back and such a state is unambiguously defined on $\W(S(\sB))\otimes\W(S(\sC))$ by
\[\omega_{\M}(W_{\sB\cup\sC}(\psi))=e^{-\mu(\psi,\psi)/2},\qquad \forall\psi\in S(\sB)\oplus S(\sC) \; ,\]
where $\mu:(S(\sB)\oplus S(\sC))\times(S(\sB)\oplus S(\sC))\rightarrow\mathbb{R}$ is a real scalar product which majorizes the symplectic product.

\subsection{State}

Before we start the construction of the state, we should comment on the theorems in \citep{KayWald91} which proved that there does not exist any Hadamard state on the whole Kruskal extension of the SdS spacetime. The first nonexistence Theorem proved in section 6.3 of that reference is based on the fact that, from the Domain of Dependence property, the union of the algebras defined on the horizons $\Hor_{c}^{+}$ and $\Hor_{c}^{0-}$ (we will call this algebra $\W(S_{c}^{R})$; see figure \ref{Sch-dS_conformal_diagram}) would be orthogonal to the union of the algebras defined on the horizons $\Hor_{b}^{+}$ and $\Hor_{b}^{0-}$ ($\W(S_{b}^{L})$). On the other hand, they proved that $\W(S_{c}^{R})$ is dense in the union of the algebras defined on all the horizons corresponding to the cosmological horizon at $r=r_{c}$, $\W(S_{c})$, while $\W(S_{b}^{L})$ is dense in the union of the algebras defined on all the horizons corresponding to the black hole horizon at $r=r_{b}$, $\W(S_{b})$. But $\W(S_{c})$ and $\W(S_{b})$, again from the Domain of Dependence property, cannot be orthogonal, thus there is a contradiction.

We avoid this problem by not defining the state in the orange regions in figure \ref{Sch-dS_conformal_diagram}. The algebras $\W(S(\sB))$ and $\W(S(\sC))$ are not orthogonal, the same being valid for $\W(S(\Hor_{b}^{-}))$ and $\W(S(\Hor_{c}^{-}))$. The algebras $\W(S(\Hor_{b}^{+}))$ and $\W(S(\Hor_{c}^{+}))$ are indeed orthogonal, but they are not dense in $\W(S(\sB))$ and $\W(S(\sC))$, thus there is no contradiction in our case.

The second nonexistence Theorem proved there arrives again at a contradiction by using properties of a KMS state. As it will be clear below, the state we will construct here is not a KMS state, thus we are not troubled by the contradiction they arrive at.

Now, we will go on with the construction of our state.

On the set of complex, compactly supported smooth functions $\cC_{0}^{\infty}(\sB;\mathbb{C})$, we define its completion $\overline{\left(\cC_{0}^{\infty}(\sB;\mathbb{C}),\lambda\right)}$ in the norm defined by the scalar product \citep{Moretti08}
\begin{equation}
\lambda(\psi_{1},\psi_{2})\coloneqq \lim_{\epsilon\rightarrow 0^{+}}-\frac{r_{b}^{2}}{\pi}\int_{\mathbb{R}\times\mathbb{R}\times\mathbb{S}^{2}}\frac{\overline{\psi_{1}(U_{b1},\theta,\varphi)}\psi_{2}(U_{b2},\theta,\varphi)}{(U_{b1}-U_{b2}-i\epsilon)^{2}}dU_{b1}\wedge dU_{b2}\wedge d\mathbb{S}^{2} \; .
\label{product-U_b}
\end{equation}
Thus, $\overline{\left(\cC_{0}^{\infty}(\sB;\mathbb{C}),\lambda\right)}$ is a Hilbert space.

The $U_{b}$-Fourier-Plancherel transform\footnote{The Fourier-Plancherel transform is the unique extension of the Fourier transform to a map of $L^{2}(\mathbb{R}^{n})$ onto $L^{2}(\mathbb{R}^{n})$ \citep{ReedSimon-II,Yosida95}.} of $\psi$ is given by (we denote $(\theta,\varphi)$ by $\omega$)
\begin{equation}
\mathcal{F}(\psi)(K,\omega)\coloneqq \frac{1}{\sqrt{2\pi}}\int_{\mathbb{R}}e^{iKU_{b}}\psi(U_{b},\omega)dU_{b} \eqqcolon \hat{\psi}(K,\omega) \; .
\end{equation}

We can, more conveniently, write the scalar product \eqref{product-U_b} in the Fourier space:
\begin{equation}
\lambda(\psi_{1},\psi_{2})=\int_{\mathbb{R}\times\mathbb{S}^{2}}\overline{\hat{\psi}_{1}(K,\omega)}\hat{\psi}_{2}(K,\omega)2KdK\wedge r_{b}^{2}d\mathbb{S}^{2} \; .
\label{product-Fourier}
\end{equation}

Let $\hat{\psi}_{+}(K,\omega)\coloneqq \mathcal{F}(\psi)(K,\omega)_{\upharpoonright\{K\ge0\}}$. Then, the linear map
\begin{equation}
\cC_{0}^{\infty}(\sB;\mathbb{C})\ni\psi \mapsto \hat{\psi}_{+}(K,\omega)\in L^{2}\left(\mathbb{R}_{+}\times\mathbb{S}^{2},2KdK\wedge r_{b}^{2}d\mathbb{S}^{2}\right) \eqqcolon H_{\sB}
\label{Fourier_C-H}
\end{equation}
is isometric and uniquely extends, by linearity and continuity, to a Hilbert space isomorphism of
\begin{equation}
F_{U_{b}}:\overline{\left(\cC_{0}^{\infty}(\sB;\mathbb{C}),\lambda\right)} \rightarrow H_{\sB} \; .
\label{Fourier_C-H_isomorphism}
\end{equation}
One can similarly define the $(V_{b},V_{c},U_{c})$-Fourier-Plancherel transforms acting on the spaces of complex, compactly supported smooth functions restricted to the hypersurfaces $\Hor_{b}^{0}$, $\sC$ and $\Hor_{c}^{0}$ respectively, all completed in norms like \eqref{product-U_b}, and extend the transforms to Hilbert space isomorphisms.

Not every solution of the Klein-Gordon equation belonging to the space $S(\sB)$ (or any other of the spaces defined in \eqref{S_Hor_b-}-\eqref{S_Hor_c-}) is compactly supported, but we can still form isomorphisms between the completion of each of these spaces (in the norm $\lambda$ defined above) and the corresponding Hilbert space, as in \eqref{Fourier_C-H} and \eqref{Fourier_C-H_isomorphism}. Firstly, we note that the decay estimates found in \citep{DafermosRodnianski07} and presented at the definition of $S(\sB)$, together with smoothness of the functions in this space, let us conclude that these functions (and their derivatives) are square integrable in the measure $dU_{b}$, hence we can apply the Fourier-Plancherel transform to these functions. Thus, the product \eqref{product-Fourier} gives, for $\psi_{1},\psi_{2}\in S(\sB)$
\begin{align}
&\left|\lambda(\psi_{1},\psi_{2})\right|= \nonumber \\
&\left|\int_{\mathbb{R}\times\mathbb{S}^{2}}\overline{\hat{\psi}_{1}(K,\omega)}\hat{\psi}_{2}(K,\omega)2KdK\wedge r_{b}^{2}d\mathbb{S}^{2}\right|=2\left|\int_{\mathbb{R}\times\mathbb{S}^{2}}\overline{\hat{\psi}_{1}(K,\omega)}\Big(K\hat{\psi}_{2}(K,\omega)\Big)dK\wedge r_{b}^{2}d\mathbb{S}^{2}\right|= \nonumber \\
&2\left|\int_{\mathbb{R}\times\mathbb{S}^{2}}\overline{\hat{\psi}_{1}(K,\omega)}\widehat{\partial_{U_{b}}\psi}_{2}(K,\omega)dK\wedge r_{b}^{2}d\mathbb{S}^{2}\right|=2\left|\int_{\mathbb{R}\times\mathbb{S}^{2}}\overline{\psi_{1}(U_{b},\omega)}\partial_{U_{b}}\psi_{2}(U_{b},\omega)dU_{b}\wedge r_{b}^{2}d\mathbb{S}^{2}\right|<\infty \; .
\label{product-sB-integrable}
\end{align}

Let again $\hat{\psi}_{+}(K,\omega)\coloneqq \mathcal{F}(\psi)(K,\omega)_{\upharpoonright\{K\ge0\}}$, but now $\psi\in S(\sB)$. Then, the linear map
\begin{equation}
S(\sB)\ni\psi \mapsto \hat{\psi}_{+}(K,\omega)\in L^{2}\left(\mathbb{R}_{+}\times\mathbb{S}^{2},2KdK\wedge r_{b}^{2}d\mathbb{S}^{2}\right) \eqqcolon H_{\sB}
\label{Fourier_sB-H}
\end{equation}
is isometric and uniquely extends, by linearity and continuity, to a Hilbert space isomorphism of
\begin{equation}
F_{U_{b}}:\overline{\left(S(\sB),\lambda\right)} \rightarrow H_{\sB} \; ,
\label{Fourier_sB-H_isomorphism}
\end{equation}
and similarly for the horizon $\sC$. We then define the real-linear map $K_{\sB}$ as
\begin{equation}
F_{U_{b}}\eqqcolon K_{\sB}:\overline{\left(S(\sB),\lambda\right)} \rightarrow H_{\sB} \; .
\label{linear_map}
\end{equation}

When proving some properties of the state individuated by the two-point function \eqref{product-U_b} (Theorem \ref{state-existence} below), it will be convenient to analyse the restrictions of such two-point function to $\Hor_{b}^{\pm}$. The initial point of this analysis is the following

\begin{prop}
Let the natural coordinates covering $\Hor_{b}^{+}$ and $\Hor_{b}^{-}$ be $u\coloneqq (1/\kappa_{b})\ln(\kappa_{b}U_{b})$ and $u\coloneqq (-1/\kappa_{b})\ln(-\kappa_{b}U_{b})$, respectively. Let also $\mu(k)$ be the positive measure on $\mathbb{R}$, given by
\begin{equation}
d\mu(k)\equiv \frac{1}{2}\frac{ke^{\pi k/\kappa_{b}}}{e^{\pi k/\kappa_{b}}-e^{-\pi k/\kappa_{b}}}dk \; .
\label{measure-Fourier_u}
\end{equation}
Then, if $\widetilde{\psi}=(\mathcal{F}(\psi))(k,\omega)$ denotes the Fourier transform of either $\psi\in S(\Hor_{b}^{+})$ or $\psi\in S(\Hor_{b}^{-})$ with respect to $u$, then the maps
\begin{equation}
S(\Hor_{b}^{\pm})\ni\psi \mapsto \widetilde{\psi}(k,\omega)\in L^{2}\left(\mathbb{R}\times\mathbb{S}^{2},d\mu(k)\wedge r_{b}^{2}d\mathbb{S}^{2}\right) \eqqcolon H_{\Hor_{b}^{\pm}}
\label{Fourier_Hor_b+-}
\end{equation}
are isometric (when $S(\Hor_{b}^{\pm})$ are equipped with the scalar product $\lambda$) and uniquely extend, by linearity and continuity, to Hilbert space isomorphisms of
\begin{equation}
F_{u}^{(\pm)}:\overline{\left(S(\Hor_{b}^{\pm}),\lambda\right)} \rightarrow H_{\Hor_{b}^{\pm}} \; .
\label{Fourier_Hor_b-H_isomorphism}
\end{equation}
\end{prop}

\begin{proof}
The measure \eqref{measure-Fourier_u} is obtained if one starts from \eqref{product-U_b}, makes the change of variables from $U_{b}$ to $u$ and then takes the Fourier transform with respect to $u$, keeping in mind that $\lim_{\epsilon\rightarrow 0^{+}}1/(x-i\epsilon)^{2}=1/x^{2}-i\pi\delta'(x)$ \citep{DuistermaatKolk10}. The other statements of the proposition follow exactly as the corresponding ones for the Fourier-Plancherel transform of $\psi\in S(\sB)$. The only formal difference is that, from the decay estimate \eqref{S_Hor_b-}, we now employ the usual Fourier transform.
\end{proof}

Hence, the real-linear map $K_{\Hor_{b}^{\pm}}^{\beta_{b}}$ is defined as $K_{\Hor_{b}^{\pm}}^{\beta_{b}} \coloneqq F_{u}^{(\pm)}$.

We will now state and prove a theorem that shows that we can construct a quasifree pure state $\omega_{\M}$ on the Weyl algebra defined on the region $\sB\cup\sC$. On the next subsection we will show that $\omega_{\M}$ satisfies the Hadamard condition.

\begin{thm}\label{state-existence}\ \

(a)\ \ The pair $(H_{\sB},K_{\sB})$ is the one-particle structure for a quasifree pure state $\omega_{\sB}$ on $\W(S(\sB))$ uniquely individuated by the requirement that its two-point function coincides with the rhs of
\[\lambda(\psi_{1},\psi_{2})\coloneqq \lim_{\epsilon\rightarrow 0^{+}}-\frac{r_{b}^{2}}{\pi}\int_{\mathbb{R}\times\mathbb{R}\times\mathbb{S}^{2}}\frac{\overline{\psi_{1}(U_{b1},\theta,\varphi)}\psi_{2}(U_{b2},\theta,\varphi)}{(U_{b1}-U_{b2}-i\epsilon)^{2}}dU_{b1}\wedge dU_{b2}\wedge d\mathbb{S}^{2} \; .\]

(b)\ \ The state $\omega_{\sB}$ is invariant under the action of the one-parameter group of $\ast$-automorphisms generated by $X_{\upharpoonright\sB}$ and of those generated by the Killing vectors of $\mathbb{S}^{2}$.

(c)\ \ The restriction of $\omega_{\sB}$ to $\W(S(\Hor_{b}^{\pm}))$ is a quasifree state $\omega_{\Hor_{b}^{\pm}}^{\beta_{b}}$ individuated by the one-particle structure $\left(H_{\Hor_{b}^{\pm}}^{\beta_{b}},K_{\Hor_{b}^{\pm}}^{\beta_{b}}\right)$ with:
\[H_{\Hor_{b}^{\pm}}^{\beta_{b}}\coloneqq L^{2}\left(\mathbb{R}\times\mathbb{S}^{2},d\mu(k)\wedge r_{b}^{2}d\mathbb{S}^{2}\right) \quad \textrm{and} \quad K_{\Hor_{b}^{\pm}}^{\beta_{b}}={F_{u}^{\pm}}_{\upharpoonright S(\Hor_{b}^{\pm})} \; .\]

(d)\ \ If $\{\beta_{\tau}^{(X)}\}_{\tau\in\mathbb{R}}$ denotes the pull-back action on $S(\Hor_{b}^{-})$ of the one-parameter group generated by $X_{\upharpoonright\sB}$, that is $\left(\beta_{\tau}(\psi)\right)(u,\theta,\varphi)=\psi(u-\tau,\theta,\varphi),\forall\tau\in\mathbb{R},\psi\in S(\Hor_{b}^{-})$, then it holds:
\[K_{\Hor_{b}^{-}}^{\beta_{b}}\beta_{\tau}^{(X)}(\psi)=e^{i\tau\hat{k}}K_{\Hor_{b}^{-}}^{\beta_{b}}\psi\]
where $\hat{k}$ is the k-multiplicative self-adjoint operator on $L^{2}\left(\mathbb{R}\times\mathbb{S}^{2},d\mu(k)\wedge d\mathbb{S}^{2}\right)$. An analogous statement holds for $\Hor_{b}^{+}$.

(e)\ \ The states $\omega_{\Hor_{b}^{\pm}}^{\beta_{b}}$ satisfy the KMS condition with respect to the one-parameter group of $\ast$-automorphisms generated by, respectively, $\mp X_{\upharpoonright\sB}$, with Hawking's inverse temperature $\beta_{b}=\frac{2\pi}{\kappa_{b}}$.
\end{thm}

One can equally define a quasifree pure KMS state $\omega_{\sC}^{\beta_{c}}$ on $S(\sC)$, at inverse temperature $\beta_{c}=\frac{2\pi}{\kappa_{c}}$.
\begin{proof}
(a)\ \ Recall the definition of one-particle structure given right below Proposition \ref{one-particle_structure}. The map $K_{\sB}$, as defined in \eqref{linear_map}, is a real-linear map which satisfies $\overline{K_{\sB}S(\sB)}=H_{\sB}$. Therefore, we only need to show that $K_{\sB}$ satisfies the other hypotheses of that Proposition. First,
\[K_{\sB}:S(\sB)\ni\psi\mapsto K_{\sB}\psi\in H_{\sB}\]
and
\[\lambda(\psi_{1},\psi_{2})=\langle K_{\sB}\psi_{1},K_{\sB}\psi_{2} \rangle_{H_{\sB}} \; .\]
The symmetric part of this two-point function is given by
\[\mu_{\sB}(\psi_{1},\psi_{2})=\textrm{Re}\langle K_{\sB}\psi_{1},K_{\sB}\psi_{2} \rangle_{H_{\sB}} \; .\]

We need to check that $\mu_{\sB}$ majorizes the symplectic form. Since
\[\sigma_{\sB}(\psi_{1},\psi_{2})=-2\textrm{Im}\langle K_{\sB}\psi_{1},K_{\sB}\psi_{2} \rangle_{H_{\sB}} \; ,\]
we have
\begin{align*}
\lvert\sigma_{\sB}(\psi_{1},\psi_{2})\rvert^{2} &=4\lvert\textrm{Im}\langle K_{\sB}\psi_{1},K_{\sB}\psi_{2} \rangle_{H_{\sB}}\rvert^{2}\le 4\lvert\langle K_{\sB}\psi_{1},K_{\sB}\psi_{2} \rangle_{H_{\sB}}\rvert^{2} \\
&\le 4\langle K_{\sB}\psi_{1},K_{\sB}\psi_{1} \rangle_{H_{\sB}}\langle K_{\sB}\psi_{2},K_{\sB}\psi_{2} \rangle_{H_{\sB}}=4\mu_{\sB}(\psi_{1},\psi_{1})\mu_{\sB}(\psi_{2},\psi_{2}) \; .
\end{align*}
We thus proved that $(H_{\sB},K_{\sB})$ is the one-particle structure associated to the state $\omega_{\sB}$. Since $\overline{K_{\sB}S(\sB)}=H_{\sB}$, this state is pure.

(b)\ \ On $\sB$, defining
\begin{align*}
u &\coloneqq \frac{1}{\kappa_{\sB}}\ln(\kappa_{\sB}U_{b}) \,\textrm{on}\, \Hor_{b}^{+} \; , \\
u &\coloneqq -\frac{1}{\kappa_{\sB}}\ln(-\kappa_{\sB}U_{b}) \,\textrm{on}\, \Hor_{b}^{-} \; ,
\end{align*}
we have
\[{\partial_{t}}_{\upharpoonright\Hor_{b}^{-}}=\partial_{u}=-\kappa_{\sB}U_{b}\partial_{U_{b}}\]
(on $\Hor_{b}^{+}$, the future-pointing Killing vector is $-\partial_{t}=-\partial_{u}=-\kappa_{\sB}U_{b}\partial_{U_{b}}$).

The one-parameter group of symplectomorphisms $\beta_{\tau}^{(X)}$ generated by $X$ individuates $\beta_{\tau}^{(X)}(\psi)\in S(\sB)$ such that $\beta_{\tau}^{(X)}(\psi)(U_{b},\theta,\varphi)=\psi(e^{\kappa_{b}\tau}U_{b},\theta,\varphi)$. Since $\beta_{\tau}^{(X)}$ preserves the symplectic form $\sigma_{\sB}$, there must be a representation $\alpha^{(X)}$ of $\beta_{\tau}^{(X)}$ in terms of $\ast$-automorphisms of $\W(S(\sB))$. From the definition of $K_{\sB}$, one has
\begin{align*}
K_{\sB}(\beta_{\tau}^{(X)}(\psi))(K,\omega)&=\frac{1}{\sqrt{2\pi}}\int_{\mathbb{R}}e^{iKU_{b}}\psi(e^{\kappa_{b}\tau}U_{b},\omega)dU_{b} \\
&=e^{-\kappa_{b}\tau}\frac{1}{\sqrt{2\pi}}\int_{\mathbb{R}}e^{i(Ke^{-\kappa_{b}\tau})U'}\psi(U',\omega)dU'=e^{-\kappa_{b}\tau}\hat{\psi}(e^{-\kappa_{b}\tau}K,\omega) \; .
\end{align*}

One then has $K_{\sB}(\beta_{\tau}^{(X)}(\psi))(K,\omega)\eqqcolon (U_{\tau}^{(X)}\psi)(K,\omega)= e^{-\kappa_{b}\tau}K_{\sB}(\psi)(e^{-\kappa_{b}\tau}K,\omega)$, $\forall\psi\in S(\sB)$. Thus,
\begin{align*}
&\langle K_{\sB}(\beta_{\tau}^{(X)}(\psi_{1})),K_{\sB}(\beta_{\tau}^{(X)}(\psi_{2})) \rangle_{H_{\sB}}=\int_{\mathbb{R}\times\mathbb{S}^{2}}e^{-\kappa_{b}\tau}\overline{\hat{\psi}_{1}(e^{-\kappa_{b}\tau}K,\omega)}e^{-\kappa_{b}\tau}\hat{\psi}_{2}(e^{-\kappa_{b}\tau}K,\omega)2KdK\wedge r_{b}^{2}d\mathbb{S}^{2}= \\
&\int_{\mathbb{R}\times\mathbb{S}^{2}}\overline{\hat{\psi}_{1}(e^{-\kappa_{b}\tau}K,\omega)}\hat{\psi}_{2}(e^{-\kappa_{b}\tau}K,\omega)2\left(e^{-\kappa_{b}\tau}K\right)d\left(e^{-\kappa_{b}\tau}K\right)\wedge r_{b}^{2}d\mathbb{S}^{2}=\langle K_{\sB}\psi_{1},K_{\sB}\psi_{2} \rangle_{H_{\sB}} \; ,
\end{align*}
hence $U_{\tau}^{(X)}$ is an isometry of $L^{2}\left(\mathbb{R}\times\mathbb{S}^{2},2KdK\wedge r_{b}^{2}d\mathbb{S}^{2}\right)$. In view of the definition of $\omega_{\sB}$, it yields that $\omega_{\sB}(W_{\sB}(\beta_{\tau}^{(X)}(\psi)))=\omega_{\sB}(W_{\sB}(\psi))$ $\forall\psi\in S(\sB)$, and, per continuity and linearity, this suffices to conclude that $\omega_{\sB}$ is invariant under the action of the group of $\ast$-automorphisms $\alpha^{(X)}$ induced by $X$. The proof for the Killing vectors of $\mathbb{S}^{2}$ is similar.

(c)\ \ We only consider $\Hor_{b}^{+}$, the other case being analogous. The state $\omega_{\Hor_{b}^{+}}^{\beta_{b}}$, which is the restriction of $\omega_{\sB}$ to $\W(S(\Hor_{b}^{+}))$, is individuated by
\[\omega_{\Hor_{b}^{+}}^{\beta_{b}}(W_{\Hor_{b}^{+}}(\psi))=e^{-\mu_{\Hor_{b}^{+}}(\psi,\psi)/2} \quad , \quad \textrm{for } \psi\in S(\Hor_{b}^{+}) \; .\]
Then, if $\psi,\psi'\in S(\Hor_{b}^{+})$, the symmetric part of $\lambda$ is given by
\[\mu_{\Hor_{b}^{+}}(\psi,\psi')=\textrm{Re}\lambda(\psi,\psi')=\textrm{Re}\langle F_{u}^{(+)}\psi,F_{u}^{(+)}\psi' \rangle_{H_{\Hor_{b}^{+}}^{\beta_{b}}}=\textrm{Re}\langle K_{\Hor_{b}^{+}}^{\beta_{b}}\psi,K_{\Hor_{b}^{+}}^{\beta_{b}}\psi' \rangle_{H_{\Hor_{b}^{+}}^{\beta_{b}}} \; .\]
It is immediate that
\[\sigma_{\Hor_{b}^{+}}(\psi,\psi')=-2\textrm{Im}\langle K_{\Hor_{b}^{+}}^{\beta_{b}}\psi,K_{\Hor_{b}^{+}}^{\beta_{b}}\psi' \rangle_{H_{\Hor_{b}^{+}}^{\beta_{b}}} \; .\]
Therefore,
\[|\sigma_{\Hor_{b}^{+}}(\psi,\psi')|^{2}\leq 4\mu_{\Hor_{b}^{+}}(\psi,\psi)\mu_{\Hor_{b}^{+}}(\psi',\psi') \; .\]
This, and the fact that $K_{\Hor_{b}^{+}}^{\beta_{b}}$ is a real-linear map which satisfies $\overline{K_{\Hor_{b}^{+}}^{\beta_{b}}S(\Hor_{b}^{+})}=H_{\Hor_{b}^{+}}^{\beta_{b}}$, suffice to conclude that $(H_{\Hor_{b}^{+}}^{\beta_{b}},K_{\Hor_{b}^{+}}^{\beta_{b}})$ is the one-particle structure of the quasifree pure state $\omega_{\Hor_{b}^{+}}^{\beta_{b}}$ (a completely analogous statement is valid for the state $\omega_{\Hor_{b}^{-}}^{\beta_{b}}$).

(d)\ \ In $S(\Hor_{b}^{-})$, the natural action of the one-parameter group of isometries generated by $X_{\upharpoonright\Hor_{b}^{-}}$ is $\beta_{\tau}^{(X)}:\psi\mapsto\beta_{\tau}^{(X)}(\psi)$ with $\beta_{\tau}^{(X)}(\psi)(u,\theta,\varphi)\coloneqq \psi(u-\tau,\theta,\varphi)$, for all $u,\tau\in\mathbb{R}$, $(\theta,\varphi)\in\mathbb{S}^{2}$ and for every $\psi\in S(\Hor_{b}^{-})$. As previously, this is an obvious consequence of $X=\partial_{u}$ on $\Hor_{b}^{-}$. Since $\beta^{(X)}$ preserves the symplectic form $\sigma_{\Hor_{b}^{-}}$, there must be a representation $\alpha^{(X)}$ of $\beta^{(X)}$ in terms of $\ast$-automorphisms of $W(S(\Hor_{b}^{-}))$. Let us prove that $\alpha^{(X)}$ is unitarily implemented in the GNS representation of $\omega_{\Hor_{b}^{-}}^{\beta_{b}}$. To this end, we notice that $\beta$ is unitarily implemented in $H_{\Hor_{b}^{-}}$, the one-particle space of $\omega_{\Hor_{b}^{-}}^{\beta_{b}}$, out of the strongly-continuous one-parameter group of unitary operators $V_{\tau}$ such that $(V_{\tau}\widetilde{\psi})(k,\theta,\varphi)=e^{ik\tau}\widetilde{\psi}(k,\theta,\varphi)$. This describes the time-displacement with respect to the Killing vector $\partial_{u}$. Thus the self-adjoint generator of $V$ is $h:\textrm{Dom}(\hat{k})\subset L^{2}\left(\mathbb{R}\times\mathbb{S}^{2},d\mu(k)\wedge r_{b}^{2}d\mathbb{S}^{2}\right)\rightarrow L^{2}\left(\mathbb{R}\times\mathbb{S}^{2},d\mu(k)\wedge r_{b}^{2}d\mathbb{S}^{2}\right)$ with $\hat{k}(\phi)(k,\theta,\varphi)=k\phi(k,\theta,\varphi)$ and
\[\textrm{Dom}(\hat{k}) \coloneqq \left\{\phi \in L^{2}\left(\mathbb{R}\times\mathbb{S}^{2},d\mu(k)\wedge r_{b}^{2}d\mathbb{S}^{2}\right)\, \Big|\, \int_{\mathbb{R}\times\mathbb{S}^{2}}|k\phi(k,\theta,\varphi)|^{2}d\mu(k)\wedge r_{b}^{2}d\mathbb{S}^{2}<+\infty\right\} \; .\]
Per direct inspection, if one employs the found form for $V$ and exploits
\[\omega_{\Hor_{b}^{-}}^{\beta_{b}}(W_{\Hor_{b}^{-}}(\psi))=e^{-\frac{1}{2}\langle \widetilde{\psi},\widetilde{\psi} \rangle_{L^{2}\left(\mathbb{R}\times\mathbb{S}^{2},d\mu(k)\wedge r_{b}^{2}d\mathbb{S}^{2}\right)}} \; ,\]
one sees, by the same argument as in the proof of item c) above, that $\omega_{\Hor_{b}^{-}}^{\beta_{b}}$ is invariant under $\alpha^{(X)}$, so that it must admit a unitary implementation.

(e)\ \ We will prove this statement by explicitly calculating the two-point function and verifying that it satisfies the KMS condition. Let $\psi,\psi'\in S(\Hor_{b}^{-})$. Since these are real functions, $\overline{\tilde{\psi}(k,\theta,\varphi)}=\tilde{\psi}(-k,\theta,\varphi)$. Then
\begin{align}
\lambda(\beta_{\tau}^{(X)}(\psi),\psi')&=\langle e^{i\tau\hat{k}}K_{\Hor_{b}^{-}}^{\beta_{b}}\psi,K_{\Hor_{b}^{-}}^{\beta_{b}}\psi' \rangle_{H_{\Hor_{b}^{-}}^{\beta_{b}}} \nonumber \\
&=\frac{r_{b}^{2}}{2}\int_{\mathbb{R}\times\mathbb{S}^{2}} e^{-i\tau k}\overline{\tilde{\psi}(k,\theta,\varphi)}\tilde{\psi}'(k,\theta,\varphi)\frac{ke^{\pi k/\kappa_{b}}}{e^{\pi k/\kappa_{b}}-e^{-\pi k/\kappa_{b}}}dk\wedge d\mathbb{S}^{2} \nonumber \\
&=\frac{r_{b}^{2}}{2}\int_{\mathbb{R}\times\mathbb{S}^{2}} e^{-i\tau k}\overline{\tilde{\psi}'(-k,\theta,\varphi)}\tilde{\psi}(-k,\theta,\varphi)\frac{ke^{\pi k/\kappa_{b}}}{e^{\pi k/\kappa_{b}}-e^{-\pi k/\kappa_{b}}}dk\wedge d\mathbb{S}^{2} \nonumber \\
&\overset{k\rightarrow -k}{=}\frac{r_{b}^{2}}{2}\int_{\mathbb{R}\times\mathbb{S}^{2}} \overline{\tilde{\psi}'(k,\theta,\varphi)}e^{i\tau k}\tilde{\psi}(k,\theta,\varphi)\frac{ke^{-\pi k/\kappa_{b}}}{e^{\pi k/\kappa_{b}}-e^{-\pi k/\kappa_{b}}}dk\wedge d\mathbb{S}^{2} \nonumber \\
&=\frac{r_{b}^{2}}{2}\int_{\mathbb{R}\times\mathbb{S}^{2}} \overline{\tilde{\psi}'(k,\theta,\varphi)}e^{-2\pi k/\kappa_{b}}e^{i\tau k}\tilde{\psi}(k,\theta,\varphi)\frac{ke^{\pi k/\kappa_{b}}}{e^{\pi k/\kappa_{b}}-e^{-\pi k/\kappa_{b}}}dk\wedge d\mathbb{S}^{2} \nonumber \\
&=\langle K_{\Hor_{b}^{-}}^{\beta_{b}}\psi',e^{i\hat{k}(\tau+2\pi i/\kappa_{b})}K_{\Hor_{b}^{-}}^{\beta_{b}}\psi \rangle_{H_{\Hor_{b}^{-}}^{\beta_{b}}}=\lambda(\psi',\beta_{\tau+i\beta_{b}}^{(X)}(\psi))
\end{align}
\end{proof}

We have successfully applied the bulk-to-boundary technique to construct two quasifree pure KMS states, one on $\W(S(\sB))$ and the other one on $\W(S(\sC))$, where the temperatures are given by $\kappa_{b}$ and $\kappa_{c}$, respectively. Thus, by the remarks after theorems \ref{algebra-horizon-isomorphism} and \ref{state-homomorphism}, we can define a state on $\M$ such that, for $\psi\in S(\M)$,
\begin{align}
\omega_{\M}(W_{\M}(\psi))&=e^{-\mu(\psi,\psi)}=e^{-\mu_{\sB}(\psi_{\upharpoonright\sB},\psi_{\upharpoonright\sB})-\mu_{\sC}(\psi_{\upharpoonright\sC},\psi_{\upharpoonright\sC})}=e^{-\mu_{\sB}(\psi_{\upharpoonright\sB},\psi_{\upharpoonright\sB})}e^{-\mu_{\sC}(\psi_{\upharpoonright\sC},\psi_{\upharpoonright\sC})} \nonumber \\
&=\omega_{\sB}(W_{\sB}(\psi_{\upharpoonright\sB}))\omega_{\sC}(W_{\sC}(\psi_{\upharpoonright\sC})) \; .
\label{state}
\end{align}

The resulting state is thus the tensor product of two states, each a quasifree pure state, but each KMS at a different temperature. Thus $\omega_{\M}$ is not a KMS state, and neither can it be interpreted as a superposition, a mixture or as an entangled state.

Still, we must prove that the two-point function of this state is a bidistribution in $\left(\cC^{\infty}_{0}\right)^{'}(\M\times\M)$. This will be easily proved in the following
\begin{prop}\label{prop_M=B+C}
The smeared two-point function $\Lambda_{\M}:\cC_{0}^{\infty}\left(\M;\mathbb{R}\right)\times\cC_{0}^{\infty}\left(\M;\mathbb{R}\right) \rightarrow \mathbb{C}$ of the state $\omega_{\M}$ can be written as the sum
\begin{equation}
\Lambda_{\M}=\Lambda_{\sB}+\Lambda_{\sC} \; ,
\label{2ptfcn_M=B+C}
\end{equation}
with $\Lambda_{\sB}$ and $\Lambda_{\sC}$  defined from the following relations as in \eqref{product-U_b},
\[\Lambda_{\sB}(f,h)=\lambda_{\sB}(\psi^{f}_{\sB},\psi^{h}_{\sB}) \quad ; \quad \Lambda_{\sC}(f,h)=\lambda_{\sC}(\psi^{f}_{\sC},\psi^{h}_{\sC})\]
for every $f,h\in\cC_{0}^{\infty}\left(\M;\mathbb{R}\right)$.

Separately, $\Lambda_{\sB}$, $\Lambda_{\sC}$ and $\Lambda_{\M}$ individuate elements of $\left(\cC^{\infty}_{0}\right)^{'}(\M\times\M)$ that we will denote, respectively, by the same symbols. These are uniquely individuated by complex linearity and continuity under the assumption \eqref{2ptfcn_M=B+C}, as
\begin{equation}
\Lambda_{\sB}(f\otimes h)\coloneqq \lambda_{\sB}(\psi^{f}_{\sB},\psi^{h}_{\sB}) \quad ; \quad \Lambda_{\sC}(f\otimes h)\coloneqq \lambda_{\sC}(\psi^{f}_{\sC},\psi^{h}_{\sC}) \; ,
\label{bidistr_M=B+C}
\end{equation}
for every $f,h\in\cC_{0}^{\infty}\left(\M;\mathbb{R}\right)$. Here, $\psi^{f}_{\sB}$ is a ``smeared solution'', $\psi^{f}_{\sB}=\left(\mathds{E}(f)\right)_{\upharpoonright\sB}$ (similarly for the other solutions).
\end{prop}

\begin{proof}
The first statement follows trivially from the definition \eqref{state}, theorems \eqref{algebra-horizon-isomorphism} and \eqref{state-homomorphism} and the remarks at the end of section \ref{sec-Sch-dS-algebra}.

To prove the second statement, we have to prove that $\Lambda_{\sB}$ and $\Lambda_{\sC}$ are bidistributions in $\left(\cC^{\infty}_{0}\right)^{'}(\M\times\M)$. For this purpose, we note that
\[f\mapsto\Lambda_{i}(f,\cdot) \quad \textrm{and} \quad h\mapsto\Lambda_{i}(\cdot,h) \qquad ; \quad i=\sB,\sC \; ,\]
are continuous in the sense of distributions. This is true from the definition of $\lambda_{i}(\cdot,\cdot)$ and the fact that the Fourier-Plancherel transform is a continuous map. Thus, both $\Lambda_{i}(f,\cdot)$ and $\Lambda_{i}(\cdot,h)$ are in $\left(\cC^{\infty}_{0}\right)^{'}(\M)$. The Schwarz kernel theorem \citep{Hormander-I} shows that $\Lambda_{i}\in\left(\cC^{\infty}_{0}\right)^{'}(\M\times\M)$.
\end{proof}

Before we proceed to the proof that $\omega_{\M}$ is a Hadamard state, we have to clarify its interpretation. The fact that our state is not defined in regions $III$ and $IV$ of figure \ref{Sch-dS_conformal_diagram} makes $\omega_{\sB}$ very similar to the Unruh state defined in the Schwarzschild spacetime. Also the fact that $\omega_{\M}$ is Hadamard (see next section) on $\Hor_{b}^{0}$, but not on $\Hor_{b}^{\pm}$, as in the Schwarzschild case \citep{DappiaggiMorettiPinamonti09}, reinforces this similarity. But since neither is $\omega_{\M}$ defined in regions $III'$ and $IV'$, and neither is it Hadamard on $\Hor_{c}^{\pm}$, although it is Hadamard on $\Hor_{c}^{0}$, $\omega_{\sC}$ is not similar to the Unruh state in de Sitter spacetime. As shown in \citep{NarnhoferPeterThirring96}, the Unruh state in the de Sitter spacetime is the unique KMS state which can be extended to a Hadamard state in the whole spacetime. The Unruh state in Schwarzschild-de Sitter spacetime, if it existed, should be well defined and Hadamard in $\M \cup III' \cup IV'$. But such a state cannot exist, by the nonexistence theorems proved in \citep{KayWald91}. Therefore $\omega_{\M}$ cannot be interpreted as the Unruh state in Schwarzschild-de Sitter spacetime.

\subsection{The Hadamard Condition}

We must analyse the wave front set of the bidistribution individuated in Proposition \ref{prop_M=B+C} and show that it satisfies equation \eqref{Wfcond}. The proof will be given in two parts: the first part will be devoted to prove the Hadamard condition in the region $\sD$. Here we can repeat {\it verbatim} the first part of the proof given in \citep{DappiaggiMorettiPinamonti09}, where the authors showed that the Unruh state in Schwarzschild spacetime is a Hadamard state in the wedge region. Their proof could be almost entirely repeated from \citep{SahlmannVerch00}. We will thus present the statements and the main points of the proof. The second part of the proof consists of extending these results to the regions $II$ and $II'$ (see figure \ref{Sch-dS_conformal_diagram}). This part of the proof can be repeated almost {\it verbatim} from the second part of the proof given in \citep{DappiaggiMorettiPinamonti09}, when the authors proved that their state is a Hadamard state inside the black hole region. The main differences rely on the fact that here we can apply the Fourier-Plancherel transform directly to the functions in $S(\sB)$ and in $S(\sC)$, since they are square-integrable, a fact which does not hold in \citep{DappiaggiMorettiPinamonti09}, because the decay rates of the solutions in the Schwarzschild case are slower than in our case. Besides, we do not have to handle the solutions at infinity, only on the event horizons. Thus, our proof is technically simpler than the one given in \citep{DappiaggiMorettiPinamonti09}. As a last remark, we note that the proof of the Hadamard condition given there for the region inside the black hole is equally valid for the region outside the cosmological horizon (region $II'$), in our case.

{\it Part 1:}\ \ In this first part, we will prove the following
\begin{lem}\label{lemma-omega_D_Hadamard}
The wave front set of the two-point function $\Lambda_{\M}$ of the state $\omega_{\M}$, individuated in \eqref{2ptfcn_M=B+C}, restricted to a functional on $\sD\times\sD$, is given by
\begin{equation}
WF((\Lambda_{\M})_{\upharpoonright\sD\times\sD})=\left\{\left(x_{1},k_{1};x_{2},-k_{2}\right) | \left(x_{1},k_{1};x_{2},k_{2}\right)\in {\mathcal T}^{*}\left(\sD\times\sD\right) \diagdown \{0\} ; (x_{1},k_{1})\sim (x_{2},k_{2}) ; k_{1}\in \overline{V}_{+}\right\} \; .
\label{Wf-omega_D}
\end{equation}
thus the state ${\omega_{\M}}_{\upharpoonright\sD}$ is a Hadamard state.
\end{lem}

\begin{proof}
In \citep{SahlmannVerch00} the authors proved that, given a state $\omega$, if it can be written as a convex combination\footnote{See the definition of states after definition \ref{def_Banach-star}} of ground and KMS states at an inverse temperature $\beta>0$ (those authors named such state a {\it strictly passive state}), then its two-point function satisfies the microlocal spectrum condition, thus being a Hadamard state. Unfortunately, our state $\omega_{\M}$ is not such a state, then we cannot directly apply this result. Nonetheless, as remarked in \citep{DappiaggiMorettiPinamonti09}, the passivity of the state is not an essential condition of the proof. Hence we will present here the necessary material to complete the proof, along the lines of the above cited papers, that our state $\omega_{\M}$ is a Hadamard state in the region $\sD$.

First we note that, for every $f\in\cC_{0}^{\infty}\left(\mathbb{R};\mathbb{R}\right)$ and $h_{1},h_{2}\in\cC_{0}^{\infty}\left(\sD;\mathbb{R}\right)$, $\Lambda_{\sB}$ and $\Lambda_{\sC}$ satisfy
\begin{equation}
\int_{\mathbb{R}}\hat{f}(t)\Lambda_{\sB}(h_{1}\otimes\beta_{t}^{(X)}(h_{2}))dt=\int_{\mathbb{R}}\hat{f}(t+i\beta_{b})\Lambda_{\sB}(\beta_{t}^{(X)}(h_{2})\otimes h_{1})dt
\label{Lambda-timeshift}
\end{equation}
(for $\sC$, just change $\beta_{b}\rightarrow\beta_{c}$). For these states, we can define a subset of $\mathbb{R}^{2}\diagdown\{0\}$, the {\it global asymptotic pair correlation spectrum}, in the following way: we call a family $(A_{\lambda})_{\lambda>0}$ with $A_{\lambda}\in W(S(\sD))$ a {\it global testing family} in $W(S(\sD))$ provided there is, for each continuous semi-norm $\sigma$, an $s\geq 0$ (depending on $\sigma$ and on the family) such that
\[\underset{\lambda}{\textrm{sup}}\, \lambda^{s}\sigma(A_{\lambda}^{\ast}A_{\lambda})<\infty \; .\]
The set of global testing families will be denoted by {\bf A}.

Let $\omega$ be a state on $W(S(\sD))$ and ${\bf \xi}=(\xi_{1},\xi_{2})\in\mathbb{R}^{2}\diagdown\{0\}$. Then we say that ${\bf \xi}$ is a {\it regular direction} for $\omega$, with respect to the continuous one-parametric group of $\ast$-automorphisms $\{\alpha_{t}\}_{t\in\mathbb{R}}$ induced by the action of the Killing vector field\footnote{We remind the reader that, in the region $\sD$, $X=\partial_{t}$.} $X$, if there exists some $h\in\cC_{0}^{\infty}(\mathbb{R}^{2})$ and an open neighborhood $V$ of ${\bf \xi}$ in $\mathbb{R}^{2}\diagdown\{0\}$ such that, for each $s\in\mathbb{N}_{+}$, there are $C_{s},\lambda_{s}>0$ so that
\[\underset{{\bf k}\in V}{\textrm{sup}}\left|\int e^{-i\lambda^{-1}(k_{1}t_{1}+k_{2}t_{2})}h(t_{1},t_{2})\omega\left(\alpha_{t_{1}}(A_{\lambda})\alpha_{t_{2}}(B_{\lambda})\right)dt_{1}dt_{2}\right|<C_{s}\lambda^{s} \quad \textrm{as }\lambda\rightarrow 0\]
holds for all $(A_{\lambda})_{\lambda>0},(B_{\lambda})_{\lambda>0}\in {\bf A}$, and for $0<\lambda<\lambda_{s}$.

The complement in $\mathbb{R}^{2}\diagdown\{0\}$ of the set of regular directions of $\omega$ is called the {\it global asymptotic pair correlation spectrum} of $\omega$, $ACS_{\bf A}^{2}(\omega)$.

As noted in \citep{DappiaggiMorettiPinamonti09}, the fact that the two-point functions $\Lambda_{\sB}$ and $\Lambda_{\sC}$ satisfy \eqref{Lambda-timeshift}, suffices to prove
\begin{prop}\label{prop-ACS}
Let $\omega$ be an $\{\alpha_{t}\}_{t\in\mathbb{R}}$-invariant KMS state at inverse temperature $\beta >0$. Then,
\begin{align}
& either \quad ACS_{\bf A}^{2}(\omega)=\emptyset \; , \nonumber \\
& or \quad ACS_{\bf A}^{2}(\omega)=\left\{(\xi_{1},\xi_{2})\in{\mathcal T}^{*}\left(\sD\times\sD\right)\diagdown\{0\}\, |\, \xi_{1}(X)+\xi_{2}(X)=0\right\} \; .
\label{ACS}
\end{align}
\end{prop}
The proof of this Proposition can be found in the proof of item (2) of Proposition 2.1 in \citep{SahlmannVerch00}.

With this result, we can turn our attention to Theorem 5.1 in \citep{SahlmannVerch00}, where they prove that the wave front set of the two-point function of a strictly passive state which satisfies weakly the equations of motion\footnote{We say that a functional $F$ is a weak solution of a differential operator $P$ if, for $phi$ such that $P\phi=0$, $PF[\phi]=F[P\phi]=0$.}, in both variables, and whose symmetric and antisymmetric parts are smooth at causal separation, is contained in the rhs of \eqref{Wf-omega_D}. As further noted in \citep{DappiaggiMorettiPinamonti09}, the passivity of the state is only employed in the proof of step (2) of the mentioned Theorem. However, what is actually needed for this proof is the result of Proposition \ref{prop-ACS}. Moreover, as proved in step (3) of the mentioned Theorem, the antisymmetric part of the two-point function of the state is smooth at causal separation if and only if the symmetric part is also smooth at causal separation. The antisymmetric part of the two-point function of our state, by definition, satisfies this condition. Besides, the two-point function of our state $\omega_{\M}$ satisfies weakly the equations of motion in both variables. Therefore, with the only modification being the substitution of the passivity of the state by the result of Proposition \ref{prop-ACS}, we have proved, as the authors of \citep{DappiaggiMorettiPinamonti09} did, an adapted version of Theorem 5.1 of \citep{SahlmannVerch00}. At last, as stated in item (ii) of Remark 5.9 in \citep{SahlmannVerch01}, the wave front set of the two-point function of a state being contained in the rhs of \eqref{Wf-omega_D} implies that the wave front set is equal to this set. 
\end{proof}

{\it Part 2:}\ \ Our analisys here will be strongly based on the Propagation of Singularities Theorem (Theorem 6.1.1 in \citep{DuistermaatHormander72}), which we already mentioned in section \ref{subsec_wave-eq}. Let us recall the concept of {\it characteristics} and {\it bicharacteristics}, which will be used a lot in the following discussion: we are going to analyse some features of the solutions of a normally hyperbolic differential operator. Hence the characteristics are the bundle of null cones $\mathcal{N}_{g}\subset{\mathcal T}^{\ast}\M\diagdown\{0\}$ defined by
\[\mathcal{N}_{g}\coloneqq \left\{(x,k_{x})\in{\mathcal T}^{\ast}\M\diagdown\{0\}\, |\, g^{\mu\nu}(x)(k_{x})_{\mu}(k_{x})_{\nu}=0\right\} \; .\]
The {\it bicharacteristic strip} generated by $(x,k_{x})\in\mathcal{N}_{g}$ is given by
\[B(x,k_{x})\coloneqq \left\{(x',k_{x'})\in\mathcal{N}_{g}\, |\, (x',k_{x'})\sim (x,k_{x})\right\} \; .\]

The PST, applied to the weak bisolution $\Lambda_{\M}$ implies, on the one hand, that
\begin{equation}
WF(\Lambda_{\M})\subset\left(\{0\}\cup\mathcal{N}_{g}\right)\times\left(\{0\}\cup\mathcal{N}_{g}\right) \; ,
\label{WF_0+null}
\end{equation}
while, on the other hand,
\begin{equation}
\textrm{if }(x,y;k_{x},k_{y})\in WF(\Lambda_{\M}) \; , \; \textrm{then} \; B(x,k_{x})\times B(y,k_{y})\subset WF(\Lambda_{\M}) \; .
\label{bicharac_WF}
\end{equation}

We will now quote from \citep{DappiaggiMorettiPinamonti09} a couple of technical results which will be useful in the final proof.

The first proposition characterizes the decay properties, with respect to $p\in{\mathcal T}^{\ast}\M$, of the distributional Fourier transforms:
\[\psi^{f_{p}}_{\sB} \coloneqq \mathds{E}\left(fe^{i\langle p,\cdot \rangle}\right)_{\upharpoonright\sB} \quad ; \quad \psi^{f_{p}}_{\sC} \coloneqq \mathds{E}\left(fe^{i\langle p,\cdot \rangle}\right)_{\upharpoonright\sC} \; ,\]
where we have used the complexified version of the causal propagator, which enjoys the same causal and topological properties as those of the real one. Henceforth $\langle \cdot , \cdot \rangle$ denotes the scalar product in $\mathbb{R}^{4}$ and $|\cdot|$ the corresponding norm.

\begin{prop}\label{prop_rapid-decrease}
Let us take $(x,k_{x})\in\mathcal{N}_{g}$ such that (i) $x\in II$ (or $II'$) and (ii) the unique inextensible geodesic $\gamma$ cotangent to $k_{x}$ at $x$ intersects $\sB$ ($\sC$) in a point whose $U_{b}$ ($V_{c}$) coordinate is non-negative. Let us also fix $\chi'\in\cC_{0}^{\infty}(\sB;\mathbb{R})$ ($\chi'\in\cC_{0}^{\infty}(\sC;\mathbb{R})$) with $\chi'=1$ if $U_{b}\in\left(-\infty,U_{b_{0}}\right]$ and $\chi'=0$ if $U_{b}\in\left[U_{b_{1}},+\infty\right)$ ($\chi'=1$ if $V_{c}\in\left(-\infty,V_{c_{0}}\right]$ and $\chi'=0$ if $V_{c}\in\left[V_{c_{1}},+\infty\right)$) for constants $U_{b_{0}}<U_{b_{1}}<0$ ($V_{c_{0}}<V_{c_{1}}<0$).

For any $f\in\cC_{0}^{\infty}(\M)$ with $f(x)=1$ and sufficiently small support, $k_{x}$ is a direction of rapid decrease for both $p\mapsto\lVert\chi'\psi^{f_{p}}_{\sB}\rVert_{\sB}$ and $p\mapsto\lVert\psi^{f_{p}}_{\sC}\rVert_{\sC}$ ($p\mapsto\lVert\psi^{f_{p}}_{\sB}\rVert_{\sB}$ and $p\mapsto\lVert\chi'\psi^{f_{p}}_{\sC}\rVert_{\sC}$), where $\rVert\cdot\lVert_{\sB}$ is the norm induced by $\lambda_{\sB}$ (and similarly for $\sC$; see equations \eqref{Fourier_sB-H}-\eqref{linear_map}).
\end{prop}

\begin{proof}
The proof here is an adapted version of the proof of Proposition 4.4 of \citep{DappiaggiMorettiPinamonti09}. It consists in analysing the behavior of the constant $C_{\phi}$ appearing in \eqref{S_sB} and \eqref{S_sC} for large values of $p\in V_{k}$ ($k$ a direction of rapid decrease). The constant $C_{\phi}$ is given in \citep{DafermosRodnianski07} as a constant dependent on the geometry of the spacetime multiplied by the square root of
\begin{equation}
{\bf E}_{0}(\phi_{l},\dot{\phi}_{l})=\lVert \nabla\phi_{l} \rVert^{2}+\lVert \dot{\phi}_{l} \rVert^{2} \; ,
\end{equation}
where $\lVert \cdot \rVert$ is the Riemannian $L^{2}$ norm on $\Sigma\cap J^{-}(\sD)$ (see Figure \ref{Sch-dS_conformal_diagram}).

Now, we can choose the support of $f$ so small that every inextensible geodesic starting from $supp(f)$, with cotangent vector equal to $k_{x}$, intersects $\sB$ in a point with coordinate $U_{b}>0$ (similarly for $\sC$). Hence, we can fix $\rho\in\cC_{0}^{\infty}(\sK;\mathbb{R})$ such that (i) $\rho=1$ on $J^{-}(supp(f);\M)\cap\Sigma$ and (ii) the null geodesics emanating from $supp(f)$ with $k_{x}$ as cotangent vector do not meet the support of $\rho$. Henceforth we can proceed exactly as in the proof given in \citep{DappiaggiMorettiPinamonti09} using the properties mentioned in this paragraph, together with the compactness of the support of $f$, to coclude the proof of this proposition. We only remark that, differently from the Schwarzschild case, our Cauchy surface $\Sigma$ does not intercept the bifurcation surfaces $\cB_{b}$ and $\cB_{c}$, hence the reasoning depicted here is valid both for $x$ in $II$ ($II'$) and for $x$ on $\Hor_{b}^{0}$ ($\Hor_{c}^{0}$).
\end{proof}

The second technical result is the following Lemma, which states that 
\begin{lem}\label{lem_isolated_sing}
Isolated singularities do not enter the wave front set of $\Lambda_{\M}$, i.e.
\[(x,y;k_{x},0)\notin WF(\Lambda_{\M}) \quad ; \quad (x,y;0,k_{y})\notin WF(\Lambda_{\M})\]
\[\textrm{if }x,y,\in\M \quad ; \quad k_{x}\in{\mathcal T}^{\ast}_{x}\M \, , \, k_{y}\in{\mathcal T}^{\ast}_{y}\M \; .\]
Hence, \eqref{WF_0+null} yields
\begin{equation}
WF(\Lambda_{\M})\subset\mathcal{N}_{g}\times\mathcal{N}_{g} \; .
\label{WF_null}
\end{equation}
\end{lem}

\begin{proof}
We start by noting that the antisymmetric part of $\Lambda_{\M}$ is the advanced-minus-retarded operator $\mathds{E}$ and that the wave front set of $\mathds{E}$ contains no null covectors. Hence, $(x,y;k_{x},0)\in WF(\Lambda_{\M}) \Leftrightarrow (y,x;0,k_{x})\in WF(\Lambda_{\M})$, otherwise $WF(\mathds{E})$ would contain a null covector. Thus it suffices to analyse $(x,y;k_{x},0)\in{\mathcal T}^{\ast}\left(\M\times\M\right)\diagdown\{0\}$ and to show that it does not lie in $WF(\Lambda_{\M})$. Besides, from the proof of {\it Part 1} above, if $(x,y)\in\sD\times\sD\Rightarrow(x,y;k_{x},0)\notin WF(\Lambda_{\M})$. From the Propagation of Singularities Theorem, if there exists $(q,k_{q})\in B(x,k_{x})$ such that $q\in\sD$ $(x\notin\sD)$, then again $(x,y;k_{x},0)\notin WF(\Lambda_{\M})$.

For the case $x\in II$, $y\in\sD$ with $B(x,k_{x})\cap{\mathcal T}^{\ast}(\sD)\diagdown 0=\emptyset$, there must exist $q\in \Hor_{b}^{+}\cup\cB_{b}$ such that $(q,k_{q})\in B(x,k_{x})$. Besides, we can introduce a partition of unit with $\chi,\chi'\in\cC_{0}^{\infty}(\sB;\mathbb{R})$, $\chi+\chi'=1$ such that $\chi=1$ in a neighborhood of $q$. Hence, with the same definitions as in the Proposition above,
\begin{equation}
\Lambda_{\M}(f_{k_{x}},h)=\lambda_{\sB}(\chi\phi^{f_{k_{x}}}_{\sB},\phi^{h}_{\sB})+\lambda_{\sB}(\chi'\phi^{f_{k_{x}}}_{\sB},\phi^{h}_{\sB})+\lambda_{\sC}(\phi^{f_{k_{x}}}_{\sC},\phi^{h}_{\sC}) \; .
\label{Lambda-M_part_unity}
\end{equation}
Since all the terms in equation \eqref{Lambda-M_part_unity} are continuous with respect to the corresponding $\lambda$-norms, the second and third terms in \eqref{Lambda-M_part_unity} are dominated by $C\lVert \chi' \psi^{f_{k_{x}}}_{\sB} \rVert_{\sB} \lVert \psi^{h}_{\sB} \rVert_{\sB}$ and $C'\lVert \psi^{f_{k_{x}}}_{\sC} \rVert_{\sC} \lVert \psi^{h}_{\sC} \rVert_{\sC}$, respectively, where $C$ and $C'$ are positive constants. From Proposition \ref{prop_rapid-decrease}, we know that $\lVert \chi' \psi^{f_{k_{x}}}_{\sB} \rVert_{\sB}$ and $\lVert \psi^{f_{k_{x}}}_{\sC} \rVert_{\sC}$ are rapid decreasing terms in $k_{x}\in\mathcal{T}^{\ast}(\M)\diagdown\{0\}$ for any $f$ with sufficiently small support and for $k_{x}$ in an open conical neighborhood of any null direction. By a similar argument as the one presented in the proof of that Proposition, one can conclude that $\lVert \psi^{h}_{\sB} \rVert_{\sB}$ and $\lVert \psi^{h}_{\sC} \rVert_{\sC}$ are bounded. Hence, we need only focus our attention on the first term, $\lambda_{\sB}(\chi \psi^{f_{k_{x}}}_{\sB},\psi^{h}_{\sB})$.

Choosing again $f$ and $h$ with sufficiently small, compact support, we can choose $\chi''\in\cC_{0}^{\infty}(\sB;\mathbb{R})$ such that both $\chi''(p)=1$ for every $p\in supp(\psi^{h}_{\sB})$ and $supp(\chi)\cap supp(\chi'')=\emptyset$. We can write the $\lambda$-product as
\begin{equation}
\lambda_{\sB}(\psi^{f_{k_{x}}}_{\sB},\psi^{h}_{\sB})=\int_{\sB\times\sB}\chi(x')(\mathds{E}(f_{k_{x}}))(x')T(x',y')\chi''(y')\psi^{h}_{\sB}dU_{x'}d\mathbb{S}^{2}(\theta_{x'},\varphi_{x'})dU_{y'}d\mathbb{S}^{2}(\theta_{y'},\varphi_{y'}) \; .
\label{prod-chi_E_T_chi''}
\end{equation}
$\psi^{f_{k_{x}}}_{\sB}$ was written as $(\mathds{E}(f_{k_{x}}))(x')$ and $T(x',y')$ is the integral kernel of $\lambda_{\sB}$, viewed as a distribution in $(\cC_{0}^{\infty})'(\sB\times\sB)$. The integral kernel $\chi T \chi''(x',y')$, with one entry $x'$ restricted to the support of $\chi$, and the other $y'$, restricted to the support of $\chi''$, is always smooth. Besides, if one keeps $x'$ fixed, this kernel is dominated by a smooth function whose $H^{1}$-norm in $y'$ is, uniformly in $x'$, finite\footnote{The $H^{1}$-norm of a function $f$ is defined as \[\lVert f \rVert_{H^{1}(\Omega)}=\left(\sum_{|\alpha|\leq 1} \lVert D^{\alpha}f \rVert_{L^{2}(\Omega)}^{2}\right)^{1/2} \; ,\] where $\Omega$ is an open measurable space and $\alpha$ is a multi-index.}. Hence the $H^{1}(\sB)_{U_{b}}$-norm $\lVert (T\chi'')\circ\chi\mathds{E}f_{k_{x}}  \rVert_{H^{1}(\sB)_{U_{b}}}$ is dominated by the product of two integrals, one over $x'$ and one over $y'$. Since $\chi$ is a compactly supported function, the integral kernel of $\chi T \chi''$ is rapidly decreasing in $k_{x}$. Furthermore, as stated above, $\lVert \psi^{h}_{\sB} \rVert_{\sB}$ is bounded. Putting all this together, we have
\begin{equation}
\lvert \lambda_{\sB}(\psi^{f_{k_{x}}}_{\sB},\psi^{h}_{\sB}) \rvert \leq C''\lVert (T\chi'')\circ\chi\mathds{E}f_{k_{x}}  \rVert_{H^{1}(\sB)_{U_{b}}}\lVert \psi^{h}_{\sB} \rVert_{\sB} \; .
\end{equation}
The fast decrease of the first norm, together with the boundedness of the second norm, imply that $(k_{x},0)$ is a direction of fast decrease of $\lambda_{\sB}(\psi^{f_{k_{x}}}_{\sB},\psi^{h}_{\sB})$.

Now, let us look at the case $x\in\sD$, $y\in II$. Adopting a coordinate system in which the coordinate along the integral lines of $X$ is denoted by $t$, and the others are denoted by $\underline{x}$, the pull-back action of the one parameter group generated by $X$ acts like $(\beta^{(X)}_{\tau}f)(t,\underline{x})=(t-\tau,\underline{x})$. Exploiting the same splitting for the covectors, we write ${\mathcal T}^{\ast}_{x}(\M\diagdown\sD)\diagdown\{0\}\equiv \mathbb{R}^{4}\ni k_{x}=(k_{xt},\underline{k_{x}})$.

We can now construct the two non-null and non-vanishing covectors $q=(0,\underline{k_{x}})$ and $q'=(-k_{xt},0)$. Since $(x,y;q,q')\notin WF(\Lambda_{\M})$, from Proposition 2.1 in \citep{Verch99} there exists an open neighborhood $V'$ of $(q,q')$, as well as a function $\psi'\in\cC_{0}^{\infty}(\mathbb{R}^{4}\times\mathbb{R}^{4};\mathbb{C})$ with $\psi'(0,0)=1$ such that, denoting $x'=(\tau,\underline{x'})$, $y'=(\tau',\underline{y'})$, there exist constants $C_{n}\geq 0$ and $\lambda_{n}>0$, such that for all $p\geqslant 1$, for all $0<\lambda<\lambda_{n}$ and for all $n\geq 1$,
\begin{equation}
\underset{k,k'\in V'}{\textrm{sup}}\left|\int d\tau d\tau' d\underline{x'}d\underline{y'}\psi'(x',y') e^{i\lambda^{-1}(k_{t}\tau+\underline{k}\underline{x'})}e^{i\lambda^{-1}(k'_{t}\tau'+\underline{k'}\underline{y'})}\Lambda_{\M}\left(\beta^{(X)}_{\tau}\otimes\beta^{(X)}_{\tau'}(F_{(\underline{x'},\underline{y'}),\lambda}^{(p)})\right)\right|<C_{n}\lambda^{n} \; ,
\label{Lambda_M-rapid_decrease}
\end{equation}
as $\lambda\rightarrow 0$, where
\[F_{(\underline{x'},\underline{y'}),\lambda}^{(p)}(z,u)\coloneqq F(x+\lambda^{-p}(z-\underline{x'}-x),y+\lambda^{-p}(u-\underline{y'}-y)) \quad \textrm{and} \quad \widehat{F}(0,0)=1 \; ,\]
where $\widehat{F}$ is the usual Fourier transform. Since $\Lambda_{\M}$ is invariant under $\beta^{(X)}_{-\tau-\tau'}\otimes\beta^{(X)}_{-\tau-\tau'}$, we infer that $\Lambda_{\M}\left(\beta^{(X)}_{\tau}\otimes\beta^{(X)}_{\tau'}(F_{(\underline{x'},\underline{y'}),\lambda}^{(p)})\right)=\Lambda_{\M}\left(\beta^{(X)}_{-\tau'}\otimes\beta^{(X)}_{-\tau}(F_{(\underline{x'},\underline{y'}),\lambda}^{(p)})\right)$. This implies that \eqref{Lambda_M-rapid_decrease} also holds if one replaces (i) $\psi'$ by $\psi(x',y')=\psi((\tau',\underline{x}),(\tau,\underline{y'}))$ and (ii) $V'$ by $V=\left\{(-k'_{t},\underline{k}),(-k_{t},\underline{k'})\in\mathbb{R}^{4}\times\mathbb{R}^{4}\, |\, ((k_{t},\underline{k}),(k'_{t},\underline{k'}))\in V'\right\}$. This is an open neighborhood of $(k_{x},0)$ as one can immediately verify since $(q,q')\in V'$, so that $(k_{x},0)\in V$, and the map $\mathbb{R}^{4}\times\mathbb{R}^{4}\ni((k_{t},\underline{k}),(k'_{t},\underline{k'}))\mapsto((-k'_{t},\underline{k}),(-k_{t},\underline{k'}))\in\mathbb{R}^{4}\times\mathbb{R}^{4}$ is an isomorphism. Hence, once again from Proposition 2.1 in \citep{Verch99}, $(x,y;k_{x},0)\notin WF(\Lambda_{\M})$.

For the case when both $x,y\in\M\diagdown\sD$, if a representative of either $B(x,k_{x})$ or $B(y,k_{y})$ lies in ${\mathcal T}^{\ast}(\sD)$, then we fall back in the case above. If no representative of both $B(x,k_{x})$ and $B(y,k_{y})$ lies in ${\mathcal T}^{\ast}(\sD)$, we can introduce a partition of unit on $\sB$ (or $\sC$) for both variables, and get a decomposition like \eqref{Lambda-M_part_unity}, for both variables. The terms of this decomposition can be analysed exactly as above.
\end{proof}

Now, we need to analyse the points of $\Lambda_{\M}$ such that $(x,y;k_{x},k_{y})\in\mathcal{N}_{g}\times\mathcal{N}_{g}$ with either $x$, either $y$, or both of them in $\M\diagdown\sD$. The case where either $x$ or $y$ is in $\M\diagdown\sD$ will be treated in {\bf Case A} below. The case when both $x$ and $y$ lie in $\M\diagdown\sD$ will be treated in {\bf Case B}.


{\bf Case A:}\ \ If $x\in\M\diagdown\sD$ and $y\in\sD$ (the symmetric case being analogous), suppose that $(x,k_{x};y,-k_{y})\in WF(\Lambda_{\M})$ and there exists a representative of $(q,k_{q})\in B(x,k_{x})$ such that $(q,k_{q})\in{\mathcal T}^{\ast}(\sD)\diagdown\{0\}$. Then $(q,y;k_{q},-k_{y})\in WF((\Lambda_{\M})_{\upharpoonright(\sD\times\sD)})$ and, by the results of {\it Part 1} above, $WF((\Lambda_{\M})_{\upharpoonright(\sD\times\sD)})$ is of Hadamard form. Since there exists only one geodesic passing through a point with a given cotangent vector, the Propagation of Singularities Theorem allow us to conclude that $(x,k_{x})\sim(y,k_{y})$ with $k_{x}\in\overline{V}_{+}$, thus $WF(\Lambda_{\M})$ is of Hadamard form. We remark that this reasoning is valid for both $x\in II$ and $x\in II'$.

We are still left with the possibility that $x\in\M\diagdown\sD$ and $y\in\sD$, but no representative of $B(x,k_{x})$ lies in ${\mathcal T}^{\ast}(\sD)\diagdown\{0\}$. We intend to show that, in this case, $(x,y;k_{x},-k_{y})\notin WF(\Lambda_{\M})$ for every $k_{y}$. Without loss of generality, we will consider $x\in II$, the case $x\in II'$ being completely analogous.

We start by choosing two functions $f,h\in\cC_{0}^{\infty}(\M;\mathbb{R})$ such that $f(x)=1$ and $h(y)=1$. Since $B(x,k_{x})$ has no representative in $\sD$, there must exist $(q,k_{q})\in B(x,k_{x})$ with $q\in\sB$ such that the coordinate $U_{q}$ is non-negative. Now, considering the supports of $f$ and $h$ to be sufficiently small, we can devise a function $\chi$ such that $\chi(U_{q},\theta,\varphi)=1$ for all $(\theta,\varphi)\in\mathbb{S}^{2}$ and $\chi=0$ on $J^{-}(supp \, h)\cap\sB$. Besides, we can define $\chi'\coloneqq 1-\chi$ and, by using a coordinate patch which identifies an open neighborhood of $supp(f)$ with $\mathbb{R}^{4}$, one can arrange a conical neighborhood $\Gamma_{k_{x}}\in\mathbb{R}^{4}\diagdown\{0\}$ of $k_{x}$ such that all the bicharacteristics $B(s,k_{s})$ with $s\in supp(f)$ and $k_{s}\in\Gamma_{k_{x}}$ do not meet any point of $supp(\chi')$. One can analyse the two-point function $\Lambda_{\M}$ as
\begin{equation}
\Lambda_{\M}(f_{k_{x}}\otimes h_{k_{y}})=\lambda_{\sB}(\chi \psi^{f_{k_{x}}}_{\sB},\psi^{h_{k_{y}}}_{\sB})+\lambda_{\sB}(\chi' \psi^{f_{k_{x}}}_{\sB},\psi^{h_{k_{y}}}_{\sB})+\lambda_{\sC}(\psi^{f_{k_{x}}}_{\sC},\psi^{h_{k_{y}}}_{\sC}) \; .
\label{Lambda_M-chi}
\end{equation}
Lemma \ref{lem_isolated_sing} above tells us that only nonzero covectors are allowed in the wave front set of $\Lambda_{\M}$. The analysis of the points of the form $(x,y;k_{x},k_{y})\in\mathcal{N}_{g}\times\mathcal{N}_{g}$ is similar to the analysis presented after equation \eqref{Lambda-M_part_unity} in the proof of the mentioned Lemma.  

{\bf Case B:}\ \ The only situation not yet discussed is the case of $x,y\notin\sD$ and $B(x,k_{x})$, $B(y,k_{y})$ having no representatives in ${\mathcal T}^{\ast}(\sD)\diagdown\{0\}$ (if either $B(x,k_{x})$ or $B(y,k_{y})$ has a representative in ${\mathcal T}^{\ast}(\sD)\diagdown\{0\}$, then we fall back in the previous cases).

As in {\bf Case A}, we will consider $x,y\in II$, the case $x,y\in II'$ being completely analogous. We introduce a partition of unit $\chi,\chi'$ on $\sB$, $\chi,\chi'\in\cC_{0}^{\infty}(\sB;\mathbb{R})$ and $\chi+\chi'=1$. Moreover, these functions can be devised such that the inextensible null geodesics $\gamma_{x}$ and $\gamma_{y}$, which start respectively at $x$ and $y$ with cotangent vectors $k_{x}$ and $k_{y}$ intersect $\sB$ in $U_{x}$ and $U_{y}$ (possibly $U_{x}=U_{y}$; we omit the subscript $_{b}$ to simplify the notation), respectively, included in two open neighborhoods, $\mathcal{O}_{x}$ and $\mathcal{O}_{y}$ (possibly $\mathcal{O}_{x}=\mathcal{O}_{y}$) where $\chi'$ vanishes. Hence, the two-point function reads
\begin{align}
\lambda_{\M}(\psi^{f_{k_{x}}},\psi^{h_{k_{y}}})&=\lambda_{\sB}(\chi\psi^{f_{k_{x}}}_{\sB},\chi\psi^{h_{k_{y}}}_{\sB})+\lambda_{\sB}(\chi\psi^{f_{k_{x}}}_{\sB},\chi'\psi^{h_{k_{y}}}_{\sB}) \nonumber \\
&+\lambda_{\sB}(\chi'\psi^{f_{k_{x}}}_{\sB},\chi\psi^{h_{k_{y}}}_{\sB})+\lambda_{\sB}(\chi'\psi^{f_{k_{x}}}_{\sB},\chi'\psi^{h_{k_{y}}}_{\sB})+\lambda_{\sC}(\psi^{f_{k_{x}}}_{\sC},\psi^{h_{k_{y}}}_{\sC}) \; .
\label{lambda-chi_chi'}
\end{align}
The results of Proposition \ref{prop_rapid-decrease}, Lemma \ref{lem_isolated_sing} and of {\bf Case A} above tell us that all but the first term in the rhs of \eqref{lambda-chi_chi'} are smooth. We will then focus on the first term. Writing the integral kernel of $\lambda_{\sB}$ as $T$, interpreted as a distribution in $(\cC_{0}^{\infty})'(\sB\times\sB)$, we notice that, as an element of $(\cC_{0}^{\infty})'(\sB\times\sB)$, $\lambda_{\sB}$ can be written as
\begin{equation}
\lambda_{\sB}(\chi\psi^{f_{k_{x}}}_{\sB},\chi\psi^{h_{k_{y}}}_{\sB})=\chi T\chi\left(\mathds{E}_{\upharpoonright \sB}\otimes\mathds{E}_{\upharpoonright \sB}(f\otimes h)\right) \; ,
\end{equation}
where $\mathds{E}_{\upharpoonright \sB}$ is the causal propagator with one entry restricted to $\sB$ and $\chi T\chi\in(\cC^{\infty})'(\sB\times\sB)$ (as an element of the dual space to $\cC^{\infty}$, $\chi T\chi$ is itself a compactly supported bidistribution). For the composition $\chi T\chi(\mathds{E}_{\upharpoonright \sB}\otimes\mathds{E}_{\upharpoonright \sB})$ to make sense as a composition of bidistributions, Theorem 8.2.13 of \citep{Hormander-I} shows that it is sufficient that
\begin{equation}
WF(\chi T\chi)\cap WF'(\mathds{E}_{\upharpoonright \sB}\otimes\mathds{E}_{\upharpoonright \sB})_{Y\times Y}=\emptyset \; .
\label{WF_comp_cond}
\end{equation}
The subscript $_{Y}$ makes sense if the bidistribution is viewed as an element of $(\cC_{0}^{\infty})'(X\times Y)$ and, for a general bidistribution $\Lambda_{2}$ of this sort\footnote{The subscript $_{Y}$ means that the ``original'' wave front set must contain the zero covector of ${\mathcal T}^{\ast}X$ and the $'$ means that the nonzero covector has its sign inverted. For more details, see section 8.2 of \citep{Hormander-I}},
\begin{equation}
WF'(\Lambda_{2})_{Y}=\left\{(y,\eta);\, (x,y;0,-\eta)\in WF(\Lambda_{2})\, \textrm{for }x\in X\right\} \; .
\label{WF'_Y}
\end{equation}
The wave front set of $\mathds{E}$ was calculated in \citep{Radzikowski96}:
\begin{equation}
WF(\mathds{E})=\left\{(x,y;k_{x},k_{y})\in{\mathcal T}^{\ast}(\M\times\M)\diagdown\{0\} \vert (x,k_{x})\sim(y,-k_{y})\right\} \; .
\label{WF_E}
\end{equation}
The wave front set of $\mathds{E}\otimes\mathds{E}$ can be read out from Theorem 8.2.9 of \citep{Hormander-I} as
\begin{equation}
WF(\mathds{E}\otimes\mathds{E})\subset\left(WF(\mathds{E})\times WF(\mathds{E})\right)\cup\left((supp\mathds{E}\times\{0\})\times WF(\mathds{E})\right)\cup\left(WF(\mathds{E})\times(supp\mathds{E}\times\{0\})\right) \; .
\label{WF_E_otimes_E}
\end{equation}
From this last equation and the fact that the zero covector is not contained in $WF(\mathds{E})$, we conclude that
\begin{equation}
WF'(\mathds{E}_{\upharpoonright \sB}\otimes\mathds{E}_{\upharpoonright \sB})_{Y\times Y}=\emptyset \; .
\label{WF_E_otimes_E_YxY}
\end{equation}
Thus the composition makes sense as a composition of bidistributions, and Theorem 8.2.13 of \citep{Hormander-I} proves that, under these conditions,
\begin{equation}
WF(\chi T\chi(\mathds{E}_{\upharpoonright \sB}\otimes\mathds{E}_{\upharpoonright \sB}))\subset WF(\mathds{E}_{\upharpoonright \sB}\otimes\mathds{E}_{\upharpoonright \sB})_{X\times X}\cup WF'(\mathds{E}_{\upharpoonright \sB}\otimes\mathds{E}_{\upharpoonright \sB})\circ WF(\chi T\chi) \; .
\label{WF_chiTchiEotimesE}
\end{equation}
The same reasoning which led to equation \eqref{WF_E_otimes_E_YxY} leads to the conclusion that the first term in the rhs of \eqref{WF_chiTchiEotimesE} is empty.

The wave front set of $T$ was calculated in Lemma 4.4 of \citep{Moretti08}. We will again introduce a coordinate system at which the coordinate along the integral lines of $X$ is denoted by $t$, the remaining coordinates being denoted by $\underline{x}$. The same splitting will be used for covectors. The wave front set of $T$ is written as
\[WF(T)=A\cup B \; ,\]
where
\begin{align}
A&\coloneqq\left\{\left((t,\underline{x}),(t',\underline{x'});(k_{t},k_{\underline{x}}),(k_{t'},k_{\underline{x'}})\right)\in{\mathcal T}^{\ast}(\sB\times\sB)\diagdown\{0\} \, \vert \, x=x' ; k_{x}=-k_{x'} ; k_{t}>0\right\} \nonumber \\
B&\coloneqq\left\{\left((t,\underline{x}),(t',\underline{x'});(k_{t},k_{\underline{x}}),(k_{t'},k_{\underline{x'}})\right)\in{\mathcal T}^{\ast}(\sB\times\sB)\diagdown\{0\} \, \vert \, \underline{x}=\underline{x'} ; k_{\underline{x}}=-k_{\underline{x'}} ; k_{t}=k_{t'}=0\right\} \; .
\label{WF_T}
\end{align}
With these at hand, the author of \citep{Moretti08} proved that the wave front set \eqref{WF_chiTchiEotimesE} is of Hadamard form.

Hence we have completed the proof of
\begin{thm}\label{theorem-omega_M_Hadamard}
The wave front set of the two-point function $\Lambda_{\M}$ of the state $\omega_{\M}$, individuated in \eqref{2ptfcn_M=B+C} is given by
\begin{equation}
WF(\Lambda_{\M})=\left\{\left(x_{1},k_{1};x_{2},-k_{2}\right) | \left(x_{1},k_{1};x_{2},k_{2}\right)\in {\mathcal T}^{*}\left(\M\times\M\right) \diagdown \{0\} ; (x_{1},k_{1})\sim (x_{2},k_{2}) ; k_{1}\in \overline{V}_{+}\right\} \; ,
\label{Wf-omega_M}
\end{equation}
thus the state $\omega_{\M}$ is a Hadamard state.
\end{thm}

\chapter{Conclusions}\label{chap_concl}

The existence of Hadamard states on a general globally hyperbolic spacetime is long known \citep{FullingNarcowichWald81}, although it was proven in a rather indirect way. On the one hand, the Hadamard condition is known to be sufficient for the quantum field theory to be renormalizable \citep{BruFreKoe96,BruFre00,HoWa02}. Hence, Hadamard states have been successfully employed in the renormalization of the energy-momentum tensor \citep{Moretti03,Wald77} and the evaluation of the backreaction of quantum fields in the background spacetime \citep{DappiaggiFredenhagenPinamonti08,DappiaggiHackPinamonti09,DappiaggiHackPinamonti10} (see \citep{FredenhagenHack13} for a recent review). On the other hand, it was recently proven that it is not possible to uniquely construct a Hadamard state in any globally hyperbolic spacetime \citep{FewsterVerch12,FewsterVerch_sf12}. Moreover, it is not yet known whether, generically, the Hadamard condition is a necessary condition for renormalizability \citep{FewsterVerch13}.

It is thus interesting to construct explicit examples of Hadamard states in specific (classes of) spacetimes. This was the objective of this Thesis. Although our first two examples can be applied in quite a broad class of spacetimes, they are not uniquely singled out by the spacetime geometry. The last example presented here, on the other hand, is constructed solely from the geometrical features of the Schwarzschild-de Sitter spacetime. In other spacetimes possessing two horizons but without spherically symmetric spatial sections, the decay rates \eqref{Sch-dS_decay} would be different (see \citep{DafermosRodnianski07}).

The {\it States of Low Energy} \citep{ThemBrum13} formulated in chapter \ref{chap_sle} were intended to minimize the expectation value of the smeared energy density and were constructed in expanding spacetimes without spatial symmetries. 
Despite its quite general definition, we do not claim that this is the most general possible, because, generically, the metric in a globally hyperbolic spacetime is of the form \eqref{metric-GH}. Our construction relies on the existence of time-independent modes. A more general construction would require different techniques that do not rely on the ocurrence of such modes, such as an analysis based on pseudo-differential calculus, as the authors of \citep{GerardWrochna12} made in order to construct Hadamard states. Furthermore, the states here defined depend on the particular test function chosen for the smearing. As noted at the end of section \ref{SLE-est}, such a dependence would only be dropped in a static spacetime, the case in which our state would reduce to the static vacuum. From a more practical point of view, the particle production process on RW spacetimes was shown by Degner \citep{Degner09} not to be strongly dependent on the test function chosen. He also showed \citep{Degner13} that, on de Sitter spacetime, if the support of the test function is taken on the infinite past, the SLE asymptotically converges to the so-called Bunch-Davies vacuum state \citep{Allen85}. This result is also valid for asymptotically de Sitter spacetimes \citep{DappiaggiMorettiPinamonti09c}. In general, these states differ from the SLE because, in their case, whenever the spacetime possesses an everywhere timelike Killing vector field, the one-parameter group which implements the action of this vector field on the GNS representation associated to that state has a positive self-adjoint generator. In the coordinate system used in the present work, this would amount to $\beta_{j}=0$. Thus the coincidence is only asymptotic, in the sense described above.

The {\it ``Vacuum-like'' states} \citep{BrumFredenhagen14} formulated in chapter \ref{chap_vac-like} were defined as a variation of the proposal of Sorkin and Johnston \citep{AfshordiAslanbeigiSorkin12}, who tried to construct a vacuum state in a generic globally hyperbolic spacetime. Such a state would be in conflict with the results mentioned in the first paragraph \citep{FewsterVerch12,FewsterVerch_sf12} and it was rapidly shown that the original proposal did not give rise, in general, to a Hadamard state \citep{FewsterVerch-SJ12}. The state formulated in \citep{AfshordiAslanbeigiSorkin12} was defined in a relatively compact globally hyperbolic spacetime isometrically embedded in a larger, globally hyperbolic spacetime. The kernel of the two-point function of these states is the positive spectral part of the commutator function. Their proposal failed to generate a Hadamard state because the time coordinate, in the submanifold, ranged over a closed interval. We considered an open timelike interval and, moreover, we smeared the commutator function with positive test functions. We tested this idea in static and expanding spacetimes and proved that our states are well-defined pure Hadamard states. They are, however, in contrast to the S-J states, not uniquely associated to the spacetime. Several interesting questions might then be posed. First one would like to generalize the construction to generic hyperbolic spacetimes which are relatively compact subregions of another spacetime. This involves some technical problems but we do not see an unsurmountable obstruction (for example, the construction of these states also relies on the existence of time-independent modes, as above). In the best case scenario these states (as also the original S-J states) might converge to a Hadamard state as the subregion increases and eventually covers the full larger spacetime. Such a situation occurs in static spacetimes, and it would be interesting to identify the properties of a spacetime on which this procedure works. There is an interesting connection to the proposal of the fermionic projector of Finster \citep{Finster11} where an analogous construction for the Dirac field was considered. The case of the scalar field is however much easier because of the Hilbert space structure of the functions on the manifold, in contrast to the indefinite scalar product on the spinor bundle of a Lorentzian spacetime. Another interesting question concerns the physical interpretation. We do not expect that these states should be interpreted as some kind of vacuum, but we would like to better understand the relation of these states with the States of Low Energy. A numerical analysis of the expectation value of the energy momentum tensor in these states, similarly to \citep{Degner09,Degner13}, might be a good firts step in this direction. Both the {\it ``Vacuum-like'' states} and the {\it States of Low Energy} are suitable as initial states for the problem of backreaction.

The state we constructed in the Schwarzschild-de Sitter spacetime in chapter \ref{chap_Sch-dS} \citep{BrumJoras14}, to our knowledge, is the first explicit example of a Hadamard state in this spacetime. It is not defined in the complete extension of the spacetime, but rather in the ``physical'' (nonextended) region between the singularity at $r=0$ and the singularity at $r=\infty$. In this sense, our state is not to be interpreted as the Hartle-Hawking state in this spacetime, whose nonexistence was proven in \citep{KayWald91}. Neither can it be interpreted as the Unruh state because, in the de Sitter spacetime, the Unruh state can be extended to the whole spacetime while retaining the Hadamard property \citep{NarnhoferPeterThirring96}. Our state cannot have such a feature. We constructed the state solely from geometrical features of the Schwarzschild-de Sitter spacetime and, consequently, it is invariant under the action of its group of symmetries. We showed it can be isometrically mapped to a state on the past horizons, a feature which, for the analogous state constructed in the Schwarzschild case \citep{DappiaggiMorettiPinamonti09}, sufficed to prove it a KMS state. Our state is not a KMS state because, under this mapping, the functional is written as the tensor product of two functionals, each corresponding to a KMS state at a different temperature. We further remark that, even in the Schwarzschild spacetime, the existence of the Hartle-Hawking state, whose features were analysed in \citep{KayWald91}, was only recently proved in \citep{Sanders13}, where the author analysed a Wick rotation in the Killing time coordinate. We believe that the method put forward in \citep{Sanders13}, if applied to the Schwarzschild-de Sitter spacetime, would give rise to the contradictions pointed out in \citep{KayWald91}.



\chapter*{References}

\begin{thebibliography}{100}

\bibitem{AfshordiAslanbeigiSorkin12}
Niayesh Afshordi, Siavash Aslanbeigi, and Rafael~D. Sorkin.
\newblock {A Distinguished Vacuum State for a Quantum Field in a Curved
  Spacetime: Formalism, Features, and Cosmology}.
\newblock {\em JHEP}, 1208:137, 2012, 1205.1296.

\bibitem{Alcubierre94}
Miguel Alcubierre.
\newblock The warp drive: hyper-fast travel within general relativity.
\newblock {\em Classical and Quantum Gravity}, 11:L73, 1994.

\bibitem{Allen85}
Bruce Allen.
\newblock {Vacuum states in de Sitter space}.
\newblock {\em Phys. Rev. D}, 32:3136--3149, 1985.

\bibitem{ArakiYamagami82}
Huzihiro Araki and Shigeru Yamagami.
\newblock On quasi-equivalence of quasifree states of the canonical commutation
  relations.
\newblock {\em Publications of the Research Institute for Mathematical
  Sciences}, 18:283--338, 1982.

\bibitem{AvetisyanVerch12}
Zhirayr Avetisyan and Rainer Verch.
\newblock {Explicit harmonic and spectral analysis in Bianchi I-VII-type
  cosmologies}.
\newblock {\em Classical and Quantum Gravity}, 30:155006, 2013.

\bibitem{Avetisyan12}
Zhirayr~G. Avetisyan.
\newblock {A unified mode decomposition method for physical fields in
  homogeneous cosmology}.
\newblock 2012, 1212.2408.

\bibitem{BazanskiFerrari86}
S.L. Ba$\dot{z}$a\'{n}ski and V.~Ferrari.
\newblock {Analytic extension of the Schwarzschild-de Sitter metric}.
\newblock {\em Il Nuovo Cimento B Series 11}, 91:126--142, 1986.

\bibitem{BarGinouxPfaffle07}
C.~B{\"a}r, N.~Ginoux, and F.~Pf{\"a}ffle.
\newblock {\em Wave Equations on Lorentzian Manifolds and Quantization}.
\newblock Esi Lectures in Mathematics and Physics. European Mathematical
  Society, Z{\"u}rich, 2007.

\bibitem{BeniniDappiaggiHack13}
Marco Benini, Claudio Dappiaggi, and Thomas-Paul Hack.
\newblock Quantum field theory on curved backgrounds -- a primer.
\newblock {\em International Journal of Modern Physics A}, 28:1330023, 2013.

\bibitem{Berard86}
P.H. Berard.
\newblock {\em Spectral Geometry: Direct and Inverse Problems}, volume 1207 of
  {\em Lecture Notes in Mathematics}.
\newblock Springer, Berlin, 1986.

\bibitem{BernalSanchez03}
Antonio~N. Bernal and Miguel S{\'a}nchez.
\newblock {On Smooth Cauchy Hypersurfaces and Geroch's Splitting Theorem}.
\newblock {\em Communications in Mathematical Physics}, 243:461--470, 2003.

\bibitem{BostelmannCadamuroFewster13}
Henning Bostelmann, Daniela Cadamuro, and Christopher~J. Fewster.
\newblock {Quantum energy inequality for the massive Ising model}.
\newblock {\em Phys. Rev. D}, 88:025019, 2013.

\bibitem{BostelmannFewster09}
Henning Bostelmann and Christopher~J. Fewster.
\newblock Quantum inequalities from operator product expansions.
\newblock {\em Communications in Mathematical Physics}, 292:761--795, 2009.

\bibitem{BraRob-I}
O.~Bratteli and D.~W. Robinson.
\newblock {\em Operator Algebras and Quantum Statistical Mechanics 1: C*- and
  W*-Algebras. Symmetry Groups. Decomposition of States}.
\newblock Springer, New York, 2003.

\bibitem{BraRob-II}
O.~Bratteli and D.~W. Robinson.
\newblock {\em Operator Algebras and Quantum Statistical Mechanics 2:
  Equilibrium States. Models in Quantum Statistical Mechanics}.
\newblock Springer, Berlin, 2003.

\bibitem{BrumFredenhagen14}
Marcos Brum and Klaus Fredenhagen.
\newblock {`Vacuum$-$like' Hadamard states for quantum fields on curved
  spacetimes}.
\newblock {\em Classical and Quantum Gravity}, 31:025024, 2014.

\bibitem{BrumJoras14}
Marcos Brum and Sergio~E. Jor\'{a}s.
\newblock {Hadamard state in Schwarzschild-de Sitter spacetime}.
\newblock 2014, 1405.7916.

\bibitem{BruFre00}
Romeo Brunetti and Klaus Fredenhagen.
\newblock {Microlocal Analysis and Interacting Quantum Field Theories:
  Renormalization on Physical Backgrounds}.
\newblock {\em Communications in Mathematical Physics}, 208:623--661, 2000.

\bibitem{BruFreKoe96}
Romeo Brunetti, Klaus Fredenhagen, and Martin K\"ohler.
\newblock {The microlocal spectrum condition and Wick polynomials of free
  fields on curved spacetimes}.
\newblock {\em Communications in Mathematical Physics}, 180:633--652, 1996.

\bibitem{BFV03}
Romeo Brunetti, Klaus Fredenhagen, and Rainer Verch.
\newblock The generally covariant locality principle -- a new paradigm for
  local quantum field theory.
\newblock {\em Communications in Mathematical Physics}, 237:31--68, 2003.

\bibitem{BuchholzSolveen13}
Detlev Buchholz and Christoph Solveen.
\newblock Unruh effect and the concept of temperature.
\newblock {\em Classical and Quantum Gravity}, 30:085011, 2013.

\bibitem{ChernChenLam99}
S.S. Chern, W.~Chen, and K.S. Lam.
\newblock {\em Lectures on Differential Geometry}.
\newblock Series on university mathematics. World Scientific, Singapore, 1999.

\bibitem{Chernoff73}
Paul~R Chernoff.
\newblock Essential self-adjointness of powers of generators of hyperbolic
  equations.
\newblock {\em Journal of Functional Analysis}, 12:401 -- 414, 1973.

\bibitem{ChillFre09}
Bruno Chilian and Klaus Fredenhagen.
\newblock The time slice axiom in perturbative quantum field theory on globally
  hyperbolic spacetimes.
\newblock {\em Communications in Mathematical Physics}, 287:513--522, 2009.

\bibitem{Choquet09}
Y.~Choquet-Bruhat.
\newblock {\em General Relativity and the Einstein Equations}.
\newblock Oxford Mathematical Monographs. OUP, Oxford, 2009.

\bibitem{ChoudhuryPadmanabhan07}
T.~Roy Choudhury and T.~Padmanabhan.
\newblock {Concept of temperature in multi-horizon spacetimes: analysis of
  Schwarzschild-De Sitter metric}.
\newblock {\em General Relativity and Gravitation}, 39:1789--1811, 2007.

\bibitem{DafermosRodnianski07}
Mihalis Dafermos and Igor Rodnianski.
\newblock {The Wave equation on Schwarzschild-de Sitter spacetimes}.
\newblock 2007, 0709.2766.

\bibitem{DappiaggiFredenhagenPinamonti08}
Claudio Dappiaggi, Klaus Fredenhagen, and Nicola Pinamonti.
\newblock Stable cosmological models driven by a free quantum scalar field.
\newblock {\em Phys. Rev. D}, 77:104015, 2008.

\bibitem{DappiaggiHackPinamonti10}
Claudio Dappiaggi, Thomas-Paul Hack, Jan Moller, and Nicola Pinamonti.
\newblock {Dark Energy from Quantum Matter}.
\newblock 2010, 1007.5009.

\bibitem{DappiaggiHackPinamonti09}
Claudio Dappiaggi, Thomas-Paul Hack, and Nicola Pinamonti.
\newblock {The Extended Algebra of Observables for Dirac Fields and the Trace
  Anomaly of their Stress-Energy Tensor}.
\newblock {\em Reviews in Mathematical Physics}, 21:1241--1312, 2009.

\bibitem{DappiaggiMorettiPinamonti06}
Claudio Dappiaggi, Valter Moretti, and Nicola Pinamonti.
\newblock Rigorous steps towards holography in asymptotically flat spacetimes.
\newblock {\em Reviews in Mathematical Physics}, 18:349--415, 2006.

\bibitem{DappiaggiMorettiPinamonti09b}
Claudio Dappiaggi, Valter Moretti, and Nicola Pinamonti.
\newblock Cosmological horizons and reconstruction of quantum field theories.
\newblock {\em Communications in Mathematical Physics}, 285:1129--1163, 2009.

\bibitem{DappiaggiMorettiPinamonti09c}
Claudio Dappiaggi, Valter Moretti, and Nicola Pinamonti.
\newblock {Distinguished quantum states in a class of cosmological spacetimes
  and their Hadamard property}.
\newblock {\em Journal of Mathematical Physics}, 50:062304, 2009.

\bibitem{DappiaggiMorettiPinamonti09}
Claudio Dappiaggi, Valter Moretti, and Nicola Pinamonti.
\newblock {Rigorous construction and Hadamard property of the Unruh state in
  Schwarzschild spacetime}.
\newblock {\em Adv.Theor.Math.Phys.}, 15:355--448, 2011, 0907.1034.

\bibitem{Degner13}
Andreas Degner.
\newblock {\em Properties of {S}tates of {L}ow {E}nergy on {C}osmological
  {S}pacetimes}.
\newblock Phd thesis, University of Hamburg, 2013.
\newblock http://www.desy.de/uni-th/theses/Diss\_Degner.pdf.

\bibitem{Degner09}
Andreas Degner and Rainer Verch.
\newblock Cosmological particle creation in states of low energy.
\newblock {\em Journal of Mathematical Physics}, 51:--, 2010.

\bibitem{Dimock80}
J.~Dimock.
\newblock Algebras of local observables on a manifold.
\newblock {\em Communications in Mathematical Physics}, 77:219--228, 1980.

\bibitem{DuistermaatHormander72}
J.~J. Duistermaat and L.~H{\"o}rmander.
\newblock Fourier integral operators. {II}.
\newblock {\em Acta Math.}, 128:183--269, 1972.

\bibitem{DuistermaatKolk10}
J.~J. Duistermaat and J.~A.~C. Kolk.
\newblock {\em Distributions}.
\newblock Cornerstones. Birkh\"auser Boston Inc., Boston, MA, 2010.
\newblock Theory and applications.

\bibitem{EpsGlaJaffe65}
H.~Epstein, V.~Glaser, and A.~Jaffe.
\newblock {Nonpositivity of the Energy Density in Quantized Field Theories}.
\newblock {\em Il Nuovo Cimento}, 36:1016--1022, 1965.

\bibitem{Fewster00}
Christopher~J Fewster.
\newblock A general worldline quantum inequality.
\newblock {\em Classical and Quantum Gravity}, 17:1897, 2000.

\bibitem{FewsterOsterbrink08}
Christopher~J Fewster and Lutz~W Osterbrink.
\newblock Quantum energy inequalities for the non-minimally coupled scalar
  field.
\newblock {\em Journal of Physics A: Mathematical and Theoretical}, 41:025402,
  2008.

\bibitem{FewsterPfenning03}
Christopher~J. Fewster and Michael~J. Pfenning.
\newblock A quantum weak energy inequality for spin-one fields in curved
  space-time.
\newblock {\em Journal of Mathematical Physics}, 44:4480--4513, 2003.

\bibitem{FewsterSmith08}
Christopher~J. Fewster and Calvin~J. Smith.
\newblock {Absolute Quantum Energy Inequalities in Curved Spacetime}.
\newblock {\em Annales Henri Poincar\'e}, 9:425--455, 2008.

\bibitem{FewsterVerch02}
Christopher~J. Fewster and Rainer Verch.
\newblock {A Quantum Weak Energy Inequality for Dirac Fields in Curved
  Spacetime}.
\newblock {\em Communications in Mathematical Physics}, 225:331--359, 2002.

\bibitem{FewsterVerch12}
Christopher~J. Fewster and Rainer Verch.
\newblock Dynamical locality and covariance: What makes a physical theory the
  same in all spacetimes?
\newblock {\em Annales Henri Poincar\'e}, 13:1613--1674, 2012.

\bibitem{FewsterVerch_sf12}
Christopher~J. Fewster and Rainer Verch.
\newblock Dynamical locality of the free scalar field.
\newblock {\em Annales Henri Poincar\'e}, 13:1675--1709, 2012.

\bibitem{FewsterVerch-SJ12}
Christopher~J Fewster and Rainer Verch.
\newblock On a recent construction of ``vacuum-like'' quantum field states in
  curved spacetime.
\newblock {\em Classical and Quantum Gravity}, 29:205017, 2012.

\bibitem{FewsterVerch13}
Christopher~J Fewster and Rainer Verch.
\newblock {The necessity of the Hadamard condition}.
\newblock {\em Classical and Quantum Gravity}, 30:235027, 2013.

\bibitem{Finster11}
Felix Finster.
\newblock {A Formulation of Quantum Field Theory Realizing a Sea of Interacting
  Dirac Particles}.
\newblock {\em Letters in Mathematical Physics}, 97:165--183, 2011.

\bibitem{Ford78}
L.~H. Ford.
\newblock {Quantum Coherence Effects and the Second Law of Thermodynamics}.
\newblock {\em Proceedings of the Royal Society of London. A. Mathematical and
  Physical Sciences}, 364:227--236, 1978.

\bibitem{Ford91}
L.~H. Ford.
\newblock Constraints on negative-energy fluxes.
\newblock {\em Phys. Rev. D}, 43:3972--3978, 1991.

\bibitem{FordRoman90}
L.~H. Ford and Thomas~A. Roman.
\newblock Moving mirrors, black holes, and cosmic censorship.
\newblock {\em Phys. Rev. D}, 41:3662--3670, 1990.

\bibitem{FordRoman92}
L.~H. Ford and Thomas~A. Roman.
\newblock "cosmic flashing" in four dimensions.
\newblock {\em Phys. Rev. D}, 46:1328--1339, 1992.

\bibitem{Ford09}
L.H. Ford.
\newblock {Negative Energy Densities in Quantum Field Theory}.
\newblock {\em Int.J.Mod.Phys.}, A25:2355--2363, 2010, 0911.3597.

\bibitem{FredenhagenHack13}
Klaus Fredenhagen and Thomas-Paul Hack.
\newblock {Quantum field theory on curved spacetime and the standard
  cosmological model}.
\newblock {\em ArXiv e-prints}, 2013, 1308.6773.

\bibitem{Fulling89}
S.~A. Fulling.
\newblock {\em Aspects of Quantum Field Theory in Curved Spacetime}.
\newblock London Mathematical Society Student Texts. Cambridge University
  Press, Cambridge, 1989.

\bibitem{FullingNarcowichWald81}
S.~A. Fulling, F.J Narcowich, and Robert~M Wald.
\newblock Singularity structure of the two-point function in quantum field
  theory in curved spacetime, \{II\}.
\newblock {\em Annals of Physics}, 136:243 -- 272, 1981.

\bibitem{FullingSweenyWald78}
Stephen~A. Fulling, Mark Sweeny, and Robert~M. Wald.
\newblock Singularity structure of the two-point function in quantum field
  theory in curved spacetime.
\newblock {\em Communications in Mathematical Physics}, 63:257--264, 1978.

\bibitem{GerardWrochna12}
Christian Gerard and Michal Wrochna.
\newblock {Construction of Hadamard states by pseudo-differential calculus}.
\newblock 2012, 1209.2604.

\bibitem{GibbonsHawking77}
G.~W. Gibbons and S.~W. Hawking.
\newblock Cosmological event horizons, thermodynamics, and particle creation.
\newblock {\em Phys. Rev. D}, 15:2738--2751, 1977.

\bibitem{GriffithsPodolsky09}
Jerry~B. Griffiths and Ji{\v{r}}{\'{\i}} Podolsk{\'y}.
\newblock {\em Exact space-times in {E}instein's general relativity}.
\newblock Cambridge Monographs on Mathematical Physics. Cambridge University
  Press, Cambridge, 2009.

\bibitem{Haag96}
R.~Haag.
\newblock {\em Local Quantum Physics: Fields, Particles, Algebras}.
\newblock Texts and Monographs in Physics. Springer Verlag, Berlin, 1996.

\bibitem{HHW67}
R.~Haag, N.M. Hugenholtz, and M.~Winnink.
\newblock On the equilibrium states in quantum statistical mechanics.
\newblock {\em Communications in Mathematical Physics}, 5:215--236, 1967.

\bibitem{HaagKastler64}
Rudolf Haag and Daniel Kastler.
\newblock An algebraic approach to quantum field theory.
\newblock {\em Journal of Mathematical Physics}, 5:848--861, 1964.

\bibitem{HNS84}
Rudolf Haag, Heide Narnhofer, and Ulrich Stein.
\newblock On quantum field theory in gravitational background.
\newblock {\em Communications in Mathematical Physics}, 94:219--238, 1984.

\bibitem{HackPhD}
Thomas-Paul Hack.
\newblock {On the Backreaction of Scalar and Spinor Quantum Fields in Curved
  Spacetimes}.
\newblock 2010, 1008.1776.

\bibitem{Hawking74}
S.W. Hawking.
\newblock Particle creation by black holes.
\newblock {\em Communications in Mathematical Physics}, 43:199--220, 1975.

\bibitem{HawkingEllis73}
S.W. Hawking and G.F.R. Ellis.
\newblock {\em The Large Scale Structure of Space-Time}.
\newblock Cambridge Monographs on Mathematical Physics. Cambridge University
  Press, Cambridge, 1973.

\bibitem{HoWa01}
Stefan Hollands and Robert~M. Wald.
\newblock {Local Wick Polynomials and Time Ordered Products of Quantum Fields
  in Curved Spacetime}.
\newblock {\em Communications in Mathematical Physics}, 223:289--326, 2001.

\bibitem{HoWa02}
Stefan Hollands and Robert~M. Wald.
\newblock {Existence of Local Covariant Time Ordered Products of Quantum Fields
  in Curved Spacetime}.
\newblock {\em Communications in Mathematical Physics}, 231:309--345, 2002.

\bibitem{HoWa03}
Stefan Hollands and Robert~M. Wald.
\newblock {On the Renormalization Group in Curved Spacetime}.
\newblock {\em Communications in Mathematical Physics}, 237:123--160, 2003.

\bibitem{HoWa05}
Stefan Hollands and Robert~M. Wald.
\newblock {Conservation of the Stress Tensor in Perturbative Interacting
  Quantum Field Theory in Curved Spacetimes}.
\newblock {\em Reviews in Mathematical Physics}, 17:227--311, 2005.

\bibitem{Hormander-I}
L.~H{\"o}rmander.
\newblock {\em The Analysis of Linear Partial Differential Operators I:
  Distribution theory and Fourier Analysis}.
\newblock Springer, Berlin, 1983.

\bibitem{Hormander-IV}
L.~H{\"o}rmander.
\newblock {\em The Analysis of Linear Partial Differential Operators IV:
  Fourier Integral Operators}.
\newblock Classics in Mathematics. Springer, Berlin, 1994.

\bibitem{Hormander-III}
L.~H{\"o}rmander.
\newblock {\em The Analysis of Linear Partial Differential Operators III:
  Pseudo-Differential Operators}.
\newblock Classics in Mathematics. Springer, Berlin, 2007.

\bibitem{HuLingZhang06}
Bo~Hu, Yi~Ling, and Hongbao Zhang.
\newblock Quantum inequalities for a massless spin-$3/2$ field in minkowski
  spacetime.
\newblock {\em Phys. Rev. D}, 73:045015, 2006.

\bibitem{Johnston09}
Steven Johnston.
\newblock Feynman propagator for a free scalar field on a causal set.
\newblock {\em Phys. Rev. Lett.}, 103:180401, 2009.

\bibitem{Jost11}
J.~Jost.
\newblock {\em Riemannian Geometry and Geometric Analysis}.
\newblock Universitext - Springer-Verlag. Springer, Heidelberg, 6 edition,
  2011.

\bibitem{JunSchrohe02}
W.~Junker and E.~Schrohe.
\newblock {Adiabatic Vacuum States on General Spacetime Manifolds: Definition,
  Construction, and Physical Properties}.
\newblock {\em Annales Henri Poincar\'e}, 3:1113--1181, 2002.

\bibitem{Kay78}
Bernard~S. Kay.
\newblock Linear spin-zero quantum fields in external gravitational and scalar
  fields.
\newblock {\em Communications in Mathematical Physics}, 62:55--70, 1978.

\bibitem{KayWald91}
Bernard~S. Kay and Robert~M. Wald.
\newblock Theorems on the uniqueness and thermal properties of stationary,
  nonsingular, quasifree states on spacetimes with a bifurcate killing horizon.
\newblock {\em Physics Reports}, 207:49 -- 136, 1991.

\bibitem{LakeRoeder77}
Kayll Lake and R.~C. Roeder.
\newblock Effects of a nonvanishing cosmological constant on the spherically
  symmetric vacuum manifold.
\newblock {\em Phys. Rev. D}, 15:3513--3519, 1977.

\bibitem{LuRo90}
Christian L\"uders and John~E. Roberts.
\newblock Local quasiequivalence and adiabatic vacuum states.
\newblock {\em Communications in Mathematical Physics}, 134:29--63, 1990.

\bibitem{ManuceauVerbeure68}
J.~Manuceau and A.~Verbeure.
\newblock {Quasi-free states of the CCR$-$Algebra and Bogoliubov
  transformations}.
\newblock {\em Communications in Mathematical Physics}, 9:293--302, 1968.

\bibitem{Moretti03}
Valter Moretti.
\newblock Comments on the stress-energy tensor operator in curved spacetime.
\newblock {\em Communications in Mathematical Physics}, 232:189--221, 2003.

\bibitem{Moretti08}
Valter Moretti.
\newblock {Quantum Out-States Holographically Induced by Asymptotic Flatness:
  Invariance under Spacetime Symmetries, Energy Positivity and Hadamard
  Property}.
\newblock {\em Communications in Mathematical Physics}, 279:31--75, 2008.

\bibitem{NarnhoferPeterThirring96}
H.~Narnhofer, I.~Peter, and W.~Thirring.
\newblock {How Hot is the de-Sitter Space?}
\newblock {\em International Journal of Modern Physics B}, 10:1507--1520, 1996.

\bibitem{Olbermann07}
Heiner Olbermann.
\newblock {States of low energy on Robertson$-$Walker spacetimes}.
\newblock {\em Classical and Quantum Gravity}, 24:5011, 2007.

\bibitem{ONeill83}
Barrett O'Neill.
\newblock {\em Semi-{R}iemannian geometry}.
\newblock Pure and Applied Mathematics. Academic Press Inc., New York, 1983.

\bibitem{Osinovsky73}
M.~E. Osinovsky.
\newblock Bianchi universes admitting full groups of motions.
\newblock In {\em Annales de l'institut Henri Poincar{\'e} (A) Physique
  th{\'e}orique}, volume~19, pages 197--210. Gauthier-villars, 1973.

\bibitem{Parker69}
Leonard Parker.
\newblock {Quantized Fields and Particle Creation in Expanding Universes. I}.
\newblock {\em Phys. Rev.}, 183:1057--1068, 1969.

\bibitem{Radzikowski96}
Marek~J. Radzikowski.
\newblock {Micro-local approach to the Hadamard condition in quantum field
  theory on curved space-time}.
\newblock {\em Communications in Mathematical Physics}, 179:529--553, 1996.

\bibitem{RadzikowskiVerch96}
Marek~J. Radzikowski and Rainer Verch.
\newblock A local-to-global singularity theorem for quantum field theory on
  curved space-time.
\newblock {\em Communications in Mathematical Physics}, 180:1--22, 1996.

\bibitem{ReedSimon-II}
M.~Reed and B.~Simon.
\newblock {\em Methods of Modern Mathematical Physics. Fourier Analysis,
  Self-Adjointness}.
\newblock Elsevier, New York, 1975.

\bibitem{ReedSimon-I}
M.C. Reed and B.~Simon.
\newblock {\em Methods of Modern Mathematical Physics: Functional analysis. 1}.
\newblock Methods of Modern Mathematical Physics. Academic Press, cop., 1980.

\bibitem{SahlmannVerch00}
Hanno Sahlmann and Rainer Verch.
\newblock {Passivity and Microlocal Spectrum Condition}.
\newblock {\em Communications in Mathematical Physics}, 214:705--731, 2000.

\bibitem{SahlmannVerch01}
Hanno Sahlmann and Rainer Verch.
\newblock {Microlocal Spectrum Condition and Hadamard Form for Vector-Valued
  Quantum Fields in Curved Spacetime}.
\newblock {\em Reviews in Mathematical Physics}, 13:1203--1246, 2001.

\bibitem{Sanders10}
Ko~Sanders.
\newblock {Equivalence of the (Generalised) Hadamard and Microlocal Spectrum
  Condition for (Generalised) Free Fields in Curved Spacetime}.
\newblock {\em Communications in Mathematical Physics}, 295:485--501, 2010.

\bibitem{Sanders13}
Ko~Sanders.
\newblock {On the construction of Hartle-Hawking-Israel states across a static
  bifurcate Killing horizon}.
\newblock 2013, 1310.5537.

\bibitem{Sanders13b}
Ko~Sanders.
\newblock {Thermal Equilibrium States of a Linear Scalar Quantum Field in
  Stationary Space-Times}.
\newblock {\em International Journal of Modern Physics A}, 28:1330010, 2013.

\bibitem{Shubin01}
M.~A. Shubin.
\newblock {\em Pseudodifferential operators and spectral theory}.
\newblock Springer-Verlag, Berlin, 2001.

\bibitem{Simon72}
Barry Simon.
\newblock Topics in functional analysis.
\newblock In R.~F. Streater, editor, {\em Mathematics of Contemporary Physics},
  London, 1972. Academic Press Inc.

\bibitem{Simon05}
Barry Simon.
\newblock {\em Trace ideals and their applications}, volume 120 of {\em
  Mathematical Surveys and Monographs}.
\newblock American Mathematical Society, Providence, RI, second edition, 2005.

\bibitem{Solveen12}
Christoph Solveen.
\newblock {Local thermal equilibrium and KMS states in curved spacetime}.
\newblock {\em Classical and Quantum Gravity}, 29:245015, 2012.

\bibitem{Sorkin11}
Rafael~D Sorkin.
\newblock Scalar field theory on a causal set in histories form.
\newblock {\em Journal of Physics: Conference Series}, 306:012017, 2011.

\bibitem{StreaterWightman64}
R.~F. Streater and A.~S. Wightman.
\newblock {\em P{CT}, spin and statistics, and all that}.
\newblock Princeton Landmarks in Physics. Princeton University Press,
  Princeton, NJ, 2000.
\newblock Corrected third printing of the 1978 edition.

\bibitem{Tanimoto04}
Masayuki Tanimoto.
\newblock Harmonic analysis of linear fields on the nilgeometric cosmological
  model.
\newblock {\em Journal of Mathematical Physics}, 45:4896--4919, 2004.

\bibitem{Them10}
Kolja Them.
\newblock Zust\"ande niedriger {E}nergie auf {B}ianchi {R}aumzeiten.
\newblock Diploma thesis, University of Hamburg, 2010.
\newblock http://www.desy.de/uni-th/lqp/psfiles/dipl-them.pdf.

\bibitem{ThemBrum13}
Kolja Them and Marcos Brum.
\newblock States of low energy in homogeneous and inhomogeneous expanding
  spacetimes.
\newblock {\em Classical and Quantum Gravity}, 30:235035, 2013.

\bibitem{Verch94}
Rainer Verch.
\newblock {Local definiteness, primarity and quasiequivalence of quasifree
  Hadamard quantum states in curved spacetime}.
\newblock {\em Communications in Mathematical Physics}, 160:507--536, 1994.

\bibitem{Verch99}
Rainer Verch.
\newblock Wavefront sets in algebraic quantum field theory.
\newblock {\em Communications in Mathematical Physics}, 205:337--367, 1999.

\bibitem{Visser95}
M.~Visser.
\newblock {\em Lorentzian wormholes: from Einstein to Hawking}.
\newblock Aip Series in Computational and Applied Mathematical Physics. AIP
  Press, American Institute of Physics, New York, 1995.

\bibitem{Wald84}
R.~M. Wald.
\newblock {\em General Relativity}.
\newblock Physics/Astrophysics. University of Chicago Press, Chicago, 1984.

\bibitem{Wald94}
R.~M. Wald.
\newblock {\em Quantum Field Theory in Curved Spacetime and Black Hole
  Thermodynamics}.
\newblock Chicago Lectures in Physics. University of Chicago Press, Chicago,
  1994.

\bibitem{Wald77}
Robert~M. Wald.
\newblock The back reaction effect in particle creation in curved spacetime.
\newblock {\em Communications in Mathematical Physics}, 54:1--19, 1977.

\bibitem{Yosida95}
K{\=o}saku Yosida.
\newblock {\em Functional analysis}.
\newblock Classics in Mathematics. Springer-Verlag, Berlin, 1995.
\newblock Reprint of the sixth (1980) edition.

\bibitem{YuWu04}
Hongwei Yu and Puxun Wu.
\newblock {Quantum inequalities for the free Rarita-Schwinger fields in flat
  spacetime}.
\newblock {\em Phys. Rev. D}, 69:064008, 2004.

\end{thebibliography}










\end{document}